\documentclass[12pt]{report}

\usepackage{amsmath, amsbsy, amssymb, mathrsfs, setspace, epsfig, bbm, 
			epsf, psfrag, booktabs, placeins, etoolbox, alltt, url, 
			afterpage, hyperref}
\usepackage[list=on]{subfig}
\usepackage[table,x11names]{xcolor}

\allowdisplaybreaks 

\usepackage[longnamesfirst]{natbib}
\usepackage{amscd} 
\usepackage{amsthm}  
\usepackage{graphicx} 

\graphicspath{{./figs/}}
\DeclareGraphicsExtensions{.eps,.ps,.pdf}

\usepackage{fancyvrb}
\usepackage{framed}

\usepackage{rotating}
\usepackage{uconnthesis} 

\numberwithin{figure}{chapter}
\numberwithin{table}{chapter}

\newcommand{\bm}[1]{{\boldsymbol {#1}}}
\newcommand{\bZ}{{\bm Z}}

\newcommand{\dif}{\mathrm{d}}
\newcommand\ceil[1]{\lceil#1\rceil}

\newtheorem{thm}{Theorem}

\DefineVerbatimEnvironment{Highlighting}{Verbatim}{commandchars=\\\{\},fontsize=\footnotesize}
\definecolor{shadecolor}{RGB}{248,248,248}
\newenvironment{Shaded}{\begin{snugshade}}{\end{snugshade}}
\newcommand{\KeywordTok}[1]{\textcolor[rgb]{0.13,0.29,0.53}{\textbf{{#1}}}}
\newcommand{\DataTypeTok}[1]{\textcolor[rgb]{0.13,0.29,0.53}{{#1}}}
\newcommand{\DecValTok}[1]{\textcolor[rgb]{0.00,0.00,0.81}{{#1}}}

\newcommand{\StringTok}[1]{\textcolor[rgb]{0.31,0.60,0.02}{{#1}}}

\newcommand{\CommentTok}[1]{\textcolor[rgb]{0.56,0.35,0.01}{\textit{{#1}}}}

\newcommand{\OtherTok}[1]{\textcolor[rgb]{0.56,0.35,0.01}{{#1}}}

\newcommand{\NormalTok}[1]{{#1}}
\makeatletter
\patchcmd{\@verbatim}
  {\verbatim@font}
  {\verbatim@font\footnotesize}
  {}{}
\makeatother

\begin{document}
\title{Automated, Efficient, and Practical Extreme Value Analysis with Environmental Applications}

\name{Brian M. Bader}

\myPreviousDegreeLong{
    B.A., Mathematics, Stony Brook University, NY, USA, 2009 \\
    M.A., Statistics, Columbia University, NY, USA, 2011}
\myPreviousDegreeShort{B.A. Mathematics, M.A. Statistics}

\degree{Doctor of Philosophy}
\myNewDegreeShort{Ph.D.}
\dept{Statistics}
\thesistype{Dissertation}

\myMajorAdvisor{Jun Yan}
\myAssociateAdvisorA{Kun Chen}
\myAssociateAdvisorB{Dipak K. Dey}
\myAssociateAdvisorC{Xuebin Zhang}


\firstpage
Although the fundamental probabilistic theory of extremes has 
been well developed, there are many practical considerations 
that must be addressed in application. The contribution of 
this thesis is four-fold. The first concerns the choice of 
$r$ in the $r$ largest order statistics modeling of extremes. 
Practical concern lies in choosing the value of $r$; a larger 
value necessarily reduces variance of the estimates, 
however there is a trade-off in that it may also introduce 
bias. Current model diagnostics are somewhat restrictive, 
either involving prior knowledge about the domain of the 
distribution or using visual tools. We propose a pair of 
formal goodness-of-fit tests, which can be carried out 
in a sequential manner to select $r$. A recently developed 
adjustment for multiplicity in the ordered, sequential setting 
is applied to provide error control. It is shown via simulation 
that both tests hold their size and have adequate power to detect 
deviations from the null model.

The second contribution pertains to threshold selection in the 
peaks-over-threshold approach. Existing methods for threshold 
selection in practice are informal as in visual diagnostics 
or rules of thumb, computationally expensive, or do not account 
for the multiple testing issue. We take a methodological approach, 
modifying existing goodness-of-fit tests combined with appropriate 
error control for multiplicity to provide an efficient, automated 
procedure for threshold selection in large scale problems.

The third combines a theoretical and methodological approach to 
improve estimation within non-stationary regional frequency models 
of extremal data. Two alternative methods of estimation to maximum 
likelihood (ML), maximum product spacing (MPS) and a hybrid  
L-moment / likelihood approach are incorporated in this framework. 
In addition to having desirable theoretical properties compared to 
ML, it is shown through simulation that these alternative estimators 
are more efficient in short record lengths.

The methodology developed is demonstrated with climate based 
applications. Last, an overview of computational issues for extremes 
is provided, along with a brief tutorial of the R package \texttt{eva}, 
which improves the functionality of existing extreme value software, 
as well as contributing new implementations.

\beforepreface


\prefacesection{Acknowledgements}
``On this life that we call home, the years go fast 
and the days go so slow.'' --- Modest Mouse

This is sound advice for anyone considering to pursue a 
PhD. It's hard to believe this chapter of my life is coming 
to an end --- the past four years have gone so fast, yet 
at times I thought it would never come soon enough. 
It has been full of ups and downs, but at the 
lowest of times, the only thing to do was continue on. 
I've made many new friends (who will hopefully turn 
into old), mentors, and gained precious knowledge that I 
think will benefit me throughout the rest of my life.

I'd like to thank Dr. Ming-Hui Chen for guiding me through 
the qualifying exam process and allowing me to thrive in 
the Statistical Consulting Service. I've gained valuable 
experience from participating in this group. I appreciate 
Dr. Dipak Dey for taking time out of his busy schedule as 
a dean to give me advice, recommendations, and to be 
on this committee. The same appreciation goes to Dr. Kun 
Chen for agreeing to be on my committee. Suggestions by 
Dr. Vartan Choulakian and Dr. Zhiyi Chi helped improve 
some of the methodology and data analysis in this research.

Dr. Jun Yan, my major advisor, has guided me throughout the 
research process and pushed me along to make sure I completed 
all the necessary milestones in a timely manner. I must admit 
that I was slightly intimidated of him at first, but I now 
believe that he is most likely the best choice of advisor 
(for me) and I am glad things fell into place as such. He is 
truly a kind person and has always been understanding of any 
problems I have had over the past two years. Although I am 
not pursuing the academic route at the moment, I appreciate 
his enthusiastic nudge for me to go in that direction.

Additionally, Dr. Xuebin Zhang has graciously spent his 
time and energy into conversations with myself and Dr. 
Yan to improve our manuscripts and our knowledge of 
environmental extremes. Of course I am thankful of the 
support he and Environment and Climate Change Canada gave 
by funding some of this research.

The journey would not have been the same with a different 
cohort --- I am grateful to all their support and friendship 
during such trying times. I have to acknowledge my family 
for encouraging me to follow my academic pursuits even if 
it meant not seeing me as often as they'd like during these 
four years. The same goes for all my friends back home.

Last, but not least, I wouldn't have made it through without 
the full support of my wife Deirdre and two cats Eva and 
Cuddlemonkey (who joined our family as a result of all 
this). I cannot express my total love and gratitude for 
them in words.

This research was partially supported by an NSF grant 
(DMS 1521730), a University of Connecticut Research Excellence 
Program grant, and Environment and Climate Change Canada.

\figurespagetrue
\tablespagetrue



\afterpreface
\chapter{Introduction}
\label{ch1}

\section{Overview of Extreme Value Analysis}
\label{ch1:intro}

Both statistical modeling and theoretical results of extremes 
remain a subject of active research. Extreme value theory 
provides a solid statistical framework to handle atypical, or 
heavy-tailed phenomena. There are many important applications 
that require modeling of extreme events. In hydrology, a government 
or developer may want an estimate of the maximum flood event that is 
expected to occur every 100 years, say, in order to determine the 
needs of a structure. In climatology, extreme value theory is 
used to determine if the magnitude of extremal 
events are time-dependent or not. Similarly, in finance, market 
risk assessment can be approached from an extremes standpoint. 
See~\cite{coles2001introduction, tsay2005analysis, yan2016extremes} 
for more specific examples.

Further, the study of extremes in a spatial context has 
been an area of interest for many researchers. In an environmental 
setting, one may want to know if certain geographic and/or 
climate features have an effect on extremes of precipitation, 
temperature, wave height, etc. Recently, many explicit models for 
spatial extremes have been developed. For an overview, 
see~\cite{davison2012statistical}. Another approach in the same 
context, regional frequency anaysis (RFA), allows one to `trade space 
for time' in order to improve the efficiency of certain parameter 
estimates. Roughly speaking, after estimating site-specific parameters, 
data are transformed onto the same scale and pooled in order to 
estimate the shared parameters. This approach offers two particular 
advantages over fully-specified multivariate models. First, only 
the marginal distributions at each site need to be explicitly 
chosen -- the dependence between sites can be handled by appropriate 
semi-parametric procedures and second, it can handle very short 
record lengths. See~\cite{hosking2005regional} 
or~\cite{wang2014incorporating} for a thorough review.

A major quantity of interest in extremes is the $t$-period return 
level. This can be thought of as the maximum event that will 
occur on average every $t$ periods and can be used in various 
applications such as value at risk in finance and flood zone 
predictions. Thus, it is quite important to obtain accurate 
estimates of this quantity and in some cases, determine if it 
is non-stationary. Given some specified extremal distribution, 
the stationary $t$-period return level $z_t$ can be expressed 
in terms of its upper quantile
\begin{equation*}
z_t = Q(1 - 1/t)
\end{equation*}
where $Q(p)$ is the quantile function of this distribution.

Within the extreme value framework, there are several different 
approaches to modeling extremes. In the following, the various 
approaches will be discussed and a data example of extreme daily 
precipitation events in California is used to motivate the content 
of this thesis.

\subsection{Block Maxima / GEV$_r$ Distribution}
\label{ch1:bm_rlargest}

The block maxima approach to extremes involves splitting the data 
into mutually exclusive blocks and selecting the top order statistic 
from within each block. Typically blocks can be chosen naturally; for 
example, for daily precipitation data over $n$ years, a possible block 
size $B$ could be $B$ = 365, with the block maxima referring to the 
largest annual daily precipitation event. To clarify ideas, here 
the underlying data is of size $365 \times n$ and the sample of extremes 
is size $n$, the number of available blocks. There is a requirement 
that the block size be `large enough' to ensure adequate convergence 
in the limiting distribution of the block maxima; further discussion 
of this topic will be delegated to later sections.

It has been shown~\citep[e.g.][]{leadbetter2012extremes, de2007extreme, 
coles2001introduction} that the only non-degenerate limiting distribution 
of the block maxima of a sample of size $B$ i.i.d. random variables, when 
appropriately normalized and as $B \to \infty$, must be the 
Generalized Extreme Value (GEV) distribution. The GEV distribution has 
cumulative distribution function given by 
\begin{equation}
\label{eq:gev_cdf}
F(y | \mu, \sigma, \xi) = 
\begin{cases} 
\exp\Big[-\Big(1 + \xi \frac{y - \mu}{\sigma} \Big)^{-\frac{1}{\xi}}\Big], & \xi \neq 0, \\
\exp\Big[- \exp\Big( - \frac{y - \mu}{\sigma} \Big) \Big], & \xi = 0,
\end{cases} 
\end{equation}
with location parameter $\mu$, scale parameter $\sigma > 0$, 
shape parameter $\xi$, and $ 1 + \xi (y - \mu) / \sigma > 0 $. By taking 
the first derivative with respect to $y$ the probability density function 
is obtained as
\begin{equation}
\label{eq:gev_pdf}
f(y | \mu, \sigma, \xi) = 
\begin{cases}
\frac{1}{\sigma} \Big(1+\xi \frac{y - \mu}{\sigma} \Big)^{-(\frac{1}{\xi}+1)} \exp\Big[-\Big(1 + \xi \frac{y - \mu}{\sigma} \Big)^{-\frac{1}{\xi}}\Big], & \xi \neq 0, \\
\frac{1}{\sigma} \exp\Big( - \frac{y - \mu}{\sigma} \Big) \exp\Big[- \exp\Big( - \frac{y - \mu}{\sigma} \Big) \Big], & \xi = 0.
\end{cases}
\end{equation}
Denote this distribution as GEV($\mu, \sigma, \xi$). The shape 
parameter $\xi$ controls the tail of the distribution. When 
$\xi>0$, the GEV distribution has a heavy, unbounded upper tail. 
When $\xi = 0$, this is commonly referred to as the Gumbel 
distribution and has a lighter tail. Figure~\ref{fig:gev_density} 
shows the GEV density for various shape parameter values.

\begin{figure*}[tbp]
\center
    \includegraphics[scale=0.5]{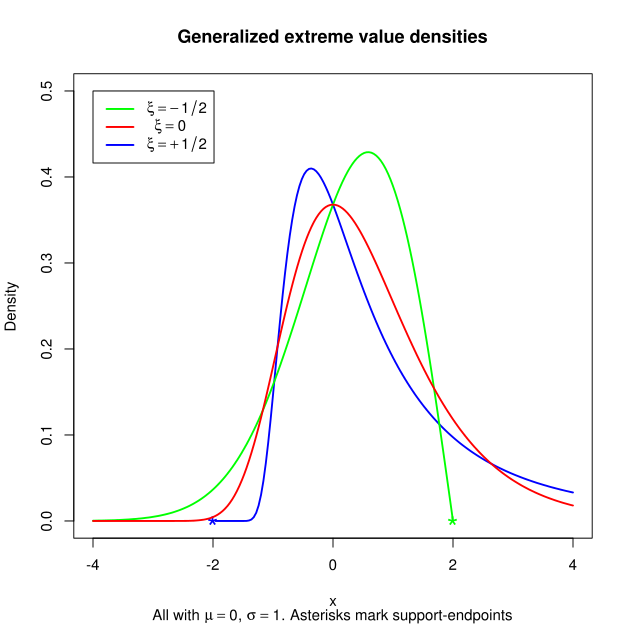}
    \caption{The density function of the Generalized Extreme Value distribution 
    for shape parameter values of $-0.5$, $0$, and $0.5$ with location 
    and scale parameters fixed at zero and one, respectively.}
    \label{fig:gev_density}
\end{figure*}

\cite{weissman1978estimation} generalized this result further, 
showing that the limiting joint distribution of the $r$ largest order 
statistics of a random sample of size $B$ as $B \to \infty$ (denoted 
here as the GEV$_r$ distribution) has probability density function
\begin{equation}
\label{eq:gevr_pdf}
f_r (y_1, y_2, ..., y_r | \mu, \sigma, \xi)
= \sigma^{-r} \exp\Big\{-(1+\xi z_r)^{-\frac{1}{\xi}}  - \left(\frac{1}{\xi}+1\right)\sum_{j=1}^{r}\log(1+\xi z_j)\Big\}
\end{equation}
for location parameter $\mu$, scale parameter $\sigma > 0$
and shape parameter $\xi$, 
where $y_1 >  \cdots> y_r$, $z_j = (y_j - \mu) / \sigma$, 
and $ 1 + \xi z_j > 0 $ for $j=1, \ldots, r$. The joint 
distribution for $\xi = 0$ can be found by taking the limit 
$\xi \to 0$ in conjunction with the Dominated Convergence 
Theorem and the shape parameter controls the tails of this 
distribution as discussed in the univariate GEV case. 
When $r = 1$,  this distribution is exactly the GEV distribution. 
The parameters $\theta = (\mu, \sigma, \xi)^{\top}$ remain the 
same for $j = 1, \ldots, r$, $r \ll B$, but the convergence rate
to the limit distribution reduces sharply as $r$ increases.
The conditional distribution of the $r$th component given the top 
$r - 1$ variables in~\eqref{eq:gevr_pdf} is the GEV distribution right
truncated by $y_{r-1}$, which facilitates simulation from the GEV$_r$
distribution; see Appendix~\ref{app:gevrsim}.

\subsection{Peaks Over Threshold (POT) Approach}
\label{ch1:pot}

Another approach to modeling extremes is the peaks over threshold 
method (POT). Instead of breaking up the underlying data into blocks 
and extracting the top observations from each block, POT sets some 
high threshold and uses only the observations above the threshold. 
Thus, POT is only concerned with the relevant observations, regardless 
of temporal ordering. Figure \ref{fig:pot_vs_bm} displays the 
differences between the POT and block maxima approaches.

\begin{figure*}[tbp]
\center
    \includegraphics[scale=0.5]{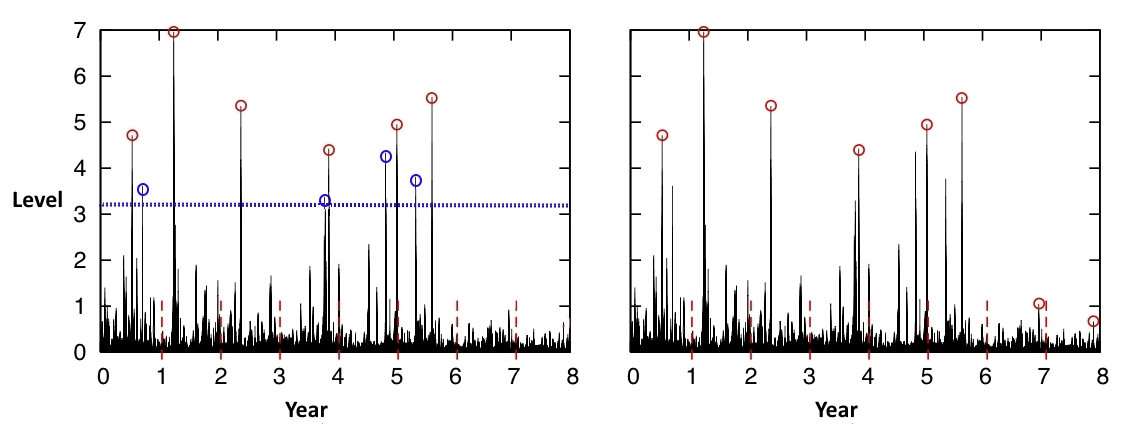}
    \caption{A comparison of extremes selected via the peaks over 
    threshold (left) versus block maxima approach for example 
    time series data.}
    \label{fig:pot_vs_bm}
\end{figure*}

Extreme value theory~\citep{mcneil1997peaks} says that 
given a suitably high threshold, data above the threshold will 
follow the Generalized Pareto (GPD) distribution. Under general 
regularity conditions, the only possible non-degenerate limiting 
distribution of properly rescaled exceedances of a threshold 
$u$ is the GPD as $u \to\infty$~\citep[e.g.,][]{pickands1975statistical}. 
The GPD has cumulative distribution function
\begin{equation}
\label{eq:gpd_cdf}
F(y | \theta) = \begin{cases}
1 - \Big[1 + \frac{\xi y}{\sigma_u}\Big]^{-1/\xi} & \xi \neq 0, \quad
y > 0, \quad 1 + \frac{\xi y}{\sigma_u} > 0, \\
1 - \exp{\Big[-\frac{y}{\sigma_u}\Big]} & \xi = 0, \quad y > 0,
\end{cases}
\end{equation}
where $\theta = (\sigma_u, \xi)$, $\xi$ is a shape parameter, and
$\sigma_u > 0$ is a threshold-dependent scale parameter.
The GPD also has the property that for some threshold $v > u$, the
excesses follow a GPD with the same shape parameter, but a modified
scale $\sigma_v = \sigma_u + \xi(v - u)$. Let $X_1, \ldots, X_n$ be a 
random sample of size $n$. If $u$ is sufficiently high, the exceedances 
$Y_i = X_i - u$ for all $i$ such that $X_i > u$ are approximately a 
random sample from a GPD.

The GPD has the probability density function given by
\begin{equation}
\label{eq:gpd_pdf}
f(y | \sigma_u, \xi) =
\begin{cases}
\frac{1}{\sigma_u} \Big[1 + \frac{\xi y}{\sigma_u}\Big]^{-(\frac{1}{\xi} + 1)}, & \xi \neq 0, \\
\frac{1}{\sigma_u} \exp{\Big[-\frac{y}{\sigma_u}\Big]}, & \xi = 0.
\end{cases}
\end{equation}
defined on $y \geq 0$ when $\xi \geq 0$ and 
$0 \leq y \leq -\sigma_u / \xi$ when $\xi <0$.

Both the block or threshold approach are justified in theory, but 
choice in application depends on the context or availability of 
data. For instance, it may be the case that only the block 
maxima or top order statistics from each period are available. 
The block maxima / $r$ largest approach provides 
a natural framework in which to retain temporal structure, while the 
POT requires additional care. This may be important if interest is 
in modeling non-stationary extremes. \cite{ferreira2015block} provide 
an overview of practical considerations when choosing between the two 
methods.

\subsection{Non-stationary Regional Frequency Analysis (RFA)}
\label{ch1:rfa}

Regional frequency analysis (RFA) is commonly used when historical 
records are short, but observations are available at multiple sites 
within a homogeneous region. A common difficulty with extremes is that, 
by definition, data is uncommon and thus short record length may make 
the estimation of parameters questionable. Regional frequency analysis 
resolves this problem by `trading space for time'; data from several 
sites are used in estimating event frequencies at any one 
site~\citep{hosking2005regional}. Essentially, certain parameters 
are assumed to be shared across sites, which increases the efficiency 
in estimation of those parameters.

In RFA, only the marginal distribution at each location needs to be 
specified. To set ideas, assume a region consists of $m$ sites 
over $n$ periods. Thus, observation $t$ at site $s$ can be denoted 
as $Y_{st}$. A common assumption is that data within each site are 
independent between periods; however, within each period $t$, it is 
typically the case that there is correlation between sites. For 
example, sites within close geographic distance cannot have events 
assumed to be independent. Fully specified multivariate models generally 
require this dependence structure to be explicitly defined and it 
is clear that the dependence cannot be ignored. As noted by 
authors~\cite{stedinger1983estimating} and ~\cite{hosking1988effect}, 
intersite dependence can have a dramatic effect on the variance of 
these estimators.

There are a number of techniques available to adjust the estimator 
variances accordingly without directly specifying the dependence 
structure. Some examples are combined score equations~\citep{wang2015thesis}, 
pairwise likelihood~\citep{wang2014incorporating, shang2015two}, 
semi-parametric bootstrap~\citep{heffernan2004conditional}, and
composite likelihood~\citep{chandler2007inference}.

\section{Motivation}
\label{ch1:motivation}

This thesis focuses on developing sound statistical methodology and 
theory to address practical concerns in extreme value applications. 
A common theme across the various methods developed here is 
automation, efficiency, and utility. There is a wide literature 
available of theoretical results and although recently the statistical 
modeling of extremes in application has gained in popularity, there 
are still methodological improvements that can be made. One of the 
major complications when modeling extremes in practice is deciding 
``what is extreme?''. The block maxima approach simplifies this 
idea somewhat, but for the POT approach, the threshold must be 
chosen. Similarly, if one wants to use the $r$ largest order 
statistics from each block (to improve efficiency of the estimates), 
how is $r$ chosen? Again, most of the current approaches do not 
address all three aspects mentioned earlier - automation, efficiency, 
and utility. For example, visual diagnostics cannot be easily automated, 
while certain resampling approaches are not scalable (efficiency), 
and many of the existing theoretical results may require some prior 
knowledge about the domain of attraction of the limiting distribution 
and/or require computational methods in finite samples.

As data continues to grow, there is a need for automation and 
efficiency / scalability. Climate summaries are currently 
available for tens of thousands of surface sites around the 
world via the Global Historical Climatology Network 
(GHCN)~\citep{menne2012overview}, 
ranging in length from 175 years to just hours. 
Other sources of large scale climate information are the 
National Oceanic and Atmospheric Administration (NOAA), 
United States Geological Survey (USGS), National Centers for 
Environmental Information (NCEI), Environment and Climate 
Change Canada, etc.

One particular motivating example is creating a return level 
map of extreme precipitation for sites across the western 
United States. Climate researchers may want to know how 
return levels of extreme precipitation vary across a geographical 
region and if these levels are affected by some external force 
such as the El Ni\~{n}o--Southern Oscillation (ENSO). In California 
alone, there are over 2,500 stations with some daily precipitation 
records available. Either using a jointly estimated model such 
as regional frequency analysis or analyzing data site-wise, 
appropriate methods are needed to accommodate modeling a 
large number of sites.

\subsection{Choice of $r$ in the $r$ 
Largest Order Statistics Model}
\label{ch1:motivation_rlargest}

While the theoretical framework of extreme value theory is sound, 
there are many practical problems that arise in applications. One such 
issue is the choice of $r$ in the GEV$_r$ distribution. Since $r$ is not 
explicitly a parameter in the distribution~\eqref{eq:gevr_pdf}, the usual 
model selection techniques (i.e. likelihood-ratio testing, AIC, BIC) 
are not available. A bias-variance trade-off exists when selecting $r$. 
As $r$ increases, the variance decreases because more data is used, but 
if $r$ is chosen too high such that the approximation of the data to the 
GEV$_r$ distribution no longer holds, bias may occur. It has been shown 
that the POT method is more efficient than block maxima in small 
samples~\citep{caires2009comparative} and thus it is often recommended 
to use that method over block maxima.

It appears that in application, 
the GEV$_r$ distribution is often not considered because of the issues 
surrounding the selection of $r$ and that simply using the block maxima 
or POT approach are more straightforward. To the author's knowledge, 
no comparison between efficiency of the GEV$_r$ distribution and POT 
method has been carried out in finite samples. In practice~\citep{an2007r, smith1986extreme} 
the recommendation for the choice of $r$ is sometimes based on the 
amount of reduction in standard errors of the estimates.

\citet{smith1986extreme}  and \citet{tawn1988extreme} used probability 
plots for the marginal distributions of the $r$th order statistic to 
assess goodness of fit. 
Note that this can only diagnose poor model fit at the marginal level -- 
it does not consider the full joint distribution. \citet{tawn1988extreme} 
suggested an alternative test of fit using a spacings result 
in \citet{weissman1978estimation}, however this requires prior knowledge 
about the domain of attraction of the limiting distribution. These issues 
become even more apparent when it is desired to fit the GEV$_r$ distribution 
to more than just one sample. This is carried out for 30 stations in the 
Province of Ontario on extreme wind speeds~\citep{an2007r} but the value 
of $r=5$ is fixed across all sites.

\subsection{Selection of Threshold in the POT Approach}
\label{ch1:motivation_gpd}

A similar, but distinct problem is threshold selection when modeling with 
the Generalized Pareto distribution. In practice, the threshold must be 
chosen, yet it cannot be chosen using traditional model selection tests 
since it is not a parameter in the distribution. This has been studied 
thoroughly in the literature. Various graphical procedures exist. 
The mean residual life (MRL) plot, introduced by 
\citet{davison1990models} uses the expectation of GPD excesses; 
for $v>u$, $E[X-v | X>v]$ is linear in $v$ when the GPD fits the 
data above $u$. The idea is to choose the smallest value of $u$ such 
that the plot is linear above this point. The Hill estimator 
\citep{hill1975simple} for the tail index $\xi$ is based on a sum of 
the log spacings of the top $k+1$ order statistics. \citet{drees2000make} 
discuss the Hill plot, which plots the Hill estimator against the 
top $k$ order statistics. The value of $k$ is chosen as the largest 
(i.e. lowest threshold) such that the Hill estimator has become stable. 
A similar figure, referred to as the threshold stability plot, compares 
the estimates of the GPD parameters at various thresholds and the idea 
is to choose a threshold such that the parameters at this threshold 
and higher are stable.

\begin{figure*}[tbp]
\center
    \includegraphics[scale=0.7]{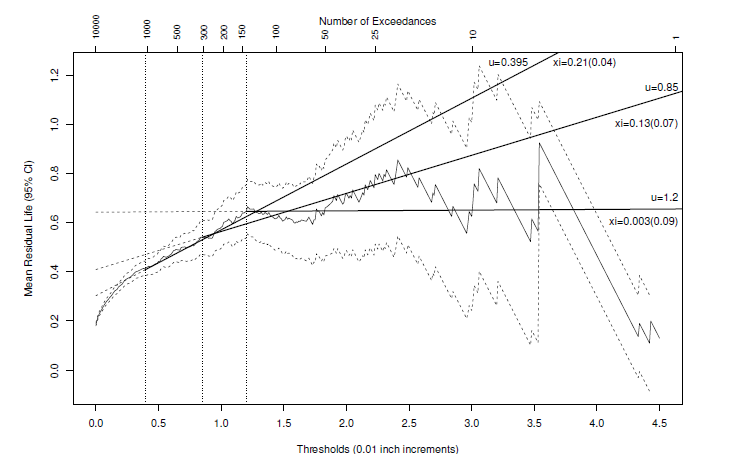}
    \caption{Mean Residual Life plot of Fort Collins daily precipitation 
    data found in R package \texttt{extRemes}.}
    \label{fig:gpd_mrl_plot}
\end{figure*}

It is clear that visual diagnostics cannot be scaled effectively. 
Even in the one sample case can be quite difficult to interpret, 
with the Hill plot being referred to as the `Hill Horror Plot'. 
Figures \ref{fig:gpd_mrl_plot}, \ref{fig:gpd_thresh_plot}, and 
\ref{fig:gpd_hill_plot} are examples of the mentioned plots applied 
to the Fort Collins daily precipitation data in R package 
\texttt{extRemes} \citep{extRemes}.

\begin{figure*}[tbp]
\center
    \includegraphics[scale=0.7]{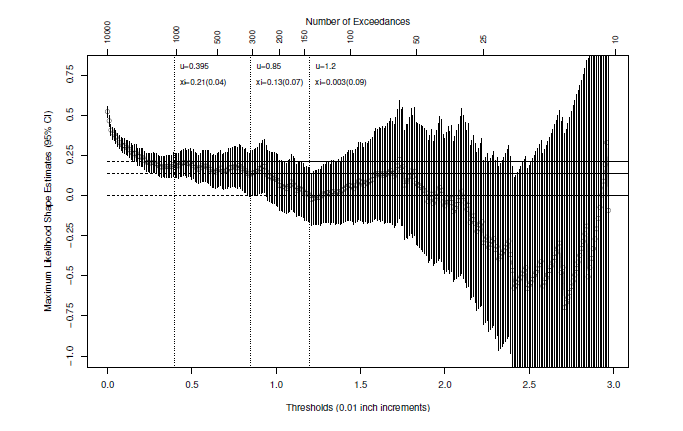}
    \caption{Threshold stability plot for the shape parameter of the 
    Fort Collins daily precipitation data found in R package \texttt{extRemes}.}
    \label{fig:gpd_thresh_plot}
\end{figure*}

\begin{figure*}[tbp]
\center
    \includegraphics[scale=0.7]{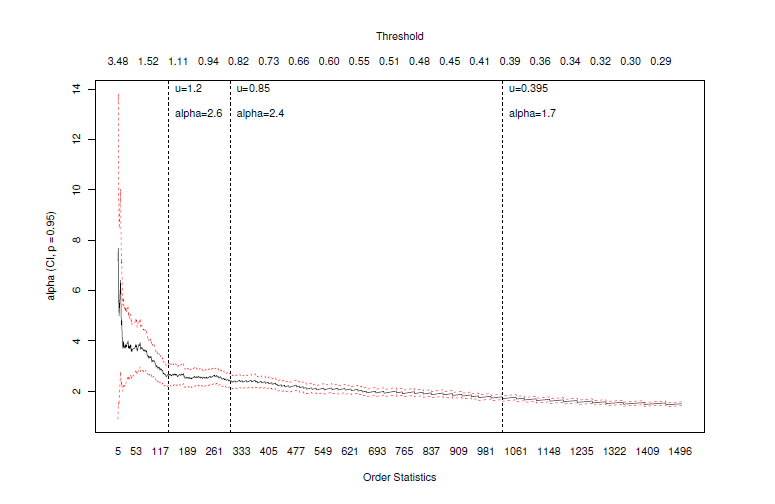}
    \caption{Hill plot of the Fort Collins daily precipitation data 
    found in R package \texttt{extRemes}.}
    \label{fig:gpd_hill_plot}
\end{figure*}

Some practitioners suggest various `rules of thumb', which involve 
selecting the threshold based on some predetermined fraction of the data 
or it can involve complicated resampling techniques. There is also the 
idea of using a mixture distribution, which involves specifying a 
`bulk' distribution for the data below the threshold and using the GPD 
to model data above the threshold. In this way, the threshold can be 
explicitly modeled as a parameter. There are some drawbacks to this 
approach however -- the `bulk' distribution must be specified, and 
care is needed to ensure that the two densities are continuous 
at the threshold pount $u_0$.

Goodness-of-fit testing can be used for threshold selection. A set of 
candidate thresholds $u_1 < \ldots < u_l$ can be tested sequentially 
for goodness-of-fit to the GPD. The goal is to select a smaller 
threshold in order to reduce variance of the estimates, but not too 
low as to introduce bias. Various authors have developed methodology 
to perform such testing, but they do not consider the multiple testing 
issue, or it can be computational intensive to perform. For a more 
thorough and detailed review of the approaches discussed in this section, 
see~\cite{scarrott2012review} and section~\ref{ch3:intr}.

\subsection{Estimation in Non-stationary RFA}
\label{ch1:motivation_rfa}

Unless otherwise noted, going forward the assumption is that the marginal 
distribution used in fitting a regional frequency model is the GEV or 
block maxima method. It is well known that due to the non-regular shape 
of the likelihood function, the MLE may not exist when the shape parameter 
of the GEV distribution, $\xi < -0.5$~\citep{smith1985maximum}. This can 
cause estimation and/or optimization issues, especially in situations 
where the record length is short. Even so, maximum likelihood is 
widely popular and relatively straightforward to implement. As an 
alternative, in the stationary RFA case, one can use L-moments~\citep{hosking1990moments} 
to estimate the parameters in RFA. L-moments has the advantage over MLE in that it 
only requires the existence of the mean, and has been shown to be more 
efficient in small samples~\citep{hosking1985estimation}. However, for 
non-stationary RFA, it is not straightforward to incorporate covariates 
using L-moments, and generally MLE is used -- see~\citep{katz2002statistics, 
lopez2013non, hanel2009nonstationary, leclerc2007non, nadarajah2005extremes} 
as examples. 

One approach to estimate time trends is by applying the stationary 
L-moment approach over sliding time 
windows~\citep{kharin2005estimating, kharin2013changes}; that is, estimate 
the stationary parameters in (mutually exclusive) periods and study the 
change in parameters. This is not a precise method, as it is hard to 
quantify whether change is significant or due to random variation. 
In the one sample case, there has been some progress to combine 
non-stationarity and L-moment estimation. \cite{ribereau2008estimating} 
provide a method to incorporate covariates in the location parameter, 
by estimating the covariates first via least squares, and then 
transforming the data to be stationary in order to estimate the 
remaining parameters via L-moments. \cite{coles1999likelihood} briefly 
discuss an iterative procedure to estimate covariates through maximum 
likelihood and stationary parameters through L-moments. However, these 
approaches only consider non-stationary in the location parameter and 
it may be of interest to perform linear modeling of the scale and 
shape parameters.

\section{Outline of Thesis}
\label{ch1:outline}

The rest of this thesis is as follows. Chapter~\ref{ch:r-largest} 
builds on the discussion in Section~\ref{ch1:motivation_rlargest}, 
developing two goodness-of-fit tests for selection of $r$ in the $r$ 
largest order statistics model. The first is a score test, which 
requires approximating the null distribution via a parametric or 
multiplier bootstrap approach. Second, named the entropy difference 
test, uses the expected difference between log-likelihood of 
the distributions of the $r$ and $r-1$ top order statistics to 
produce an asymptotic test based on normality. The tests are studied 
for their power and size, and newly developed error control methods 
for order, sequential testing is applied. The utility of the tests 
are shown via applications to extreme sea level and precipitation 
datasets.

Chapter~\ref{ch:gpd} tackles the problem of threshold selection in 
the peaks over threshold model, discussed in 
Section~\ref{ch1:motivation_gpd}. A goodness-of-fit testing approach 
is used, with an emphasis on automation and efficiency. Existing 
tests are studied and it is found that the Anderson--Darling 
has the most power in various scenarios testing a single, fixed 
threshold. The same error control method discussed in 
Section~\ref{ch2:seq} can be adapted here to control for multiplicity 
in testing ordered, sequential thresholds for goodness-of-fit. Although 
the asymptotic null distribution of the Anderson--Darling testing 
statistic for the GPD has been derived~\citep{choulakian2001goodness}, 
it requires solving an integral equation. We develop a method to obtain 
approximate p-values in a computationally efficient manner. The test, 
combined with error control for the false discovery rate, is shown via 
a large scale simulation study to outperform familywise and no error 
controls. The methodology is applied to obtain a return level map of 
extreme precipitation of at hundreds of sites in the western United 
States.

When analyzing climate extremes at many sites, it may be desired 
to combine information across sites to increase efficiency of the 
estimates, for example, using a regional frequency model. In addition, 
one may want to incorporate non-stationary. Currently, the only 
estimation methods available in this framework may have drawbacks 
in certain cases, as discussed in 
Section~\ref{ch1:motivation_rfa}. In Chapter~\ref{ch:rfa}, we 
introduce two alternative methods of estimation in non-stationary 
RFA, that have advantageous theoretical properties when compared to 
current estimation methods such as MLE. It is shown via simulation 
of spatial extremes with extremal and non-extremal dependence 
that the two new estimation methods empirically outperform MLE. A 
non-stationary regional frequency flood-index model is fit to annual 
maximum daily winter precipitation events at 27 locations in 
California, with an interest in modeling the effect of the El 
Ni\~{n}o--Southern Oscillation Index on these events.

Chapter~\ref{ch:soft} provides a brief tutorial to the companion 
software package \texttt{eva}, which implements the majority of 
the methodology developed here. It provides new implementations of 
certain techniques, such as maximum product spacing estimation, and 
data generation and density estimation for the GEV$_r$ distribution. 
Additionally, it improves on existing implementations of extreme value 
analysis, particularly numerical handling of the near-zero shape 
parameter, profile likelihood, and user-friendly model fitting 
for univariate extremes. Lastly, a discussion of this body of work, 
and possible future direction follows in Chapter~\ref{ch:conclusion}.

\chapter{Automated Selection of $r$ in the $r$ 
Largest Order Statistics Model}
\label{ch:r-largest}

\section{Introduction}
\label{ch2:intr}

The largest order statistics approach is an extension of the block 
maxima approach that is often used in extreme value modeling.
The focus of this chapter is \citep[p.28--29]{smith1986extreme}:
``Suppose we are given, not just the maximum value for each year,
but the largest ten (say) values. How might we use this data to obtain
better estimates than could be made just with annual maxima?''
The $r$ largest order statistics approach may use more information 
than just the block maxima in extreme value analysis, and is widely
used in practice when such data are available for each block.
The approach is based on the limiting distribution of the $r$ 
largest order statistics which extends the generalized extreme 
value (GEV) distribution \citep[e.g.,][]{weissman1978estimation}. 
This distribution, given in~\eqref{eq:gevr_pdf} and denoted as GEV$_r$, 
has the same parameters as the GEV distribution, which makes it useful to 
estimate the GEV parameters when the $r$ largest values are 
available for each block. The approach was investigated 
by \citet{smith1986extreme} for the limiting joint Gumbel distribution 
and extended to the more general 
limiting joint GEV$_r$ distribution by \citet{tawn1988extreme}.
Because of the potential gain in efficiency relative to the block
maxima only, the method has found many applications such as
corrosion engineering \citep[e.g.,][]{scarf1996estimation}, 
hydrology \citep[e.g.,][]{dupuis1997extreme},
coastal engineering \citep[e.g.,][]{guedes2004application},
and wind engineering \citep[e.g.,][]{an2007r}.

In practice, the choice of $r$ is a critical issue in extreme 
value analysis with the $r$ largest order statistics approach. 
In general $r$ needs to be small relative to the block size
$B$ (not the number of blocks $n$) because as $r$ increases,
the rate of convergence to the limiting joint distribution 
decreases sharply \citep{smith1986extreme}. 
There is a trade-off between the validity of the limiting result
and the amount of information required for good estimation.
If $r$ is too large, bias can occur; if too small, 
the variance of the estimator can be high.
Finding the optimal $r$ should lead to more efficient
estimates of the GEV parameters without introducing bias. 
Our focus here is the selection of $r$ for situations where 
a number of largest values are available each of $n$ blocks.
In contrast, the methods for threshold or fraction selection 
reviewed in \citet{scarrott2012review} deal with 
a single block ($n = 1$) of a large size $B$.

The selection of $r$ has not been as actively researched as 
the threshold selection problem in the one sample case. 
\citet{smith1986extreme} and \citet{tawn1988extreme} used 
probability (also known as PP) plots for the marginal distribution
of the $r$th order statistic to assess its goodness of fit. 
The probability plot provides a visual diagnosis, but different viewers 
may reach different conclusions in the absence of a p-value.
Further, the probability plot is only checking the marginal 
distribution for a specific $r$ as opposed to the joint distribution.
\citet{tawn1988extreme} suggested an alternative test of fit
using a spacings results in \citet{weissman1978estimation}.
Let $D_i$ be the spacing between the $i$th and $(i+1)$th
largest value in a sample of size $B$ from a distribution 
in the domain of attraction of the Gumbel distribution. 
Then $\{iD_i: i = 1, \ldots, r-1\}$ is approximately a set
of independent and identically distributed exponential random 
variables as $B \to \infty$. 
The connections among the three limiting forms of the GEV 
distribution \citep[e.g.,][p.123]{embrechts1997modelling}
can be used to transform from the Fr\'echet and the Weibull
distribution to the Gumbel distribution.
Testing the exponentiality of the spacings on the Gumbel
scale provides an approximate diagnosis of the joint distribution
of the $r$ largest order statistics when $B$ is large.
A limitation of this method, however, is that prior knowledge
of the domain of attraction of the distribution is needed. 
Lastly, \citet{dupuis1997extreme} proposed a robust 
estimation method, where the weights can be used to 
detect inconsistencies with the GEV$_r$ distribution
and assess the fit of the data to the joint Gumbel model.
The method can be extended to general GEV$_r$ distributions
but the construction of the estimating equations is computing
intensive with Monte Carlo integrations.

In this chapter, two specification tests are proposed to select 
$r$ through a sequence of hypothesis testing. 
The first is the score test \citep[e.g.,][]{rao2005score}, but 
because of the nonstandard setting of the GEV$_r$ distribution, 
the usual $\chi^2$ asymptotic distribution is invalid. 
A parametric bootstrap can be used to assess the significance
of the observed statistic, but is computationally demanding. 
A fast, large sample alternative to parametric bootstrap based on 
the multiplier approach \citep{kojadinovic2012goodness} is developed. 
The second test uses the difference in estimated entropy between
the GEV$_r$ and GEV$_{r-1}$ models, applied to the $r$ largest order 
statistics and the $r-1$ largest order statistics, respectively. 
The asymptotic distribution is derived with the central limit theorem. 
Both tests are intuitive to understand, easy to implement, and have 
substantial power as shown in the simulation studies.
Each of the two tests is carried out to test the adequacy 
of the GEV$_r$ model for a sequence of $r$ values. 
The very recently developed stopping rules for ordered
hypotheses in \citet{g2015sequential} are adapted to control 
the false discovery rate (FDR), the expected proportion 
of incorrectly rejected null hypotheses among all rejections,
or familywise error rate (FWER), the probability of at least one 
type~I error in the whole family of tests. All the methods are 
available in the R package \texttt{eva} \citep{Rpkg:eva} and 
some demonstration is seen in Chapter~\ref{ch:soft}.

The rest of the chapter is organized as follows. 
The problem is set up in Section~\ref{ch2:setup} with the GEV$_r$ 
distribution, observed data, and the hypothesis to be tested.
The score test is proposed in Section~\ref{ch2:score} with two 
implementations: parametric bootstrap and multiplier bootstrap. 
The entropy difference (ED) test is proposed and the asymptotic
distribution of the testing statistic is derived in Section~\ref{ch2:ed}.
A large scale simulation study on the empirical size and power of 
the tests are reported in Section~\ref{ch2:sim}. 
In Section~\ref{ch2:seq}, the multiple, sequential testing problem
is addressed by adapting recent developments on this application.
The tests are applied to sea level and precipitation datasets 
in Section~\ref{ch2:app}. A discussion concludes in Section~\ref{ch2:disc}. 
The Appendices \ref{app:gevrsim} and \ref{app:tn} contain the 
details of random number generation from the GEV$_r$ distribution 
and a sketch of the proof of the asymptotic distribution of 
the ED test statistic, respectively.

\section{Model and Data Setup}
\label{ch2:setup}

The limit joint distribution~\citep{weissman1978estimation} 
of the $r$~largest order statistics of a random sample of 
size $B$ as $B \to \infty$ is the GEV$_r$ distribution with 
probability density function given in~\eqref{eq:gevr_pdf}.

The $r$~largest order statistics approach is an extension of the 
block maxima approach in extreme value analysis when a number
of largest order statistics are available for each one of a collection 
of independent blocks \citep{smith1986extreme,tawn1988extreme}.
Specifically, let $(y_{i1}, \ldots, y_{ir})$ be the observed $r$~largest 
order statistics from block $i$ for $i = 1, \ldots, n$.
Assuming independence across blocks, the GEV$_r$ distribution is 
used in place of the GEV distribution~\eqref{eq:gev_pdf} in the 
block maxima approach to make likelihood-based inference about $\theta$.
Let $l_i^{(r)} (\theta) = l^{(r)} (y_{i1}, \ldots, y_{ir} | \theta)$,
where
\begin{equation}
\label{eq:gevr_ll}
l^{(r)} (y_1, \ldots, y_r | \theta) = -r\log{\sigma} - (1+\xi z_r)^{-\frac{1}{\xi}}  - \left(\frac{1}{\xi}+1\right)\sum_{j=1}^{r}\log(1+\xi z_j)
\end{equation}
is the contribution to the log-likelihood from a single block
$(y_1, \ldots, y_r)$ with location parameter $\mu$, scale 
parameter $\sigma > 0$ and shape parameter $\xi$, 
where $y_1 >  \cdots> y_r$, $z_j = (y_j - \mu) / \sigma$, and 
$ 1 + \xi z_j > 0 $ for $j=1, \ldots, r$.
The maximum likelihood estimator (MLE) of $\theta$ using the
$r$ largest order statistics is 
$\hat\theta_n^{(r)} = \arg\max \sum_{i=1}^n l_i^{(r)} (\theta)$.

Model checking is a necessary part of statistical analysis.
The rationale of choosing a larger value of $r$ is to use 
as much information as possible, but not set $r$ too
high so that the GEV$_r$ approximation becomes poor
due to the decrease in convergence rate.
Therefore, it is critical to test the goodness-of-fit of 
the GEV$_r$ distribution with a sequence of null hypotheses
\begin{center}
$H_0^{(r)}$: the GEV$_r$ distribution fits the sample of the
$r$ largest order statistics well
\end{center}
for $r=1, \ldots, R$, where $R$ is the maximum, 
predetermined number of top order statistics to test. 
Two test procedures for $H_0^{(r)}$ are developed for 
a fixed $r$ first to help choose $r \geq 1$ such that 
the GEV$_r$ model still adequately describes the data.
The sequential testing process and the multiple testing
issue are investigated in Section~\ref{ch2:seq}.

\section{Score Test}
\label{ch2:score}

A score statistic for testing goodness-of-fit hypothesis
$H_0^{(r)}$ is constructed in the usual way with the score 
function and the Fisher information matrix
\citep[e.g.,][]{rao2005score}.
For ease of notation, the superscript $(r)$ is dropped.
Define the score function
\[
S(\theta) = \sum_{i=1}^n S_i(\theta) = 
\sum_{i=1}^n \partial l_i(\theta) / \partial \theta
\]
and Fisher information matrix $I(\theta)$, 
which have been derived in \citet{tawn1988extreme}. 
The behaviour of the maximum likelihood estimator is 
the same as that derived for the block maxima approach 
\citep{smith1985maximum,tawn1988extreme}, which
requires $\xi > -0.5$.
The score statistic is
\begin{displaymath}
V_n = \frac{1}{n} S^{\top}(\hat\theta_n) I^{-1}(\hat\theta_n) S(\hat\theta_n).
\end{displaymath}

Under standard regularity conditions, $V_n$ would asymptotically
follow a $\chi^2$ distribution with 3 degrees of freedom. 
The GEV$_r$ distribution, however, violates the regularity conditions 
for the score test \citep[e.g.,][pp. 516-517]{casella2002statistical},
as its support depends on the parameter values unless $\xi = 0$.
For illustration, Figure~\ref{fig:gevscore_AsymChiSq} presents a visual 
comparison of the empirical distribution of $V_n$ with $n = 5000$
from 5000 replicates, overlaid with the $\chi^2(3)$ distribution,
for $\xi \in \{-0.25, 0.25\}$ and $r \in \{1, 2, 5\}$.
The sampling distribution of $V_n$ appears to be much heavier tailed
than $\chi^2(3)$, and the mismatch increases as $r$ increases as a
result of the reduced convergence rate.

\begin{figure*}[tbp]
    \centering
    \includegraphics[width=\textwidth]{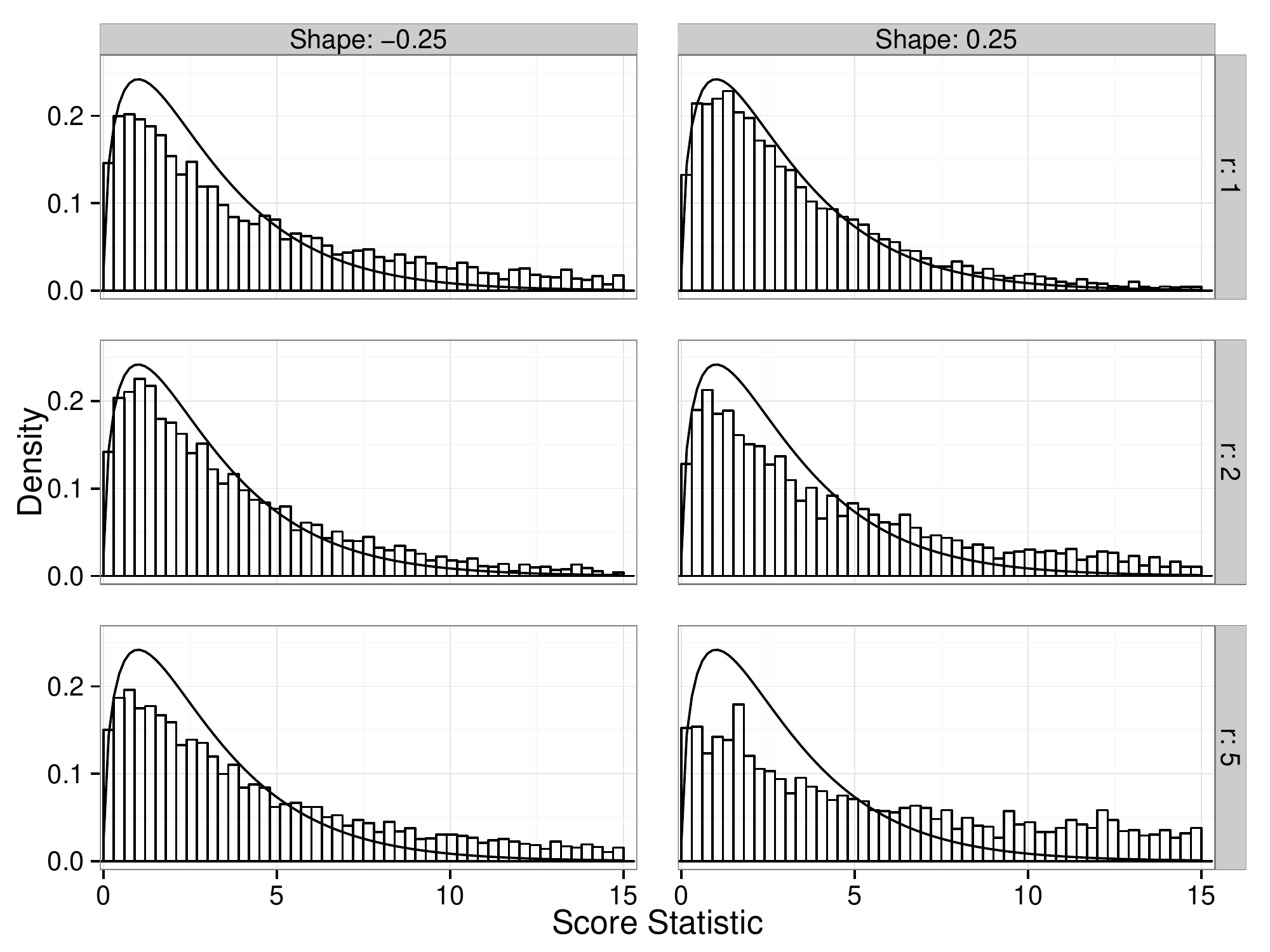}
    \caption{Comparisons of the empirical vs. $\chi^2(3)$ distribution (solid curve) 
    based on 5000 replicates of the score test statistic under the null GEV$_r$ distribution. 
    The number of blocks used is $n =5000$ with parameters $\mu=0$, $\sigma=1$, and 
    $\xi \in (-0.25, 0.25)$.
    }
    \label{fig:gevscore_AsymChiSq}
\end{figure*}

Although the regularity conditions do not hold, the score 
statistic still provides a measure of goodness-of-fit since 
it is a quadratic form of the score, which has expectation 
zero under the null hypothesis.
Extremely large values of $V_n$ relative to its sampling
distribution would suggest lack of fit, and, hence, 
possible misspecification of $H_0^{(r)}$.
So the key to applying the score test is to get an 
approximation of the sampling distribution of $V_n$.
Two approaches for the approximation are proposed.

\subsection{Parametric Bootstrap}
\label{ch2:pb}

The first solution is parametric bootstrap.
For hypothesis $H_0^{(r)}$, the test procedure goes as follows:
\begin{enumerate}
\item
Compute $\hat\theta_n$ under $H_0$ with the observed data.
\item
Compute the testing statistic $V_n$.
\item
For every $k \in  \{1, ..., L \}$ with a large number $L$, repeat:
  \begin{enumerate}
  \item
  Generate a bootstrap sample of size $n$ for the $r$ largest statistics 
  from GEV$_r$ with parameter vector $\hat\theta_n$.
  \item
  Compute the $\hat\theta_n^{(k)}$ under $H_0$ with the bootstrap sample.
  \item
  Compute the score test statistic $V_n^{(k)}$.
  \end{enumerate}
\item
Return an approximate p-value of $V_n$ as 
$L^{-1} \sum_{k=1}^{L} 1(V_n^{(k)} > V_n)$.
\end{enumerate}

Straightforward as it is, the parametric bootstrap approach
involves sampling from the null distribution and computing the MLE 
for each bootstrap sample, which can be very computationally expensive.
This is especially true as the sample size $n$ and/or the number of
order statistics $r$ included in the model increases.

\subsection{Multiplier Bootstrap}
\label{ch2:mb}

Multiplier bootstrap is a fast, large sample alternative to parametric 
bootstrap in goodness-of-fit testing \citep[e.g.,][]{kojadinovic2012goodness}.
The idea is to approximate the asymptotic distribution of 
$n^{-1/2} I^{-1/2}(\theta) S(\theta)$ using its asymptotic representation
\[
n^{-1/2} I^{-1/2}(\theta) S(\theta) = 
\frac{1}{\sqrt{n}}\sum_{i=1}^n \phi_i(\theta),
\]
where $\phi_i(\theta) =  I^{-1/2} (\theta) S_i(\theta)$.
Its asymptotic distribution is the same as the asymptotic distribution of
\[
W_n(\bZ, \theta) = \frac{1}{\sqrt{n}} \sum_{i=1}^n (Z_i  - \bar Z)\phi_i(\theta),
\]
conditioning on the observed data, 
where ${\bZ} = (Z_1, ..., Z_n)$ is a set of independent and identically
distributed multipliers (independent of the data), with expectation 0 and
variance 1, and $\bar{Z} = \frac{1}{n} \sum_{i=1}^{n} Z_i$. 
The multipliers must satisfy 
$\int_0^\infty \{\Pr(|Z_1| > x)\}^{\frac{1}{2}} \dif x < \infty $. 
An example of a possible multiplier distribution is $N(0, 1)$.

The multiplier bootstrap test procedure is summarized as follows:
\begin{enumerate}
\item
Compute $\hat\theta_n$ under $H_0$ with the observed data.
\item
Compute the testing statistic $V_n$.
\item
For every $k \in  \{1, ..., L \}$ with a large number $L$, repeat:
  \begin{enumerate}
  \item
    Generate $\bZ^{(k)} = (Z_1^{(k)}, \ldots, Z_n^{(k)})$ from $N(0, 1)$.
  \item
    Compute a realization from the approximate distribution of 
    $W_n(\bZ, \theta)$ with $W_n(\bZ^{(k)}, \hat\theta_n)$.
  \item 
    Compute $V_n^{(k)}(\hat\theta_n) = W_n^{\top}(\bZ^{(k)}, \hat\theta_n) W_n(\bZ^{(k)}, \hat\theta_n)$.
  \end{enumerate}
\item
Return an approximate p-value of $V_n$ as 
$L^{-1} \sum_{k=1}^{L} 1(V_n^{(k)} > V_n)$.
\end{enumerate}

This multiplier bootstrap procedure is much faster than parametric
bootstrap procedure because, for each sample, it only
needs to generate $\bZ$ and compute $W_n(\bZ, \hat\theta_n)$.
The MLE only needs to be obtained once from the observed data.

\section{Entropy Difference Test}
\label{ch2:ed}

Another specification test for the GEV$_r$ model is derived based 
on the  difference in entropy for the GEV$_r$ and GEV$_{r-1}$ models. 
The entropy for a continuous random variable with density $f$ is
\citep[e.g.,][]{singh2013entropy}
\begin{displaymath}
E[-\ln  f(y)] = - \int_{-\infty}^{\infty}  f(y) \log f(y)   \dif y.
\end{displaymath}
It is essentially the expectation of negative log-likelihood.
The expectation can be approximated with the sample average of
the contribution to the log-likelihood from the observed data,
or simply the log-likelihood scaled by the sample size $n$.
Assuming that the $r - 1$ top order statistics fit the GEV$_{r-1}$
distribution, the difference in the log-likelihood between GEV$_{r-1}$ 
and GEV$_r$ provides a measure of deviation from $H_0^{(r)}$.
Its asymptotic distribution can be derived.
Large deviation from the expected difference under $H_0^{(r)}$
suggests a possible misspecification of $H_0^{(r)}$.

From the log-likelihood contribution in~\eqref{eq:gevr_ll}, 
the difference in log-likelihood for the $i$th block,
$D_{ir} (\theta) = l_i^{(r)}  - l_i^{(r-1)}$, is 
\begin{equation}
\label{eq:dll}
D_{ir}(\theta) = -\log{\sigma} - 
(1+\xi z_{ir})^{-\frac{1}{\xi}} + 
(1+\xi z_{i{r-1}})^{-\frac{1}{\xi}} - 
\bigl(\frac{1}{\xi}+1 \bigr)\log(1+\xi z_{ir}).
\end{equation}
Let $\bar{D_r} = \frac{1}{n} \sum_{i=1}^{n} D_{ir}$ and 
$S_{D_r}^2 = \sum_{i=1}^{n} (D_{ir} - \bar{D_r})^2 / (n - 1)$
be the sample mean and sample variance, respectively.
Consider a standardized version of $\bar D_r$ as
\begin{equation}
  \label{eq:ed}
  T^{(r)}_n (\theta) = \sqrt{n}(\bar D_r - \eta_r) / S_{D_r},
\end{equation}
where $\eta_r = -\log{\sigma} - 1 + (1+\xi)\psi(r)$, and
$\psi(x) = \dif \log \Gamma(x) / \dif x$ is the digamma function.
The asymptotic distribution of $T^{(r)}_n$ is summarized 
by Theorem~\ref{thm:ed} whose proof is relegated to 
Appendix~\ref{app:tn}. This is essentially looking at the 
difference between the Kullback--Leibler divergence and its 
sample estimate. It is also worth pointing out that $T^{(r)}_n$ 
is only a function of the $r$ and $r-1$ top order statistics (i.e., 
only the conditional distribution $f(y_r | y_{r-1})$ is required 
for its computation). Alternative estimators of $\eta_r$ can be 
used in place of $\bar D_r$ as long as the regularity conditions 
hold; see~\cite{hall1993estimation} for further details.

\begin{thm}
\label{thm:ed}
Let $T_n^{(r)}(\theta)$ be the quantity computed based on a random 
sample of size $n$ from the GEV$_r$ distribution with parameters $\theta$ 
and assume that $H_0^{(r-1)}$ is true. Then $T_n^{(r)}$ converges in 
distribution to $N(0, 1)$ as $n \to \infty$.
\end{thm}

Note that in Theorem~\ref{thm:ed}, $T_n^{(r)}$ is computed 
from a random sample of size $n$ from a GEV$_r$ distribution. 
If the random sample were from a distribution in the domain of
attraction of a GEV$_r$ distribution, the quality of the approximation 
of the GEV$_r$ distribution to the $r$ largest order statistics 
depends on the size of each block $B \to \infty$ with $r \ll B$.
The block size $B$ is not to be confused with the sample size $n$. 
Assuming $\xi > -0.5$, the proposed ED statistic for $H_0^{(r)}$ 
is $T_n^{(r)} (\hat\theta_n)$, where $\hat\theta_n$ is the MLE of 
$\theta$ with the $r$ largest order statistics for the 
GEV$_{r}$ distribution. Since $\hat\theta_n$ is consistent
for $\theta$ with $\xi > -0.5$, $T_n^{(r)}(\hat\theta_n)$ has the same
limiting distribution as $T_n^{(r)}(\theta)$ under $H_0^{(r)}$.

To assess the convergence of $T_n^{(r)}(\hat\theta_n)$ to $N(0, 1)$,
1000 GEV$_r$ replicates were simulated under configurations of 
$r \in \{2, 5, 10\}$, $\xi \in \{-0.25, 0, 0.25\}$, and 
$n \in \{50, 100\}$. Their quantiles are compared with those of 
$N(0, 1)$ via quantile-quantile plots (not presented). 
It appears that a larger sample size is needed for the normal
approximation to be good for larger $r$ and negative $\xi$.
This is expected because larger $r$ means higher dimension 
of the data, and because the MLE only exists for $\xi > -0.5$
\citep{smith1985maximum}. For $r$ less than 5 and $\xi \ge 0$, 
the normal approximation is quite good; it appears satisfactory 
for sample size as small as 50. For $r$ up to 10, sample size 
100 seems to be sufficient.

\section{Simulation Results}
\label{ch2:sim}

\subsection{Size}
\label{ch2:size}

The empirical sizes of the tests are investigated first.
For the score test, the parametric bootstrap version and 
the multiplier bootstrap version are equivalent asymptotically,
but may behave differently for finite samples.
It is of interest to know how large a sample size is needed 
for the two versions of the score test to hold their levels.
Random samples of size $n$ were generated from the 
GEV$_r$ distribution with $r \in \{1, 2, 3, 4, 5, 10\}$, 
$\mu = 0$, $\sigma = 1$, and $\xi \in \{-0.25, 0, 0.25\}$.
All three parameters $(\mu, \sigma, \xi)$ were estimated.

When the sample size is small, there can be numerical difficulty 
in obtaining the MLE. 
For the multiplier bootstrap score and ED test, the MLE only 
needs to obtained once, for the dataset being tested. However, in 
addition, the parametric bootstrap score test must obtain 
a new sample and obtain the MLE for each bootstrap replicate. 
To assess the severity of this issue, 10,000 datasets were 
simulated for $\xi \in \{-0.25, 0, 0.25\}$, $r \in \{1,2,3,4,5,10\}$, 
$n \in \{25, 50\}$, and the MLE was attempted for each dataset. 
Failure never occurred for $\xi \geq 0$. With $\xi = -0.25$ and 
sample size 25, the highest failure rate of 0.69\% occurred for 
$r=10$. When the sample size is 50, failures only occurred 
when $r=10$, at a rate of 0.04\%.

For the parametric bootstrap score test with sample size
$n \in \{25, 50, 100\}$,
Table~\ref{tab:gevr_pbsize} summarizes the empirical size of the 
test at nominal levels 1\%, 5\%, and 10\% obtained from 1000 
replicates, each carried out with bootstrap sample size $L = 1000$. 
Included only are the cases that converged successfully. Otherwise, 
the results show that the agreement between the empirical levels 
and the nominal level is quite good for samples as small as 25, 
which may appear in practice when long record data is not available.

\begin{table}[tbp]
  \centering
  \caption{Empirical size (in \%) for the parametric bootstrap score
    test under the null distribution GEV$_r$, with $\mu=0$ and
    $\sigma=1$ based on 1000 samples, each with bootstrap sample 
    size $L = 1000$.}
    \begin{tabular}{cc rrr rrr rrr}
      \toprule
      Sample Size & $r$ & \multicolumn{3}{c}{25} & \multicolumn{3}{c}{50} & \multicolumn{3}{c}{100} \\
      \cmidrule(lr){3-5}\cmidrule(lr){6-8}\cmidrule(lr){9-11}
      Nominal Size & & 1.0 & 5.0 & 10.0  & 1.0 & 5.0 & 10.0  & 1.0 & 5.0 & 10.0 \\
      \midrule
    $\xi=-0.25$ & 1     & 0.4   & 2.8   & 6.0   & 1.1   & 4.8   & 9.3   & 0.6   & 4.1   & 8.0 \\
                & 2     & 0.1   & 2.6   & 6.0   & 0.8   & 3.4   & 6.5   & 0.6   & 3.6   & 8.1 \\
                & 3     & 0.3   & 2.5   & 5.0   & 0.8   & 4.3   & 7.7   & 1.1   & 4.8   & 8.1 \\
                & 4     & 0.3   & 1.8   & 5.4   & 0.6   & 3.1   & 6.9   & 1.1   & 5.1   & 8.8 \\
                & 5     & 0.4   & 2.4   & 6.7   & 0.4   & 3.3   & 8.3   & 0.6   & 3.1   & 6.5 \\
                & 10    & 2.7   & 5.3   & 8.7   & 0.5   & 3.9   & 8.4   & 0.7   & 4.2   & 7.6 \\[6pt]
    $\xi=0$     & 1     & 1.3   & 5.2   & 8.9   & 1.6   & 5.3   & 9.0   & 0.8   & 4.7   & 9.3 \\
                & 2     & 1.4   & 5.1   & 9.4   & 2.0   & 4.9   & 10.0  & 1.0   & 4.3   & 9.9 \\
                & 3     & 1.7   & 6.2   & 10.9  & 2.1   & 6.0   & 10.2  & 0.8   & 4.9   & 9.8 \\
                & 4     & 1.5   & 4.5   & 8.5   & 1.3   & 6.0   & 10.2  & 1.0   & 4.4   & 9.8 \\
        	    & 5     & 1.6   & 5.8   & 10.4  & 2.4   & 6.2   & 9.9   & 1.2   & 5.0   & 9.7 \\
        	    & 10    & 1.5   & 4.0   & 7.3   & 1.5   & 4.3   & 8.9   & 0.7   & 4.6   & 8.2 \\[6pt]
    $\xi=0.25$  & 1     & 1.7   & 4.5   & 9.7   & 2.6   & 7.1   & 11.5  & 1.1   & 4.6   & 9.1 \\
       	        & 2     & 1.8   & 5.1   & 8.7   & 1.8   & 4.4   & 8.5   & 0.5   & 2.9   & 7.5 \\
       		    & 3     & 1.5   & 4.4   & 9.4   & 1.5   & 3.7   & 8.1   & 1.0   & 4.2   & 9.4 \\
       		    & 4     & 1.2   & 3.3   & 8.1   & 1.1   & 4.6   & 9.7   & 1.1   & 4.3   & 9.6 \\
       		    & 5     & 1.7   & 4.4   & 9.4   & 1.1   & 4.2   & 8.6   & 0.6   & 4.8   & 9.6 \\
       		    & 10    & 1.1   & 4.6   & 8.3   & 1.5   & 6.1   & 10.7  & 1.0   & 3.9   & 8.5 \\
    \bottomrule
    \end{tabular}
  \label{tab:gevr_pbsize}
\end{table}

For the multiplier bootstrap score test, the results for $n$ 
$\in \{25, 50, 100, 200\}$ 
are summarized in Table~\ref{tab:gevr_multsize}. 
When the sample size is less than 100, it appears that there
is a large discrepancy between the empirical and nominal level.
For $\xi \in \{0, 0.25\}$, there is reasonable agreement between the 
empirical level and the nominal levels for sample size at least 100.
For $\xi = - 0.25$ and sample size at least 100, the agreement
is good except for $r=1$, in which case, the empirical level is 
noticeably larger than the nominal level. 
This may be due to different rates of convergence for
various $\xi$ values as $\xi$ moves away from $-0.5$.
It is also interesting to note that, everything else being 
held, the agreement becomes better as $r$ increases.
This may be explained by the more information provided by larger $r$
for the same sample size $n$, as can be seen directly in the Fisher
information matrix \citep[pp. 247--249]{tawn1988extreme}. 
For the most difficult case with $\xi = -0.25$ and $r = 1$,
the agreement gets better as sample size increases and becomes 
acceptable when sample size was 1000 (not reported).

\begin{table}[tbp]
	\footnotesize
    \centering
    \caption{Empirical size (in \%) for multiplier bootstrap score
      test under the null distribution GEV$_r$, with $\mu=0$ and
      $\sigma=1$. 1000 samples, each with bootstrap sample size
      $L = 1000$ were used. Although not shown, the empirical 
      size for $r=1$ and $\xi=-0.25$ becomes acceptable when 
      sample size is 1000.}
    \begin{tabular}{cc rrr rrr rrr rrr}
      \toprule
      Sample Size & $r$ & \multicolumn{3}{c}{25} & \multicolumn{3}{c}{50} & \multicolumn{3}{c}{100} & \multicolumn{3}{c}{200} \\
      \cmidrule(lr){3-5}\cmidrule(lr){6-8}\cmidrule(lr){9-11}\cmidrule(lr){12-14}
      Nominal Size & & 1.0 & 5.0 & 10.0  & 1.0 & 5.0 & 10.0  & 1.0 & 5.0 & 10.0  & 1.0 & 5.0 & 10.0 \\
      \midrule
    $\xi=-0.25$ 	 & 1     & 7.0   & 13.4  & 18.9  & 6.3   & 13.8  & 19.6  & 5.4   & 11.4  & 16.3  & 5.4   & 10.5  & 15.2 \\
                     & 2     & 2.0   & 6.9   & 13.4  & 1.3   & 6.4   & 12.4  & 1.6   & 6.9   & 13.6  & 1.4   & 6.7   & 12.8 \\
                     & 3     & 2.1   & 5.8   & 11.7  & 1.1   & 5.9   & 11.1  & 1.1   & 5.0   & 10.8  & 1.5   & 5.9   & 11.8 \\
                     & 4     & 3.3   & 7.2   & 12.3  & 1.1   & 4.9   & 10.8  & 1.0   & 5.2   & 11.9  & 1.1   & 5.6   & 10.6 \\
                     & 5     & 3.6   & 9.0   & 14.0  & 2.3   & 6.8   & 11.2  & 1.1   & 6.2   & 10.6  & 1.1   & 4.5   & 9.3  \\
                     & 10    & 2.0   & 7.0   & 10.3  & 2.6   & 7.4   & 12.8  & 2.1   & 6.4   & 10.1  & 1.4   & 6.4   & 11.6 \\
    $\xi=0$  		 & 1     & 3.3   & 8.4   & 15.3  & 2.2   & 7.0   & 12.5  & 1.1   & 4.6   & 9.2   & 1.3   & 6.1   & 11.2 \\
                     & 2     & 2.8   & 8.7   & 14.4  & 1.8   & 7.5   & 13.0  & 0.9   & 5.7   & 10.3  & 0.5   & 5.0   & 10.6 \\
                     & 3     & 6.1   & 12.1  & 16.5  & 3.0   & 7.2   & 12.2  & 1.5   & 6.0   & 10.4  & 1.4   & 4.5   & 9.8  \\
                     & 4     & 5.1   & 10.4  & 14.5  & 3.6   & 10.1  & 14.9  & 1.0   & 5.6   & 10.3  & 1.1   & 5.4   & 10.6 \\
                     & 5     & 4.2   & 9.0   & 14.5  & 2.2   & 8.2   & 12.5  & 1.7   & 6.5   & 12.0  & 1.8   & 6.2   & 12.5 \\
                     & 10    & 3.1   & 9.2   & 14.4  & 2.4   & 6.4   & 9.8   & 0.6   & 4.6   & 9.0   & 1.1   & 3.8   & 9.3  \\
    $\xi=0.25$ 		 & 1     & 1.8   & 6.7   & 13.7  & 1.3   & 4.7   & 10.4  & 0.8   & 4.4   & 11.5  & 0.9   & 4.9   & 11.4 \\
                     & 2     & 5.7   & 12.7  & 17.1  & 4.7   & 9.9   & 14.9  & 3.5   & 7.4   & 11.6  & 3.2   & 7.9   & 11.7 \\
                     & 3     & 7.1   & 12.2  & 16.5  & 5.3   & 9.4   & 14.8  & 4.2   & 8.4   & 12.5  & 1.8   & 6.1   & 10.7 \\
                     & 4     & 5.4   & 9.8   & 16.8  & 3.7   & 9.0   & 13.4  & 2.6   & 6.0   & 11.4  & 1.2   & 4.9   & 11.2 \\
                     & 5     & 4.4   & 10.1  & 15.8  & 3.5   & 8.2   & 13.6  & 2.4   & 7.4   & 11.4  & 1.6   & 5.9   & 10.0 \\
                     & 10    & 3.3   & 8.9   & 15.3  & 2.4   & 6.6   & 12.3  & 1.6   & 5.8   & 10.9  & 1.7   & 6.6   & 12.4 \\
    \bottomrule
  \end{tabular}
  \label{tab:gevr_multsize}
\end{table}

To assess the convergence of $T_n^{(r)}(\hat\theta_n)$ to $N(0, 1)$,
10,000 replicates of the GEV$_r$ distribution were simulated with 
$\mu=0$ and $\sigma=1$ for each configuration of $r \in \{2,5,10\}$, 
$\xi \in \{-0.25, 0, 0.25\}$, and $n \in \{50, 100\}$. A rejection for 
nominal level $\alpha$, is denoted if 
$|T_n^{(r)}(\hat\theta_n)| > |Z_{\frac{\alpha}{2}}|$, 
where $Z_{\frac{\alpha}{2}}$ is the $\alpha/2$ percentile of the N(0,1) 
distribution. Using this result, the empirical size of the ED test can 
be summarized, and the results are presented in Table~\ref{tab:gevr_edsize}.

\begin{table}[tbp]
  \centering
  \caption{Empirical size (in \%) for the entropy difference (ED) 
  	test under the null distribution GEV$_r$, with $\mu=0$ and
    $\sigma=1$ based on 10,000 samples.}
    \begin{tabular}{cc rrr rrr}
      \toprule
      Sample Size & $r$ & \multicolumn{3}{c}{50} & \multicolumn{3}{c}{100} \\
      \cmidrule(lr){3-5}\cmidrule(lr){6-8}
      Nominal Size & & 1.0 & 5.0 & 10.0  & 1.0 & 5.0 & 10.0  \\
      \midrule
   $\xi = -0.25$ & 2     & 1.5   & 5.7   & 10.8  & 1.3   & 5.5   & 10.1 \\
          		 & 5     & 2.4   & 6.8   & 11.9  & 1.6   & 5.9   & 10.6 \\
         		 & 10    & 2.3   & 6.8   & 11.7  & 1.9   & 6.0   & 11.1 \\[6pt]
   $\xi = 0$     & 2     & 1.3   & 5.6   & 11.0  & 1.2   & 5.3   & 10.4 \\
          		 & 5     & 1.6   & 5.9   & 11.2  & 1.5   & 5.7   & 10.6 \\
          		 & 10    & 2.3   & 6.5   & 11.8  & 1.6   & 5.9   & 10.7 \\[6pt]
   $\xi = 0.25$  & 2     & 1.3   & 5.7   & 10.7  & 1.3   & 5.4   & 10.5 \\
          		 & 5     & 1.6   & 5.8   & 11.5  & 1.3   & 5.6   & 10.2 \\
          		 & 10    & 2.0   & 6.6   & 11.9  & 1.4   & 5.5   & 10.4 \\
    \bottomrule
    \end{tabular}
  \label{tab:gevr_edsize}
\end{table}

For sample size 50, the empirical size is above the nominal level 
for all configurations of $r$ and $\xi$. As the sample size increases 
from 50 to 100, the empirical size stays the same or decreases in 
every setting. For sample size 100, the agreement between nominal 
and observed size appears to be satisfactory for all configurations 
of $r$ and $\xi$. For sample size 50, the empirical size is 
slightly higher than the nominal size, but may be acceptable to some 
practitioners. For example, the empirical size for nominal size 10\% 
is never above 12\%, and for nominal size 5\%, empirical size is never 
above 7\%.

In summary, the multiplier bootstrap procedure of the 
score test can be used as a fast, reliable alternative to the parametric 
bootstrap procedure for sample size 100 or more when $\xi \ge 0$. 
When only small samples are available (less than 50 observations), the 
parametric bootstrap procedure is most appropriate since the multiplier 
version does not hold its size and the ED test relies upon samples of
size 50 or more for the central limit theorem to take effect.

\subsection{Power}
\label{ch2:power}

The powers of the score tests and the ED test are studied with
two data generating schemes under the alternative hypothesis.
In the first scheme, 4~largest order statistics were generated from 
the GEV$_4$ distribution with $\mu = 0$, $\sigma = 1$, and 
$\xi \in \{-0.25, 0, 0.25\}$, and the 5th one was generated from a KumGEV 
distribution right truncated by the 4th largest order statistic.
The KumGEV distribution is a generalization of the GEV distribution
\citep{eljabri2013new} with two additional parameters $a$ and $b$ 
which alter skewness and kurtosis.
Defining $G_r({\bf y})$ to be the distribution function of the 
GEV$_r$($\mu, \sigma, \xi$) distribution, the distribution function
of the KumGEV$_r$($\mu, \sigma, \xi, a, b$) is given by 
$F_r({\bf y}) = 1 - \{1 - [G_r({\bf y})]^a\}^b$ for $a>0$, $b>0$.
The score test and the ED test were applied to the top~5 
order statistics with sample size $n \in \{100, 200\}$.
When $a = b = 1$, the null hypothesis of GEV$_5$ is true.
Larger difference from 1 of parameters $a$ and $b$ means 
larger deviation from the null hypothesis of GEV$_5$.

\begin{table}
\footnotesize
  \centering
  \caption{Empirical rejection rate (in \%) of the multiplier score test 
  and the ED test in the first data generating scheme in Section~\ref{ch2:power} 
  from 1000 replicates.}
  \begin{tabular}{c rr rrrrrrrrr}
    \toprule
    Sample Size & $\xi$   & Test & \multicolumn{9}{c}{Value of $a$=$b$} \\
    \cmidrule(lr){4-12}
    &  &   & 0.4   & 0.6   & 0.8   & 1.0     & 1.2   & 1.4   & 1.6   & 1.8   & 2.0 \\
    \midrule
    100   & $-$0.25 & Score & 99.9  & 84.8  & 20.4  & 5.4   & 21.0  & 41.2  & 62.3  & 79.0  & 83.0 \\
          &         & ED   & 100.0 & 99.0  & 46.5  & 4.6   & 48.7  & 89.5  & 99.2  & 100.0 & 99.8 \\
          & 0       & Score & 100.0 & 87.0  & 21.6  & 7.4   & 24.2  & 48.9  & 67.8  & 79.6  & 89.4 \\
          &         & ED   & 100.0 & 98.8  & 40.0  & 5.2   & 40.6  & 87.2  & 98.5  & 100.0 & 99.7 \\
          & 0.25    & Score & 100.0 & 87.7  & 20.3  & 6.2   & 25.8  & 54.2  & 74.2  & 82.9  & 89.5 \\
          &         & ED   & 100.0 & 97.5  & 37.7  & 4.8   & 34.8  & 78.1  & 96.1  & 99.5  & 99.7 \\[6pt]
    200   & $-$0.25 & Score & 100.0 & 98.6  & 40.7  & 5.2   & 29.8  & 64.7  & 86.4  & 95.9  & 97.5 \\
          &         & ED   & 100.0 & 100.0 & 78.4  & 6.2   & 70.0  & 99.2  & 100.0 & 100.0 & 100.0 \\
          & 0       & Score & 100.0 & 99.4  & 44.6  & 6.1   & 34.9  & 75.0  & 92.4  & 97.3  & 98.6 \\
          &         & ED   & 100.0 & 99.9  & 75.0  & 5.5   & 64.6  & 98.1  & 99.8  & 100.0 & 100.0 \\
          & 0.25    & Score & 100.0 & 99.3  & 44.5  & 6.3   & 37.0  & 73.4  & 91.8  & 97.0  & 98.9 \\
          &         & ED   & 100.0 & 100.0 & 71.0  & 5.2   & 57.2  & 95.9  & 100.0 & 100.0 & 100.0 \\
          \bottomrule
    \end{tabular}
  \label{tab:gevr_kum}
\end{table}

Table~\ref{tab:gevr_kum} summarizes the empirical rejection 
percentages obtained with nominal size 5\%, 
for a sequence value of $a = b$ from 0.4 to 2.0, with increment 0.2.
Both tests hold their sizes when $a = b = 1$ and have substantial 
power in rejecting the null hypothesis for other values of $a = b$.
Between the two tests, the ED test demonstrated much higher power 
than the score test in the more difficult cases where the deviation
from the null hypothesis is small; for example, the ED test's power
almost doubled the score test's power for $a = b \in \{0.8, 1.2\}$.
As expected, the powers of both tests increase as $a = b$ moves
away from 1 or as the sample sizes increases.

In the second scheme, top~6 order statistics were generated
from the GEV$_{6}$ distribution with $\mu=0$, $\sigma=1$, and 
$\xi \in \{ -0.25, 0, 0.25\}$, and then the 5th order statistic
was replaced from a mixture of the 5th and 6th order statistics.
The tests were applied to the sample of first~5 order 
statistics with sample sizes $n \in \{100, 200\}$.
The mixing rate $p$ of the 5th order statistic took
values in $\{0.00, 0.10, 0.25, 0.50, 0.75, 0.90, 1.00\}$.
When $p = 1$ the null hypothesis of GEV$_5$ is true.
Smaller values of $p$ indicate larger deviations from the null.
Again, both tests hold their sizes when $p = 1$ and have 
substantial power for other values of $p$, which increases 
as $p$ decreases or as the sample sizes increases.
The ED test again outperforms the score test with almost
doubled power in the most difficult cases with $p \in \{0.75, 0.90\}$.
For sample size 100 with $p = 0.50$, for instance, the ED test 
has power above 93\% while the score test only has power above 69\%. 
The full results are seen in Table~\ref{tab:gevr_mix}.

\begin{table}
  \centering
  \caption{Empirical rejection rate (in \%) of the multiplier score test 
  and the ED test in the second data generating scheme  in Section~\ref{ch2:power} 
  from 1000 replicates.}
  \begin{tabular}{c rr rrrrrrr}
    \toprule
    Sample Size &  $\xi$   & Test & \multicolumn{7}{c}{Mixing Rate $p$} \\
    \cmidrule(lr){4-10}
    & & &    0.00   & 0.10   & 0.25   & 0.50   & 0.75   & 0.90   & 1.00 \\
    \midrule
    100   & $-$0.25 & Score & 99.7  & 99.5  & 95.5  & 69.4  & 24.1  & 7.8   & 5.8 \\
          &         & ED    & 100.0 & 100.0 & 100.0 & 97.7  & 51.8  & 10.9  & 6.2 \\
          & 0       & Score & 100.0 & 99.7  & 97.8  & 72.4  & 22.7  & 6.2   & 6.8 \\
          &         & ED    & 100.0 & 100.0 & 100.0 & 96.0  & 47.6  & 10.3  & 5.6 \\
          & 0.25    & Score & 99.9  & 99.7  & 96.6  & 70.8  & 24.7  & 5.8   & 5.3 \\
          &         & ED    & 100.0 & 100.0 & 99.9  & 93.6  & 43.4  & 9.8   & 5.2 \\[6pt]
    200   & $-$0.25 & Score & 99.9  & 100.0 & 99.7  & 95.6  & 43.4  & 11.4  & 5.1 \\
          &         & ED    & 100.0 & 100.0 & 100.0 & 100.0 & 83.6  & 20.0  & 5.8 \\
          & 0       & Score & 100.0 & 100.0 & 100.0 & 96.5  & 44.4  & 11.2  & 5.4 \\
          &         & ED    & 100.0 & 100.0 & 100.0 & 100.0 & 79.5  & 20.0  & 5.5 \\
          & 0.25    & Score & 100.0 & 100.0 & 100.0 & 97.2  & 46.9  & 9.2   & 5.5 \\
          &         & ED    & 100.0 & 100.0 & 100.0 & 99.7  & 72.5  & 17.9  & 4.2 \\
          \bottomrule
    \end{tabular}
  \label{tab:gevr_mix}
\end{table}

\section{Automated Sequential Testing Procedure}
\label{ch2:seq}

As there are $R$ hypotheses $H_0^{(r)}$, $r = 1, \ldots, R$,
to be tested in a sequence in the methods proposed, the 
sequential, multiple testing issue needs to be addressed.
Most methods for error control assume that all the tests can 
be run first and then a subset of tests are chosen to be rejected
\citep[e.g.,][]{benjamini2010discovering,benjamini2010simultaneous}.
The errors to be controlled are either the 
FWER \citep{shaffer1995multiple}, or the 
FDR \citep{Benjamini1995,BY2001}.
In contrast to the usual multiple testing procedures, however,
a unique feature in this setting is that the hypotheses must
be rejected in an ordered fashion: if $H_0^{(r)}$ is rejected,
$r < R$, then $H_0^{(k)}$ will be rejected for all $r < k \le R$.
Despite the extensive literature on multiple testing and
the more recent developments on FDR control and its variants, 
no definitive procedure has been available for error control in 
ordered tests until the recent work of \citet{g2015sequential}.

Consider a sequence of null hypotheses $H_1, \ldots, H_m$.
An ordered test procedure must reject $H_1, \ldots, H_k$
for some $k \in \{0, 1, \ldots, m\}$, which rules out the
classical methods for FDR control \citep{Benjamini1995}.
Let $p_1, \ldots, p_m \in [0, 1]$ be the corresponding 
p-values of the $m$ hypotheses such that $p_j$ is 
uniformly distributed over $[0,1]$ when $H_j$ is true. 
The methods of \citet{g2015sequential} transform the sequence
of p-values to a monotone sequence and then apply the original
Benjamini--Hochberg procedure on the monotone sequence.
They proposed two rejections rules, each returning a cutoff 
$\hat k$ such that $H_1, \ldots, H_{\hat k}$ are rejected.
The first is called ForwardStop,
\begin{equation}
\label{eq:forwardstop}
\hat{k}_{\mathrm{F}} = \max \left\{k \in \{1, \ldots, m\}: -\frac{1}{k} \sum_{i=1}^k \log(1-p_i) \leq \alpha \right\},
\end{equation}
and the second is called StrongStop,
\begin{equation}
\label{eq:strongstop}
\hat{k}_{\mathrm{S}} = \max \left\{k \in \{1, \ldots, m\}:
  \exp\Big(\sum_{j=k}^m \frac{\log p_j}{j}\Big)   \leq \frac{\alpha
    k}{m} \right\},
\end{equation}
where $\alpha$ is a pre-specified level. Both rules were 
shown to control the FDR at level $\alpha$ 
under the assumption of independent p-values.
ForwardStop sets the rejection threshold at the largest $k$ at which
the average of first $k$ transformed p-values is small enough.
As it does not depend on those p-values with later indices, 
this rule is robust to potential misspecification at later indices.
StrongStop offers a stronger guarantee than ForwardStop. 
If the non-null p-values indeed precede the null p-values, 
it controls the FWER at level $\alpha$ in addition to the FDR. 
Thus, for ForwardStop, this $\alpha$ refers to the FDR and for 
StrongStop, $\alpha$ refers to the FWER. 
As the decision to stop at $k$ depends on all the p-values after $k$,
its power may be harmed if, for example, the very last p-values
are slightly higher than expected under the null hypotheses.

To apply the two rules to our setting, note that our objective 
is to give a threshold $\hat r$ such that the first $\hat r$ 
of $m = R$ hypotheses are accepted instead of rejected.
Therefore, we put the p-values in reverse order:
let the ordered set of p-values $\{p_1, \ldots, p_R\}$ 
correspond to hypotheses $\{H_0^{(R)}, \ldots, H_0^{(1)}\}$.
The two rules give a cutoff $\hat{k} \in \{1, \ldots, R\}$ 
such that the hypotheses 
$H_0^{(R)}, \ldots, H_0^{(R - \hat{k} + 1)}$ are rejected.
If no $\hat{k} \in \{1, \ldots, R\}$ exists, then no rejection is made.

A caveat is that, unlike the setting of~\cite{g2015sequential}, 
the p-values of the sequential tests are dependent. 
Nonetheless, the ForwardStop and StrongStop procedures 
may still provide some error control. For example, in the 
non-sequential multiple testing scenario~\cite{BY2001} 
show that their procedure controls the FDR under certain 
positive dependency conditions, while~\cite{blanchard2009adaptive} 
implement adaptive versions of step-up procedures that 
provably control the FDR under unspecified dependence 
among p-values.

The empirical properties of the two rules for the tests 
in this paper are investigated in simulation studies.
To check the empirical FWER of the StrongStop rule, 
only data under the null hypotheses are needed.
With $R=10$, $\xi \in \{-0.25, 0.25\}$, $n \in \{30, 50, 100, 200\}$, 
$\mu=0$, and $\sigma=1$, 1000 GEV$_{10}$ samples were generated. 
For the ED, multiplier bootstrap score, and parametric bootstrap score
test, the observed FWER is compared to the expected rates at various 
nominal $\alpha$ control levels. The StrongStop procedure is 
used, as well as no error control (i.e. a rejection occurs any 
time the raw p-value is below the nominal level). The results of 
this simulation are presented in Figure~\ref{fig:gevr_FWERcontrol}.

It is clear that the StrongStop reasonably controls the FWER 
for the ED test and the agreement between the observed 
and expected rate increases as the sample size increases. 
For both the parametric and multiplier bootstrap 
versions of the score test however, the observed FWER is above 
the expected rate, at times 10\% higher. Regardless, it is 
apparent that using no error control results in an inflated 
FWER, and this inflation can only increase as the number of 
tests increase.

\begin{figure*}[tbp]
    \centering
      \includegraphics[width=\textwidth]{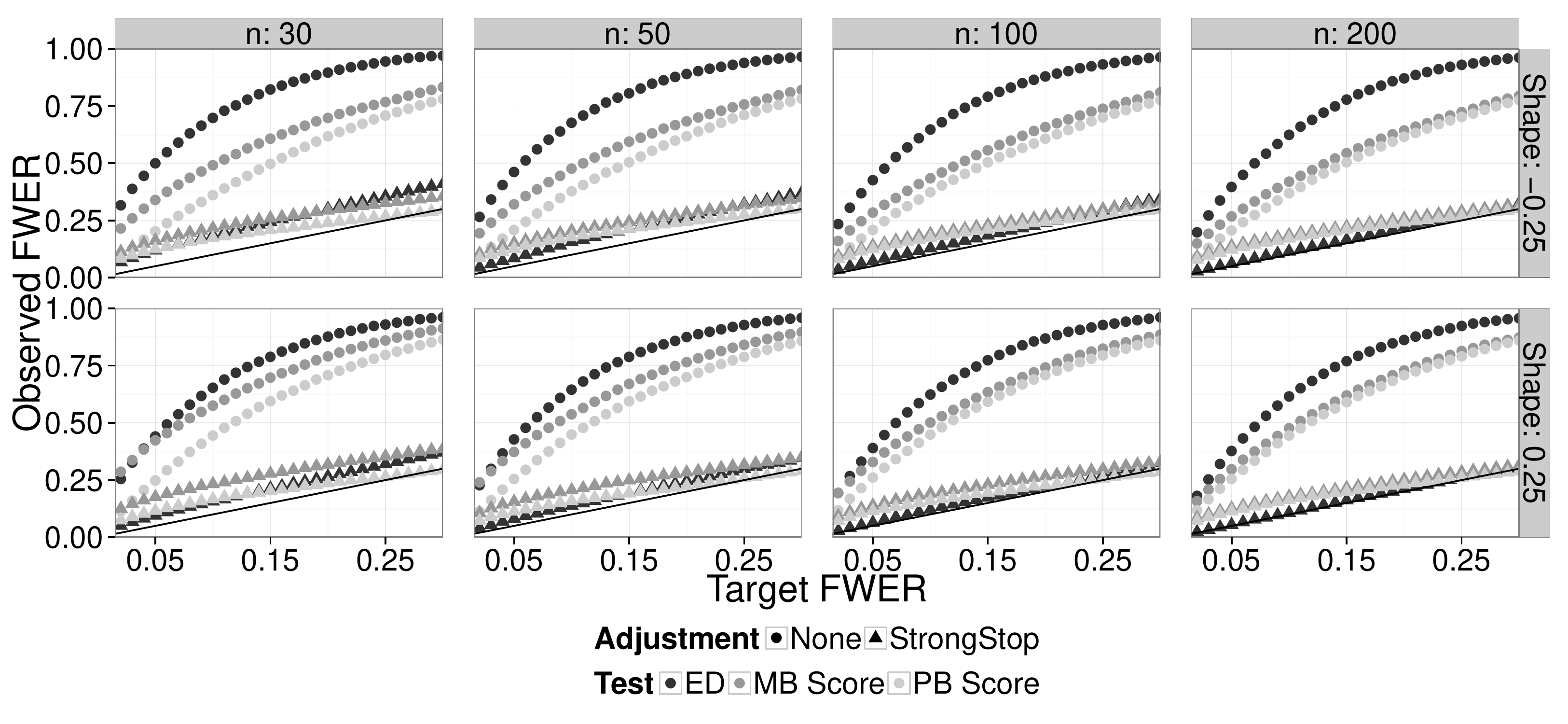}
    \caption{
      Observed FWER for the ED, parametric bootstrap (PB) score, and
      multiplier bootstrap (MB) score tests (using No Adjustment and StrongStop) 
      versus expected FWER at various nominal levels. The 45 degree line indicates 
      agreement between the observed and expected rates under $H_0$.}
    \label{fig:gevr_FWERcontrol}
\end{figure*}

To check the empirical FDR of the ForwardStop rule, 
data need to be generated from a non-null model.
To achieve this, consider the sequence of specification tests 
of GEV$_r$ distribution with $r \in \{1, \ldots, 6\}$, 
where the 5th and 6th order statistics are misspecified. 
Specifically, data from the GEV$_7$ distribution with 
$\mu=0$ and $\sigma=1$ were generated for $n$ blocks;
then the 5th order statistic is replaced with a 50/50 mixture of 
the 5th and 6th order statistics, and the 6th order statistic is 
replaced with a 50/50 mixture of the 6th and 7th order statistics. 
This is replicated 1000 times for each value of 
$\xi \in \{-0.25, 0.25\}$ and $n \in \{30, 50, 100, 200\}$. 
For nominal level $\alpha$, the observed FDR is defined as the number
of false rejections (i.e. any rejection of $r \leq 4$) divided by the
number of total rejections.

The results are presented in Figure~\ref{fig:gevr_choiceofr}. 
The plots show that the ForwardStop procedure controls 
the FDR for the ED test, while for both versions of the 
score test, the observed FDR is slightly higher than the 
expected at most nominal rates. Here, sample size does not 
appear to effect the observed rates.

\begin{figure*}[tbp]
\centering
\includegraphics[width=\textwidth]{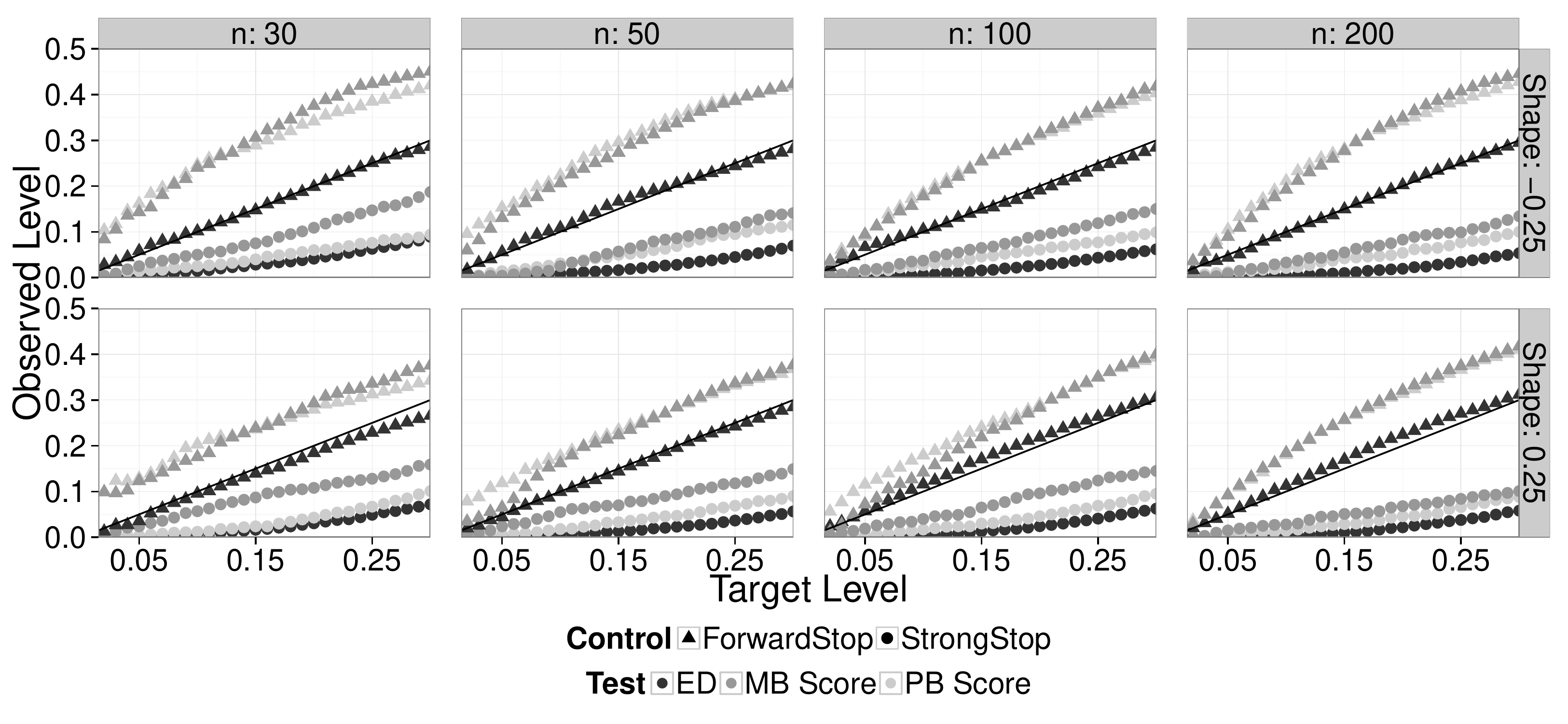}
\caption{Observed FDR (from ForwardStop) and observed FWER (from StrongStop) 
  versus expected FDR and FWER, respectively, at various nominal levels. 
  This is for the simulation setting described in Section~\ref{ch2:seq}, 
  using the ED, parametric bootstrap (PB) score, and multiplier bootstrap
  (MB) score tests. The 45 degree line indicates agreement between the 
  observed and expected rates.}
\label{fig:gevr_choiceofr}
\end{figure*}

Similarly, the observed FWER rate in this particular simulation setting 
can be found by taking the number of simulations with at least one 
false rejection (here, any rejection of $r \leq 4$) and dividing that number 
by the total number of simulations. This calculation is performed for a 
variety of nominal levels $\alpha$, using the StrongStop procedure. 
The results are presented in Figure~\ref{fig:gevr_choiceofr}. In 
this particular simulation setting, the StrongStop procedure controls the 
FWER for the ED test and both versions of the score test at all sample 
sizes investigated.

It is of interest to investigate the performance of the ForwardStop
and StrongStop in selecting $r$ for the $r$ largest order statistics method.
In the simulation setting described in the last paragraph, the
correct choice of $r$ should be 4, and a good testing procedure 
should provide a selection $\hat r$ close to 4.
The choice $\hat r \in \{0, \ldots, 6\}$ is recorded using the ED test
and bootstrap score tests with both ForwardStop and StrongStop. 
Due to space constraints, we choose to present one setting, 
where $\xi=0.25$ and $n=100$. The non-adjusted sequential procedure 
is also included, testing in an ascending manner from $r = 1$ and 
$\hat r$ is chosen by the first rejection found (if any). 
The results are summarized in Table~\ref{tab:gevr_multcomp}.

\begin{sidewaystable}
  \centering
  \scriptsize
  \caption{Percentage of choice of $r$ using the ForwardStop and
    StrongStop rules at various significance levels or FDRs, under ED,
    parametric bootstrap (PB) score, and multiplier bootstrap (MB)
    score tests,  with $n=100$ and $\xi=0.25$ for the simulation setting
    described in Section~\ref{ch2:seq}. Correct choice is $r=4$. }
    \begin{tabular}{rc rrrrrr rrrrrr rrrrrr}
    \toprule
    Test  & r     & \multicolumn{6}{c}{Unadjusted}                       & \multicolumn{6}{c}{ForwardStop}                      & \multicolumn{6}{c}{StrongStop} \\
    \cmidrule(lr){3-8} \cmidrule(lr){9-14}  \cmidrule(lr){15-20} 
      & Significance:  & 0.01  & 0.05  & 0.1   & 0.2  & 0.3 & 0.4 & 0.01  & 0.05  & 0.1   & 0.2  & 0.3 & 0.4 & 0.01  & 0.05  & 0.1   & 0.2  & 0.3 & 0.4 \\
         \midrule
    ED    & 6       & 19.0  & 3.4   & 1.5   & 0.5   & 0.0   & 0.0   & 86.6  & 69.8  & 58.5  & 43.0  & 30.0  & 22.4  & 52.4  & 22.2  & 13.1  & 5.4   & 1.7   & 1.0 \\
          & 5       & 1.9   & 2.1   & 1.1   & 0.7   & 0.2   & 0.1   & 1.3   & 1.5   & 1.1   & 0.9   & 0.2   & 0.0   & 25.0  & 18.9  & 13.6  & 6.7   & 2.9   & 1.4 \\
       \rowcolor{lightgray} & 4     & 76.2  & 79.9  & 70.2  & 50.5  & 35.2  & 22.3  & 12.0  & 25.1  & 31.7  & 33.4  & 31.6  & 26.3  & 22.6  & 58.9  & 72.9  & 84.9  & 89.0  & 85.7 \\
          & 3       & 0.8   & 4.1   & 7.6   & 11.8  & 15.3  & 15.1  & 0.1   & 3.4   & 6.2   & 12.4  & 17.0  & 18.5  & 0.0   & 0.0   & 0.3   & 2.2   & 4.5   & 6.9 \\
          & 2       & 1.1   & 5.3   & 9.5   & 16.2  & 19.9  & 22.8  & 0.0   & 0.1   & 1.9   & 5.4   & 9.4   & 11.5  & 0.0   & 0.0   & 0.1   & 0.7   & 1.8   & 4.4 \\
          & 1       & 1.0   & 5.2   & 10.1  & 20.3  & 29.4  & 39.7  & 0.0   & 0.1   & 0.6   & 4.9   & 11.8  & 21.3  & 0.0   & 0.0   & 0.0   & 0.1   & 0.1   & 0.6 \\
          & 0       & -     & -     & -     & -     & -     & -     & -     & -     & -     & -     & -     & -     & -     & -     & -     & -     & -     & - \\[6pt]
  PB Score & 6       & 35.8  & 16.1  & 8.5   & 1.8   & 0.6   & 0.2   & 53.5  & 33.0  & 23.4  & 13.9  & 8.6   & 5.9   & 40.6  & 25.1  & 18.8  & 12.1  & 7.6   & 5.6 \\
          & 5       & 2.5   & 1.8   & 0.8   & 0.5   & 0.2   & 0.1   & 1.9   & 1.3   & 0.7   & 0.8   & 0.4   & 0.4   & 29.8  & 37.7  & 29.2  & 17.5  & 10.5  & 6.6 \\
       \rowcolor{lightgray} & 4     & 58.4  & 68.1  & 63.9  & 46.1  & 31.3  & 19.8  & 42.6  & 57.8  & 58.9  & 51.1  & 42.0  & 31.1  & 29.3  & 36.5  & 50.5  & 65.7  & 73.5  & 74.6 \\
          & 3      & 0.5    & 2.3   & 4.5   & 6.8   & 7.1   & 8.0   & 1.4   & 4.2   & 9.0   & 15.5  & 16.0  & 17.1  & 0.0   & 0.4   & 1.1   & 3.3   & 4.8   & 5.6 \\
          & 2      & 0.6    & 3.0   & 4.9   & 10.3  & 11.7  & 12.7  & 0.5   & 1.7   & 3.0   & 7.3   & 10.2  & 12.3  & 0.1   & 0.1   & 0.2   & 1.0   & 2.4   & 4.5 \\
          & 1      & 0.8    & 3.7   & 6.9   & 13.9  & 18.3  & 18.5  & 0.1   & 1.3   & 2.3   & 4.6   & 8.3   & 9.6   & 0.2   & 0.2   & 0.2   & 0.4   & 1.2   & 2.7 \\
          & 0      & 1.4    & 5.0   & 10.5  & 20.6  & 30.8  & 40.7  & 0.0   & 0.7   & 2.7   & 6.8   & 14.5  & 23.6  & 0.0   & 0.0   & 0.0   & 0.0   & 0.0   & 0.4 \\[6pt]
  MB Score & 6      & 49.9   & 16.9  & 6.9   & 1.3   & 0.2   & 0.0   & 71.7  & 40.3  & 24.7  & 12.6  & 7.8   & 5.5   & 51.6  & 27.3  & 16.9  & 10.4  & 6.2   & 4.3 \\
          & 5      & 2.5    & 2.3   & 0.7   & 0.3   & 0.1   & 0.0   & 1.3   & 2.0   & 0.8   & 0.5   & 0.3   & 0.4   & 39.4  & 50.7  & 42.9  & 25.5  & 15.8  & 9.4 \\
      \rowcolor{lightgray}  & 4     & 38.3  & 59.6  & 59.1  & 44.0  & 31.2  & 18.4  & 26.6  & 53.3  & 59.3  & 49.9  & 40.1  & 28.0  & 6.2   & 18.5  & 35.0  & 55.4  & 62.3  & 64.8 \\
          & 3      & 1.6    & 2.8   & 4.0   & 6.6   & 7.5   & 6.0   & 0.3   & 2.8   & 7.4   & 15.7  & 16.0  & 17.9  & 0.6   & 1.2   & 2.5   & 3.6   & 7.0   & 7.8 \\
          & 2      & 2.7    & 4.4   & 7.1   & 11.0  & 10.0  & 10.6  & 0.1   & 0.6   & 3.4   & 8.6   & 9.7   & 9.5   & 0.7   & 0.8   & 1.2   & 2.8   & 5.3   & 7.0 \\
          & 1      & 4.2    & 8.3   & 10.6  & 14.0  & 19.5  & 20.0  & 0.0   & 0.9   & 2.7   & 4.7   & 7.9   & 7.8   & 1.5   & 1.5   & 1.5   & 2.3   & 3.4   & 6.4 \\
          & 0      & 0.8    & 5.7   & 11.6  & 22.8  & 31.5  & 45.0  & 0.0   & 0.1   & 1.7   & 8.0   & 18.2  & 30.9  & 0.0   & 0.0   & 0.0   & 0.0   & 0.0   & 0.3 \\
    \bottomrule
    \end{tabular}
  \label{tab:gevr_multcomp}
\end{sidewaystable}

In general, larger choices of $\alpha$ lead to a higher percentage of
$\hat r=4$ being correctly chosen with ForwardStop or StrongStop.
Intuitively, this is not surprising since a smaller $\alpha$ makes it 
more difficult to reject the `bad' hypotheses of $r\in \{5,6\}$. 
A larger choice of $\alpha$ also leads to a higher probability 
of rejecting too many tests; i.e. choosing $r$ too small. 
From the perspective of model specification, 
this is more desirable than accepting true negatives.  
A choice of 6, 5, or 0 is problematic, but choosing 1, 2, 
or 3 is acceptable, although some information is lost.
When no adjustment is used and an ascending sequential 
procedure is used, both tests have reasonable classification rates.
When $\alpha = 0.05$, the ED test achieves the correct choice of $r$
79.9\% of the time, with the parametric bootstrap and multiplier
bootstrap score tests achieving 68.1\% 
and 59.6\% respectively. Of course, as the number of tests (i.e., $R$)
increase, with no adjustment the correct classification rates will 
go down and the ForwardStop/StrongStop procedures will achieve better rates. 
This may not be too big an issue here as $R$ is typically small.
In the case where rich data are available and $R$ is big, the
ForwardStop and StrongStop becomes more useful as they are
designed to handle a large number of ordered hypothesis.

\section{Illustrations}
\label{ch2:app}

\subsection{Lowestoft Sea Levels}
\label{ch2:lowe}

Sea level readings in 60 and 15 minute intervals from a gauge 
at Lowestoft off the east coast of Britain during the years 
1964--2014 are available from the UK Tide Gauge Network website. 
The readings are hourly from 1964--1992 and in fifteen minute 
intervals from 1993 to present. Accurate estimates of extreme 
sea levels are of great interest. The current data are of better 
quality and with longer record than those used in 
\citet{tawn1988extreme} --- annual maxima during 1953--1983 and 
hourly data during 1970--78 and 1980--82.

Justification of the statistical model was considered in detail 
by~\cite{tawn1988extreme}. The three main assumptions needed to
justify use of the GEV$_r$ model are: 
(1) The block size $B$ is large compared to the choice of $r$;
(2) Observations within each block and
across blocks are approximately independent; and 
(3) The distribution of the block maxima follows GEV$_1$.
The first assumption is satisfied, by letting $R = 125$, and noting that
the block size for each year is $B = 365 \times 24 = 8760$ from
1964--1992 and $B = 365 \times 96 = 35040$ from 1993--2014. 
This ensures that $r \ll B$. 
The third assumption is implicitly addressed in the testing procedure;
if the goodness-of-fit test for the block maxima rejects, all
subsequent tests for $r > 1$ are rejected as well.

The second assumption can be addressed in this setting by
the concept of independent storms~\citep{tawn1988extreme}. 
The idea is to consider each storm as a separate event, 
with each storm having some storm length, say $\tau$. 
Thus, when selecting the $r$ largest values from each block, 
only a single contribution can be obtained from each storm, which 
can be considered the $r$ largest independent annual events. 
By choosing $\tau$ large enough, this ensures both approximate
independence of observations within each block and across blocks. 
The procedure to extract the independent $r$ largest annual events is
as follows: 
\begin{enumerate}
\item
Pick out the largest remaining value from the year (block) of interest. 

\item
Remove observations within a lag of $\tau / 2$ from both sides of the value 
chosen in step~1.

\item
Repeat (within each year) until the $r$ largest are extracted. 
\end{enumerate}

A full analysis is performed on the Lowestoft sea level data using 
$\tau = 60$ as the estimated storm length \citep{tawn1988extreme}. 
Using $R = 125$, both the parametric bootstrap score (with bootstrap
sample size $L = 10,000$) and ED test are applied sequentially on the
data. The p-values of the sequential 
tests (adjusted and unadjusted) can be seen in 
Figure~\ref{fig:lowestoft_pvals}. Due to the large number of tests, 
the adjustment for multiplicity is desired and thus, ForwardStop is 
used to choose $r$. For this dataset, the score test is more powerful 
than the ED test. With ForwardStop and the score test, 
Figure~\ref{fig:lowestoft_pvals} suggests that $r = 33$. 
The remainder of this analysis proceeds with the choice of $r = 33$. 
The estimated parameters and corresponding 95\% profile confidence intervals 
for $ r= 1$ through $r = 40$ are shown in Figure~\ref{fig:lowestoft_params}.

\begin{figure*}[tbp]
    \includegraphics[width=\textwidth]{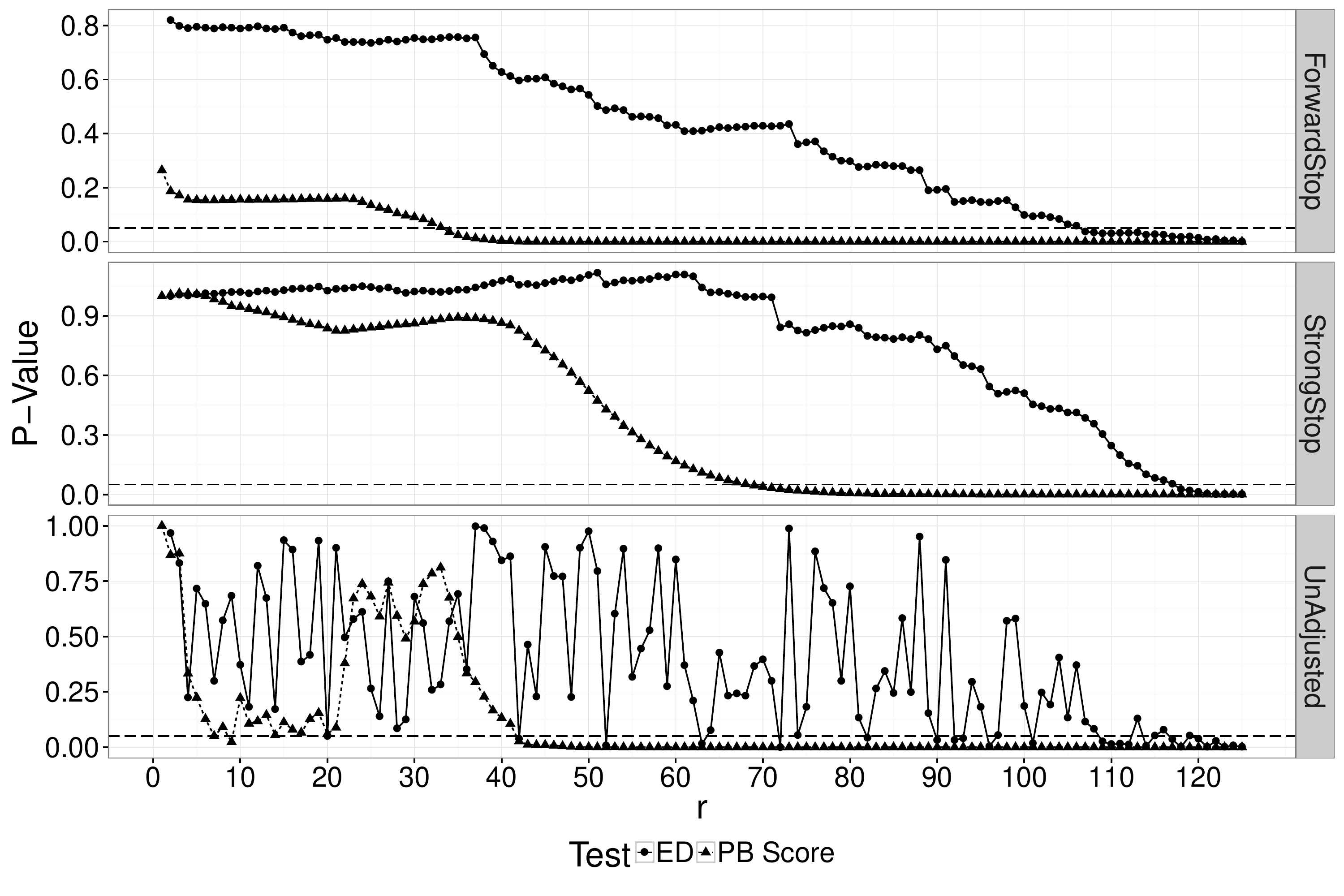}
    \caption{Adjusted p-values using ForwardStop, StrongStop, and raw (unadjusted) 
    p-values for the ED and PB Score tests applied to the Lowestoft sea level data. 
    The horizontal dashed line represents the 0.05 possible cutoff value.}
    \label{fig:lowestoft_pvals}
\end{figure*}

When $r = 33$, the parameters are estimated as $\hat{\mu} = 3.462\ (0.023)$, 
$\hat{\sigma}= 0.210\ (0.013)$, and $\hat{\xi}= -0.017\ (0.023)$, with
standard errors in parenthesis. 
An important risk measure is the $t$-year return level $z_t$
\citep[e.g.,][]{hosking1990moments,ribereau2008estimating,singo2012flood}.
It can be thought of here as the sea level that is exceeded once every 
$t$ years on average. Specifically, the $t$-year return level is 
the $1 - 1/t$ quantile of the GEV distribution (which can be substituted 
with corresponding quantities from the GEV$_r$ distribution), given by
\begin{equation}
\label{eq:gev_rl}
z_t = 
\begin{cases} 
  \mu - \frac{\sigma}{\xi}\big\{1 - [ -  \log(1-\frac{1}{t})]^{-\xi} \big\}, & \xi \neq 0, \\
  \mu - \sigma\log[-\log(1-\frac{1}{t})], & \xi = 0.
\end{cases}
\end{equation}
The return levels can be estimated with parameter values replaced with 
their estimates, and confidence intervals can be constructed using profile 
likelihood~\citep[e.g.,][p.57]{coles2001introduction}.

\begin{figure*}[tbp]
  \includegraphics[width=\textwidth]{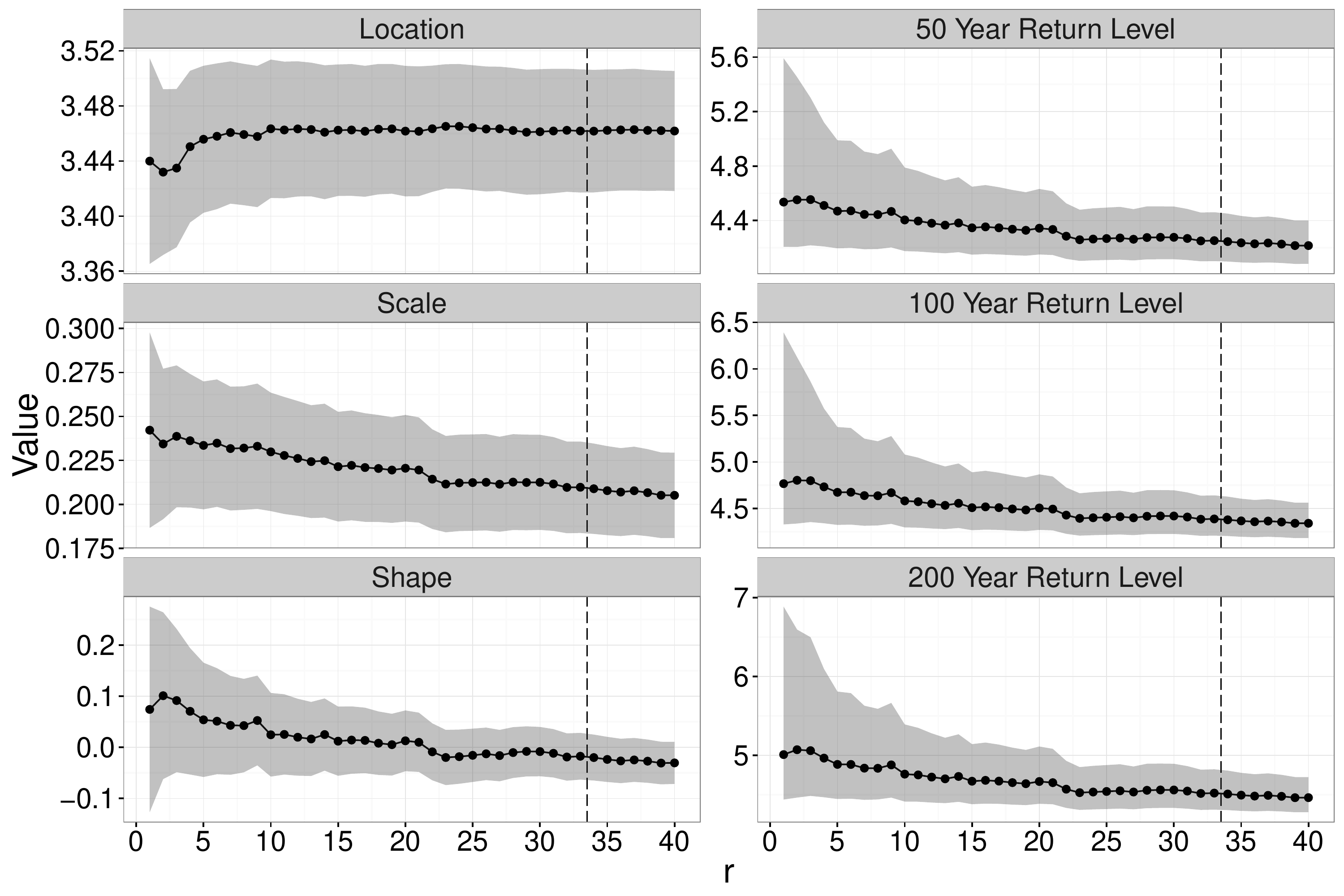}
  \caption{Location, scale, and shape parameter estimates, with 95\% delta 
  confidence intervals for $r=1, \ldots, 40$ for the Lowestoft sea level data. 
  Also included are the estimates and 95\% profile likelihood confidence 
  intervals for the 50, 100, and 200 year return levels. The vertical dashed 
  line represents the recommended cutoff value of $r$ from the analysis in 
  Section~\ref{ch2:lowe}.}
    \label{fig:lowestoft_params}
\end{figure*}

The 95\% profile likelihood confidence intervals for the 50, 100,  
and 200 year return levels (i.e. $z_{50}, z_{100}, z_{200}$) are given by 
$(4.102, 4.461)$, $(4.210, 4.641)$ and $(4.312, 4.824)$, respectively. 
The benefit of using $r=1$ versus $r=33$ can be seen in the 
return level confidence intervals in Figure~\ref{fig:lowestoft_params}. 
For example, the point estimate of the 100 year return level decreases 
slightly as $r$ increases and the width of the 95\% confidence interval 
decreases drastically from 2.061 ($r=1$) to 0.432 ($r=33$), 
as more information is used. The lower bound of the interval 
however remains quite stable, shifting from 4.330 to 4.210 --- less 
than a 3\% change. Similarly, the standard error of the shape parameter 
estimate decreases by over two-thirds when using $r=33$ versus $r=1$.

\subsection{Annual Maximum Precipitation: Atlantic City, NJ}
\label{ch2:AC}

The top 10 annual precipitation events (in centimeters) were taken from 
the daily records of a rain gauge station in Atlantic City, NJ from 
1874--2015. The year 1989 is missing, while the remaining records 
are greater than 98\% complete. This provides a total record length 
of 141 years. The raw data is a part of the Global Historical Climatology 
Network (GHCN-Daily), with an overview given by~\cite{menne2012overview}. 
The specific station identification in the dataset is USW00013724.

\begin{figure*}[tbp]
    \includegraphics[width=\textwidth]{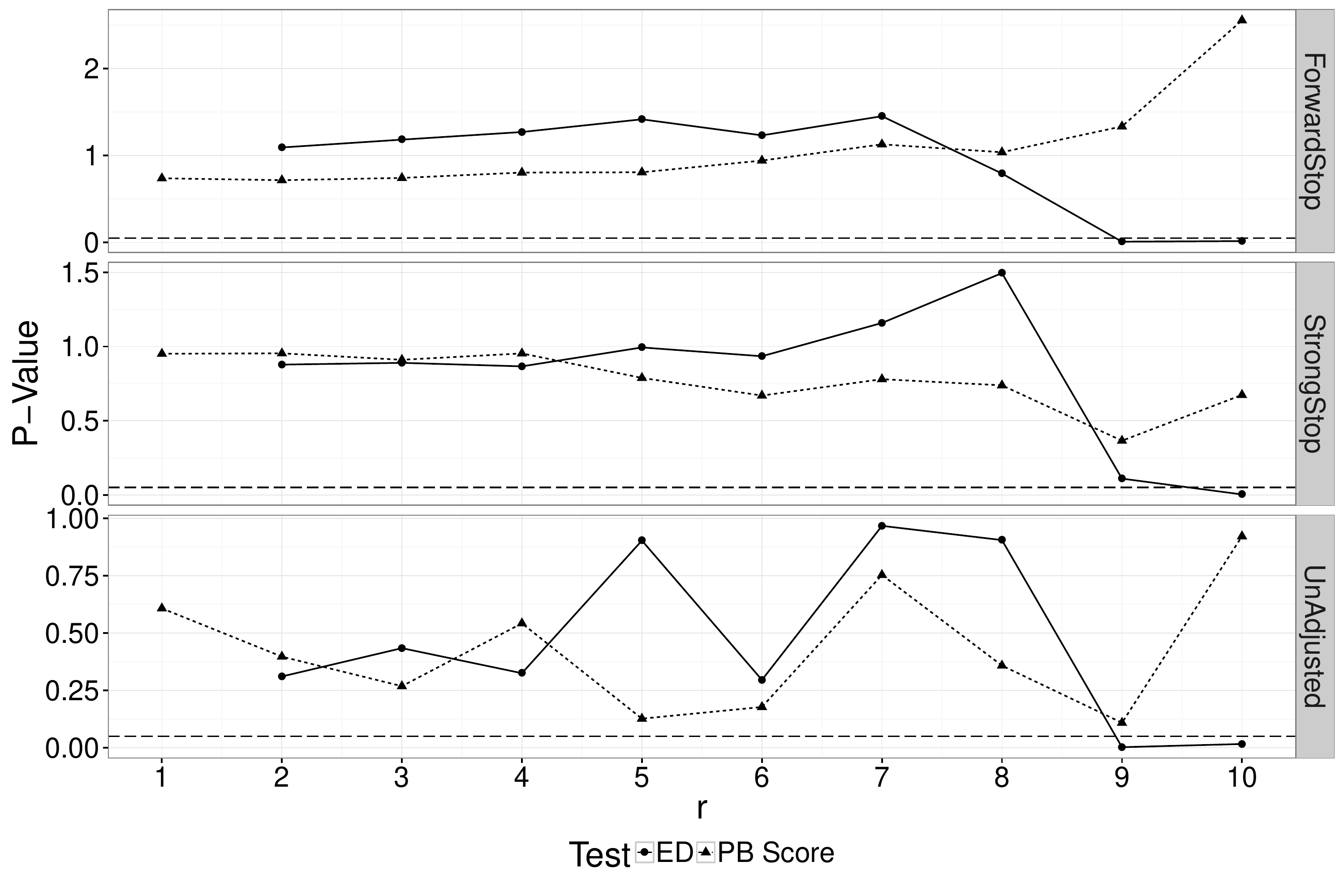}
    \caption{Adjusted p-values using ForwardStop, StrongStop, and raw (unadjusted) 
    p-values for the ED and PB Score tests applied to the Atlantic City precipitation 
    data. The horizontal dashed line represents the 0.05 possible cutoff value.}
    \label{fig:AC_pvals}
\end{figure*}

Unlike for the Lowestoft sea level data, a rather small value is set for 
$R$ at $R = 10$ because of the much lower frequency of the daily data. 
Borrowing ideas from Section~\ref{ch2:lowe}, a storm length of 
$\tau = 2$ is used to ensure approximate independence of observations.
Both the parametric bootstrap score (with $L = 10,000$) and ED test 
are applied sequentially on the data. 
The p-values of the sequential tests 
(ForwardStop, StrongStop, and unadjusted) are shown in 
Figure~\ref{fig:AC_pvals}. 
The score test does not pick up anything.
The ED test obtains p-values 0.002 and 0.016, respectively, 
for $r=9$ and $r=10$, which translates into a rejection 
using ForwardStop. Thus, Figure~\ref{fig:AC_pvals} suggests that $r=8$ 
be used for the analysis.

\begin{figure*}[tbp]
    \includegraphics[width=\textwidth]{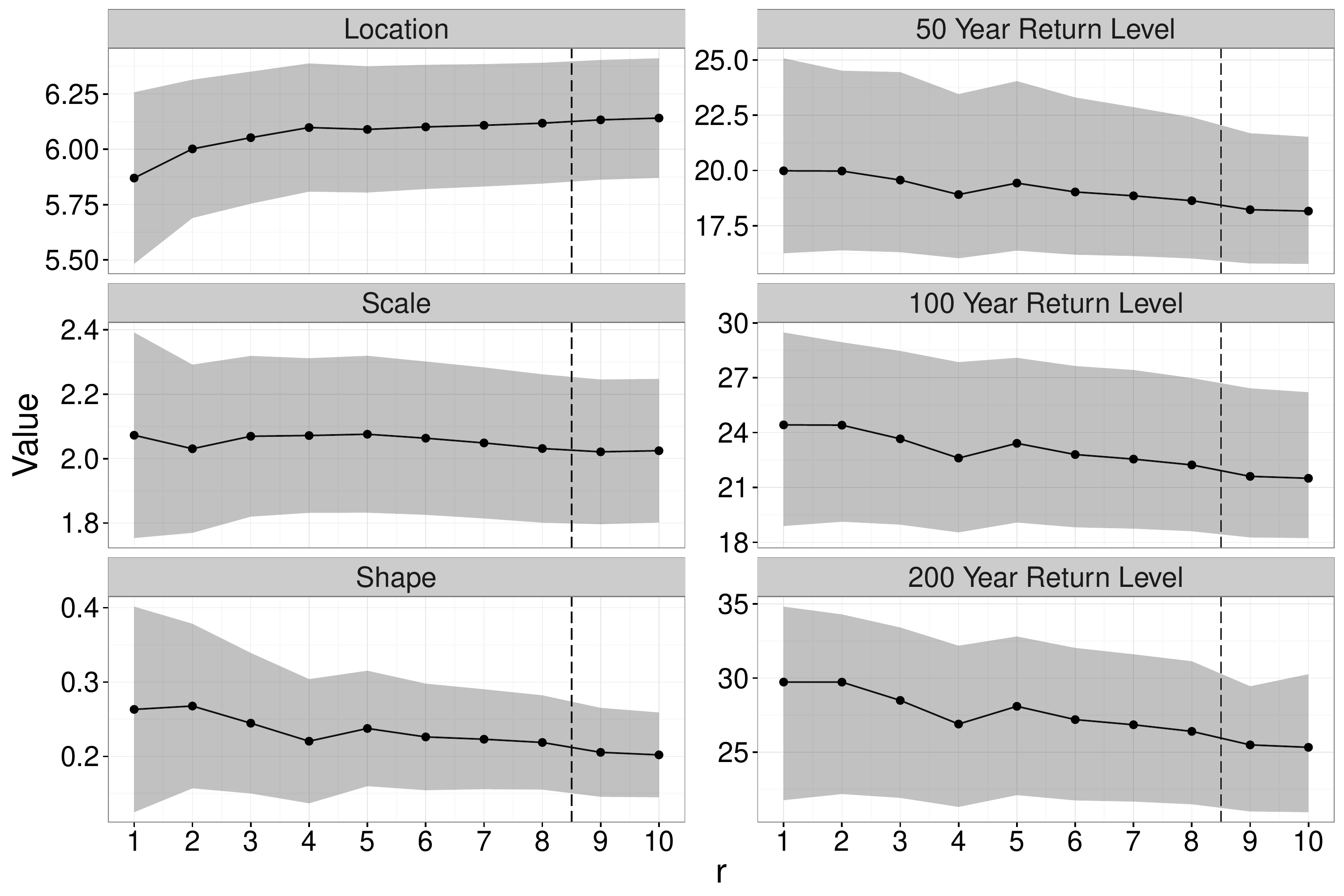}
    \caption{Location, scale, and shape parameter estimates, with 95\% delta 
    confidence intervals for $r=1$ through $r=10$ for the Atlantic City 
    precipitation data. Also included are the estimates and 95\% profile 
    likelihood confidence intervals for the 50, 100, and 200 year return 
    levels. The vertical dashed line represents the recommended cutoff value 
    of $r$ from the analysis in Section~\ref{ch2:AC}.}
    \label{fig:AC_params}
\end{figure*}

With $r=8$, the estimated parameters are given as $\hat{\mu} = 6.118\ (0.139)$, 
$\hat{\sigma}= 2.031\ (0.118)$, and $\hat{\xi}= 0.219\ (0.032)$. This suggests 
a heavy upper tail for the estimated distribution (i.e. $\hat{\xi} > 0$). The 
progression of parameters and certain return level estimates can be seen in 
Figure~\ref{fig:AC_params}. The 50, 100, and 200 year return level 95\% 
confidence intervals for $r=8$ are calculated using the profile likelihood 
method and are given by $(16.019, 22.411)$, $(18.606, 26.979)$, and 
$(21.489, 31.136)$, respectively. The advantages of using $r=8$ versus the 
block maxima for analysis are quite clear from Figure~\ref{fig:AC_params}. 
The standard error of the shape parameter decreases from 0.071 to 0.032, a 
decrease of over 50\%. Similarly, the 50 year return level 95\% 
confidence intervals decreases in width by over 25\%.

\section{Discussion}
\label{ch2:disc}

We propose two model specification tests for a fixed number of
largest order statistics as the basis for selecting $r$ for the
$r$ largest order statistics approach in extreme value analysis. 
The score test has two versions of bootstrap procedure: the
multiplier bootstrap method providing a fast, large sample alternative 
to the parametric bootstrap method, with a speedup of over 100 times.
The ED test depends on an asymptotic normal approximation of the testing 
statistic, which becomes acceptable for sample size over 50.
It assumes that the $r-1$ top order statistics included already 
fits the GEV$_{r-1}$ distribution. Therefore, the initial 
hypothesis at $r = 1$ needs to be tested with the score tests.
Both tests hold their size better when the shape parameter is 
further away from the lower limit of $-0.5$ or sample size is larger. 
When only small samples are available (50 observations or less), 
the parametric bootstrap score test is recommended.

Alternative versions of the ED test have been explored.
One may define the testing statistics as the difference in entropy 
between GEV$_1$ and GEV$_r$, instead of between GEV$r-1$ and GEV$_r$.
Nonetheless, it appeared to require a larger sample to hold 
its size from our simulation studies (not reported).
In the calculation of $T_n^{(r)}$, the block maxima MLE
$\hat\theta_n^{(1)}$ can be used as an estimate for $\theta$ 
in place of $\hat\theta_n^{(r)}$.
Again, in our simulation studies, this version of the ED test
was too conservative, thus reducing the power,
when the sample size was not large enough.
This may be explained in that the resulting $\hat{S}_{D_r}$ 
underestimates $S_{D_r}$.

If the initial block maxima distribution ($r=1$) is rejected or a 
small $r$ is selected, it may be worth noting that the choice of $B$, 
the block size must be sufficiently large relative to $R$, the number of 
largest order statistics to test. The rate of convergence for the $r$ largest 
order statistics has been previously studied~\citep{dziubdziela1978convergence, 
deheuvels1986strong, deheuvels1989strong, falk1989best}. In particular, 
\cite{falk1989best} showed that the best rate of convergence to the 
distribution in~\eqref{eq:gevr_pdf} is $O(r / B)$, uniformly in $r$ 
and can be as slow as $O(r / \log(B))$ if the underlying distribution 
is normal. Hence, an initial rejection of $r=1$ may indicate that a larger 
block size is needed.

Naively, the tests may be performed sequentially for each 
$r \in \{1, \ldots, R\}$, for a prefixed, usually small $R \ll B$, 
at a certain significance level until $H_0^{(r)}$ is rejected.
The issue of multiple, sequential testing is addressed in detail 
by adapting two very recent stopping rules to control the FDR and the 
FWER that are developed specifically for situations when hypotheses 
must be rejected in an ordered fashion \citep{g2015sequential}. 
Based on the higher power of the ED test, a broad recommendation would 
be to use the score test for $r=1$, then proceed using the ED test 
for testing $r=2, \ldots, R$. The choice of stopping rule to use 
in conjunction depends on the desired error control. ForwardStop 
controls the FDR and thus generally provides more rejections 
(i.e., smaller selected $r$), while StrongStop provides stricter 
control (FWER), thus possessing less power (i.e., larger selected $r$). 
Typically, correct specification of the GEV$_r$ distribution is 
crucial in estimating large quantiles (e.g., return levels), so 
as a general guideline ForwardStop would be suggested.

\chapter{Automated Threshold Selection in the POT Approach}
\label{ch:gpd}

\section{Introduction}
\label{ch3:intr}

Return levels, the levels of a measure of interest that is expected 
to be exceeded on average once every certain period of time
(return period), are commonly a major goal of inference in 
extreme value analysis. As opposed to block-based methods, threshold 
methods involve modeling data exceeding a suitably chosen 
high threshold with the generalized Pareto distribution (GPD)
\citep{balkema1974residual, pickands1975statistical}.
Choice of the threshold is critical in obtaining accurate 
estimates of model parameters and return levels.
The threshold should be chosen high enough for the 
excesses to be well approximated by the GPD to minimize bias, 
but not so high to substantially increase the variance of the estimator 
due to reduction in the sample size (the number of exceedances).

Although it is widely accepted in the statistics community that 
the threshold-based approach, in particular peaks-over-threshold (POT), 
use data more efficiently than the block maxima method 
\citep[e.g.,][]{caires2009comparative}, it is much less used than the
block maxima method in some fields such as climatology. 
The main issue is the need to conduct the analyses over many
locations, sometimes over hundreds of thousands of 
locations~\citep[e.g.,][]{kharin2007changes, kharin2013changes}, and
there is a lack of efficient procedures that can automatically select
the threshold in each analysis. 
For example, to make a return level map of annual maximum daily 
precipitation for three west coastal US states of California, Oregon, 
and Washington alone, one needs to repeat the estimation procedure
including threshold selection,  at each of the hundreds of stations.
For the whole US, thousands of sites need to be processed.
A graphical based diagnosis is clearly impractical. 
It is desirable to have an intuitive automated threshold selection
procedure in order to use POT in analysis.

Many threshold selection methods are available in the literature;
see \citet{scarrott2012review} and \citet{caeiro2016threshold} 
for recent reviews. 
Among them, graphical diagnosis methods are arguably the most 
popular. The mean residual life (MRL) plot \citep{davison1990models} 
uses the fact that, if $X$ follows a GPD, then for $v>u$, the 
MRL $E[X-v | X>v]$, if existing, is linear in $v$. 
The threshold is chosen to be the smallest $u$ such that the 
sample MRL is approximately linear above this point.
Parameter stability plots check whether the estimates of GPD 
parameters, possibly transformed to be comparable across 
different thresholds, are stable above some level of threshold.
\citet{drees2000make} suggested the Hill plot, which plots 
the Hill estimator of the shape parameter based on the top $k$
order statistics against $k$. Many variants of the Hill plot
have been proposed \citep[Section~4]{scarrott2012review}.
The threshold is the $k$th smallest order statistic beyond
which the parameter estimates are deemed to be stable.
The usual fit diagnostics such as quantile plots, return 
level plots, and probability plots can be helpful too, as
demonstrated in \citet{coles2001introduction}.
Graphical diagnostics give close inspection of the data,
but they are quite subjective and disagreement on a particular 
threshold can occur even among experts in the field; see, for 
example, a convincing demonstration of unclear choices using 
the Fort Collins precipitation data in Figures~\ref{fig:gpd_mrl_plot}, 
\ref{fig:gpd_thresh_plot}, and \ref{fig:gpd_hill_plot} which 
are taken from \citet{scarrott2012review}.

Other selection methods can be grouped into various categories.
One is based on the asymptotic results about estimators of 
properties of the tail distribution. The threshold is selected by 
minimizing the asymptotic mean squared error (MSE) of the 
estimator of, for example, tail index \citep{beirlant1996}, 
tail probabilities \citep{hall1997estimation}, or
extreme quantiles \citep{ferreira2003optimising}. 
Theoretically sound as these methods are, their finite sample
properties are not well understood. Some require second order
assumptions, and computational (bootstrap) based estimators 
can require tuning parameters~\citep{danielsson2001using} or may 
not be satisfactory for small samples~\citep{ferreira2003optimising}. 
Irregardless, such resampling methods may be quite 
time-consuming in an analysis involving many locations.

A second category of methods are based on goodness-of-fit 
of the GPD, where the threshold is selected as the lowest level
above which the GPD provides adequate fit to the exceedances
\citep[e.g.,][]{davison1990models, dupuis1999, 
choulakian2001goodness, northrop2014improved}.
Goodness-of-fit tests are simple to understand and perform,
but the multiple testing issue for a sequence of tests in an ordered 
fashion have not been addressed to the best of our knowledge.
Methods in the third category are based on mixtures of a GPD for the 
tail and another distribution for the ``bulk'' joined at the threshold 
\citep[e.g.,][]{macdonald2011flexible, wadsworth2012likelihood}.
Treating the threshold as a parameter to estimate, these methods
can account for the uncertainty from threshold selection in inferences. 
However, there is little known about the asymptotic properties of 
this setup and how to ensure that the bulk and tail models are robust 
to one another in the case of misspecification.

Some automated procedures have been proposed. 
The simple naive method is \emph{a priori} or fixed threshold
selection based on expertise on the subject matter at hand. 
Various rules of thumb have been suggested; for example, 
select the top 10\% of the data
\citep[e.g.,][]{dumouchel1983estimating}, 
or the top square root of the sample size
\citep[e.g.,][]{ferreira2003optimising}.
Such one rule for all is not ideal in climate applications where high 
heterogeneity in data properties is the norm. The proportion of the number 
of rain days can be very different from wet tropical climates to dry 
subtropical climates; therefore the number of exceedances over the same 
time period can be very different across different climates. Additionally, 
the probability distribution of daily precipitation can also be different in 
different climates, affecting the speed the tails converge to the 
GPD~\citep{raoult2003rate}. 
Goodness-of-fit test based procedures can be 
automated to select the lowest one in a sequence of thresholds, at 
which the goodness-of-fittest is not 
rejected~\citep[e.g.,][]{choulakian2001goodness}. 
The error control, however, is challenging because of the ordered
nature of the hypotheses,  and the usual methods from multiple testing
such as false discovery rate
\citep[e.g.,][]{benjamini2010discovering, benjamini2010simultaneous}
cannot be directly applied.

We propose an automated threshold selection procedure 
based on a sequence of goodness-of-fit tests with error 
control for ordered, multiple testing.
The very recently developed stopping rules for ordered 
hypotheses in \citet{g2015sequential} are adapted to control 
the familywise error rate (FWER), the probability of at least one type
I error in the whole family of tests \citep{shaffer1995multiple}, or
the false discovery rate (FDR),  the expected proportion of
incorrectly rejected null hypotheses among 
all rejections \citep{Benjamini1995, BY2001}.
They are applied to four goodness-of-fit tests at each candidate
threshold, including the Cram\'er--Von Mises test, 
Anderson--Darling test, Rao's score test, and Moran's test.
For the first two tests, the asymptotic null distribution of the
testing statistic is unwieldy \citep{choulakian2001goodness}.
Parametric bootstrap puts bounds on the approximate p-values
which would reduce power of the stopping rules and lacks the 
ability to efficiently scale. We propose a fast approximation 
based on the results of \citet{choulakian2001goodness} to 
facilitate the application. The performance of the procedures 
are investigated in a large scale simulation study, and 
recommendations are made. The procedure is applied to annual 
maximum daily precipitation return level mapping for three 
west coastal states of the US. Interesting findings are revealed 
from different stopping rules. The automated threshold selection 
procedure has applications in various fields, especially when batch 
processing of massive datasets is needed.

The outline of this chapter is as follows. Section~\ref{ch3:seq_testing} 
presents the generalized Pareto model, its theoretical justification, 
and how to apply the automated sequential threshold testing procedure. 
Section~\ref{ch3:tests} introduces the tests proposed to be used in the 
automated testing procedure. A simulation study demonstrates the power 
of the tests for a fixed threshold under various misspecification 
settings and it is found that the Anderson--Darling test is most 
powerful in the vast majority of cases. A large scale simulation study in 
Section~\ref{ch3:sim} demonstrates the error control and performance 
of the stopping rules for multiple ordered hypotheses, both under 
the null GPD and a plausible alternative distribution. 
In Section~\ref{ch3:app}, we return to our motivating application and
derive return levels for extreme precipitation at hundreds of west
coastal US stations to demonstrate the usage of our method. 
A final discussion is given in Section~\ref{ch3:disc}.

\section{Automated Sequential Testing Procedure}
\label{ch3:seq_testing}

Threshold methods for extreme value analysis are based on
that, under general regularity conditions, the only possible
non-degenerate limiting distribution of properly rescaled 
exceedances of a threshold $u$ is the GPD as $u \to
\infty$~\citep[e.g.,][]{pickands1975statistical}.  
The GPD has cumulative distribution function given 
in~\eqref{eq:gpd_cdf} and also has the property that for 
some threshold $v > u$, the excesses follow a GPD with the 
same shape parameter, but a modified scale 
$\sigma_v = \sigma_u + \xi(v - u)$.

Let $X_1, \ldots, X_n$ be a random sample of size $n$.
If $u$ is sufficiently high, the exceedances $Y_i = X_i - u$ for all 
$i$ such that $X_i > u$ are approximately a random sample from a GPD.
The question is to find the lowest threshold such that the GPD fits
the sample of exceedances over this threshold adequately.
Our solution is through a sequence of goodness-of-fit tests
\citep[e.g.,][]{choulakian2001goodness} or model specification tests
\citep[e.g.,][]{northrop2014improved} for the GPD to the exceedances
over each candidate threshold in an increasing order.
The multiple testing issues in this special ordered setting are
handled by the most recent stopping rules in \citet{g2015sequential}.

Consider a fixed set of candidate thresholds $u_1 < \ldots < u_l$.
For each threshold, there will be $n_i$ excesses, $i=1, \ldots, l$. 
The sequence of null hypotheses can be stated as
\begin{center}
$H_0^{(i)}$: The distribution of the $n_i$ exceedances above 
$u_i$ follows the GPD.
\end{center}
For a fixed $u_i$, many tests are available for this $H_0^{(i)}$.
An automated procedure can begin with $u_1$ and continue until 
some threshold $u_i$ provides an acceptance of $H_0^{(i)}$
\citep{choulakian2001goodness, thompson2009automated}.  
The problem, however, is that unless the test has high power, 
an acceptance may happen at a low threshold by chance and, 
thus, the data above the chosen threshold is contaminated. 
One could also begin at the threshold $u_l$ and descend until a
rejection occurs, but this would result in an increased type I error 
rate. The multiple testing problem obviously needs to be addressed, but 
the issue here is especially challenging because these tests are ordered; 
if $H_0^{(i)}$ is rejected, then $H_0^{(k)}$ 
should be rejected for all $1 \leq k < i$.
Despite the extensive literature on multiple testing and
the more recent developments on FDR control and its variants
\citep[e.g.,][]{Benjamini1995, BY2001, benjamini2010discovering,
  benjamini2010simultaneous},
no definitive procedure has been available for error control in 
ordered tests until the recent work of \citet{g2015sequential}.

We adapt the stopping rules of \citet{g2015sequential} to the
sequential testing of (ordered) null hypotheses $H_1, \ldots, H_l$.
Let $p_1, \ldots, p_l \in [0, 1]$ be the corresponding 
p-values of the $l$ hypotheses.
\citet{g2015sequential} transform the sequence of p-values to 
a monotone sequence and then apply the original method of
\citet{Benjamini1995} on the monotone sequence.
Two rejection rules are constructed, each of which returns a 
cutoff $\hat k$ such that $H_1, \ldots, H_{\hat k}$ are rejected. 
If no $\hat{k} \in \{1, \ldots, l\}$ exists, 
then no rejection is made. The first is called 
ForwardStop~\eqref{eq:forwardstop} and the second 
StrongStop~\eqref{eq:strongstop}. In the notation of 
\eqref{eq:forwardstop} and \eqref{eq:strongstop}, 
$m$ refers to the total number of tests to be performed, 
which is equivalent to the number of hypotheses, $l$.

Under the assumption of independence among the tests, 
both rules were shown to control the FDR at level $\alpha$. 
In our setting, stopping at $k$ implies that goodness-of-fit 
of the GPD to the exceedances at the first $k$ thresholds 
$\{u_1, \ldots, u_k\}$ is rejected.
In other words, the first set of $k$ null hypotheses 
$\{H_1, \ldots, H_k\}$ are rejected. 
At each $H_0^{(i)}$, ForwardStop 
is a transformed average of the previous and current p-values, 
while StrongStop only accounts for the current and subsequent 
p-values. StrongStop provides an even stronger guarantee of error 
control; that is that the FWER is controlled at level $\alpha$. 
The tradeoff for stronger error control is reduced power to reject. 
Thus, the StrongStop rule tends to select a lower threshold than the 
ForwardStop rule. This is expected since higher thresholds are
more likely to approximate the GPD well, and thus provide higher
p-values. In this sense, ForwardStop could be thought of as more 
conservative (i.e., stopping at higher threshold by
rejecting more thresholds).

The stopping rules, combined with the sequential hypothesis 
testing, provide an automated selection procedure --- all 
that is needed are the levels of desired control for the 
ForwardStop and StrongStop procedures, and a set of thresholds. 
A caveat is that the p-values of the sequential tests here
are dependent, unlike the setup of \citet{g2015sequential}.
Nonetheless, the stopping rules may still provide some reasonable 
error control as their counter parts in the non-sequential multiple
testing scenario \citep{BY2001, blanchard2009adaptive}.
A simulation study is carried out in Section~\ref{ch3:sim} to assess 
the empirical properties of the two rules.

\section{The Tests}
\label{ch3:tests}

The automated procedure can be applied with any valid test
for each hypothesis $H_0^{(i)}$ corresponding to threshold $u_i$. 
Four existing goodness-of-fit tests that can be used are presented.
Because the stopping rules are based on transformed p-values, it is
desirable to have testing statistics whose p-values can be accurately
measured; bootstrap based tests that put a lower bound on the p-values 
(1 divided by the bootstrap sample size) may lead to premature stopping.
For the remainder of this section, the superscript $i$ is dropped.
We consider the goodness-of-fit of GPD to a sample of size $n$
of exceedances $Y = X - u$ above a fixed threshold $u$.

\subsection{Anderson--Darling and Cram\'{e}r--von Mises Tests} 
\label{ch3:ad_cvm}

The Anderson--Darling and the Cram\'{e}r--von Mises tests for the 
GPD have been studied in detail \citep{choulakian2001goodness}.
Let $\hat\theta_n$ be the maximum likelihood estimator (MLE) of 
$\theta$ under $H_0$ from the the observed exceedances. 
Make the probability integral transformation based on $\hat\theta_n$
$z_{(i)} = F(y_{(i)} | \hat\theta_n)$, as in~\eqref{eq:gpd_cdf}, for 
the order statistics of the exceedances $y_{(1)} < \ldots < y_{(n)}$. 
The Anderson--Darling statistic is 
\begin{equation}
\label{eq:gpd_ad}
A_n^2 = -n - \frac{1}{n} \sum_{i=1}^n (2i - 1)\Big[\log(z_{(i)}) + \log(1 - z_{(n+1-i)}) \Big].
\end{equation}
The Cram\'{e}r--von Mises statistic is
\begin{equation}
\label{eq:gpd_cvm}
W_n^2 = \sum_{i=1}^n \big[z_{(i)} - \frac{2i - 1}{2n}\big]^2 + \frac{1}{12n}.
\end{equation}

The asymptotic distributions of $A_n^2$ and $W_n^2$ are unwieldy,
both being sum of weighted chi-squared variables with one degree
of freedom with weight found from the eigenvalues of an integral
equation \citep[Section~6]{choulakian2001goodness}.
The distributions depend only on the estimate of $\xi$.
The tests are often applied by referring to a table of a few upper
tail percentiles of the asymptotic distributions
\citep[Table~2]{choulakian2001goodness}, or through bootstrap.
In either case, the p-values are truncated by a lower bound.
Such truncation of a smaller p-value to a larger one can be proven to
weaken the stopping rules given in~\eqref{eq:forwardstop} 
and~\eqref{eq:strongstop}. In order to apply these tests in the 
automated sequential setting, more accurate p-values for the tests 
are needed.

We provide two remedies to the table in \citet{choulakian2001goodness}.
First, for $\xi$ values in the range of $(-0.5, 1)$, which is
applicable for most applications, we enlarge the table to a much 
finer resolution through a pre-set Monte Carlo method.
For each $\xi$ value from $-0.5$ to $1$ incremented by $0.1$, 
2,000,000 replicates of $A_n^2$ and $W_n^2$ are generated with
sample size $n = 1,000$ to approximate their asymptotic distributions.
A grid of upper percentiles from 0.999 to 0.001 for each $\xi$ value
is produced and saved in a table for future fast reference.
Therefore, if $\hat\xi_n$ and the test statistic falls in the range
of the table, the p-value is computed via log-linear interpolation.

The second remedy is for observed test statistics that are greater
than that found in the table (implied p-value less than 0.001).
As Choulakian pointed out (in a personal communication),
the tails of the asymptotic distributions are exponential, which 
can be confirmed using the available tail values in the table. 
For a given $\hat\xi$, regressing $-\log(\text{p-value})$ on the
upper tail percentiles in the table, for example, from 0.05 to
0.001, gives a linear model that can be extrapolated to approximate
the p-value of statistics outside of the range of the table.
This approximation of extremely small p-values help reduce loss of
power in the stopping rules.

The two remedies make the two tests very fast and are 
applicable for most applications with $\xi \in (-0.5, 1)$. 
For $\xi$ values outside of $(-0.5, 1)$, although slow, one can 
use parametric  bootstrap to obtain the distribution of the test
statistic, understanding that the p-value has a lower bound.
The methods are implemented in R package \texttt{eva} \citep{Rpkg:eva}.

\subsection{Moran's Test}
\label{ch3:moran}

Moran's goodness-of-fit test is a byproduct of the maximum product
spacing (MPS)  estimation for estimating the GPD parameters. 
MPS is a general method that allows efficient parameter estimation in
non-regular cases where the MLE fails or the
information matrix does not exist \citep{cheng1983estimating}.
It is based on the fact that if the exceedances are indeed from the
hypothesized distribution, their probability integral transformation would
behave like a random sample from the standard uniform distribution.
Consider the ordered exceedances $y_{(1)}  < \ldots < y_{(n)}$.
Define the spacings as 
\begin{equation*}
D_{i}(\theta) = F(y_{(i)} | \theta) - F(y_{(i-1)} | \theta)
\end{equation*}
for $i=1, 2, \ldots, n+1$ with $F(y_{(0)} | \theta) \equiv 0$ 
and $F(y_{(n+1)} | \theta) \equiv 1$. 
The MPS estimators are then found by minimizing
\begin{equation}
\label{eq:gpd_morans}
M(\theta) = - \sum\limits_{i=1}^{n+1} \log D_{i}(\theta).
\end{equation}
As demonstrated in \citet{wong2006note}, the MPS method is especially
useful for GPD estimation in the non-regular cases of \citet{smith1985maximum}.
In cases where the MLE exists, the MPS estimator may
have an advantage of being numerically more stable for small samples,
and have the same properties as the MLE asymptotically.

The objective function evaluated at the MPS estimator 
$\check\theta$ is Moran's statistic \citep{moran1953random}.
\citet{cheng1989goodness} showed that under the null hypothesis,
Moran's statistic is normally distributed and when properly centered 
and scaled, has an asymptotic chi-square approximation. 
Define
\begin{align*}
& \mu_M = (n+1)\big(\log(n+1) + \gamma\big) - \frac{1}{2} - \frac{1}{12(n+1)}, \\
& \sigma_M^2 = (n+1)\left(\frac{\pi^2}{6} - 1\right) - \frac{1}{2} - \frac{1}{6(n+1)},
\end{align*}
where $\gamma$ is Euler's constant. 
Moran's goodness-of-fit test statistic is
\begin{equation}
\label{eq:gpd_morans_gof}
T(\check{\theta}) = \frac{M(\check{\theta}) + 1 - C_1}{C_2},
\end{equation}
where $C_1 = \mu_M - (n/ 2)^{\frac{1}{2}} \sigma_M$ and
$C_2 = (2n)^{- \frac{1}{2}} \sigma_M$.
Under the null hypothesis, $T(\check{\theta})$ asymptotically 
follows a chi-square distribution with $n$ degrees of freedom. 
\citet{wong2006note} show that for GPD data, the test empirically
holds its size for samples as small as ten.

\subsection{Rao's Score Test}
\label{ch3:score}

\citet{northrop2014improved} considered a piecewise constant 
specification of the shape parameter as the alternative hypothesis.
For a fixed threshold $u$, a set of $k$ higher thresholds are
specified to  generate intervals on the support space. 
That is, for the set of thresholds $u_0 = u < u_1 < \ldots < u_k$,
the shape parameter is given a piecewise representation
\begin{equation}
\label{eq:piecewise}
\xi(x) = \begin{cases}
\xi_i & u_i < x \leq u_{i+1} \quad  i=0, \ldots, k-1,  \\
\xi_k & x > u_k.
\end{cases}
\end{equation}
The null hypothesis is tested as 
$H_0: \xi_0 = \xi_1 = \cdots =\xi_k$. 
Rao's score test has the advantage that only restricted MLE
$\tilde\theta$ under $H_0$ is needed, in contrast to the
asymptotically equivalent likelihood ratio test or Wald test.
The testing statistic is
\[
S(\tilde{\theta}) = 
U(\tilde{\theta})^T I^{-1}(\tilde{\theta}) U(\tilde{\theta}),
\]
where $U$ is the score function and $I$ is the fisher information
matrix, both evaluated at the restricted MLE $\tilde\theta$.
Given that $\xi > -0.5$~\citep{smith1985maximum}, the asymptotic null 
distribution of $S$ is $\chi^2_{k}$.

\citet{northrop2014improved} tested suitable thresholds in an ascending
manner, increasing the initial threshold $u$ and continuing the 
testing of $H_0$. They suggested two possibilities for automation. 
First, stop testing as soon as an acceptance occurs, but the p-values
are not guaranteed to be non-decreasing for higher starting thresholds. 
Second, stop as soon as all p-values for testing at
subsequent higher thresholds are above a certain significance level. 
The error control under multiple, ordered testing were not addressed.

\subsection{A Power Study}
\label{ch3:power}

The power of the four goodness-of-fit tests are examined 
in an individual, non-sequential testing framework.
The data generating schemes in \citet{choulakian2001goodness} 
are used, some of which are very difficult to distinguish 
from the GPD:
\begin{itemize}
\item
Gamma with shape 2 and scale 1.

\item
Standard lognormal distribution (mean 0 and scale 1 on log scale).

\item
Weibull with scale 1 and shape 0.75.

\item
Weibull with scale 1 and shape 1.25.

\item
50/50 mixture of GPD($1, -0.4$) and GPD($1, 0.4$).

\item
50/50 mixture of GPD($1, 0$) and GPD($1, 0.4$).

\item
50/50 mixture of GPD($1, -0.25$) and GPD($1, 0.25$).
\end{itemize}
Finally, the GPD($1, 0.25$) was also used to check the empirical 
size of each test. Four sample sizes were considered: 50, 100, 
200, 400. For each scenario, 10,000 samples are generated.
The four tests were applied to each sample, with a 
rejection recorded if the p-value is below 0.05. 
For the score test, a set of thresholds were set according 
to the deciles of the generated data.

\begin{table}[htbp]
  \centering
  \caption{Empirical rejection rates of four goodness-of-fit tests for
    GPD under various data generation schemes described in Section~\ref{ch3:power} 
    with nominal size 0.05. GPDMix(a, b) refers to a 50/50 mixture of GPD(1, a) 
    and GPD(1, b). 
}
    \begin{tabular}{l rrrr rrrr}
    \toprule
    Sample Size & \multicolumn{4}{c}{50} & \multicolumn{4}{c}{100} \\
    \cmidrule(lr){2-5}\cmidrule(lr){6-9}
    Test & Score & Moran & AD & CVM  & Score & Moran & AD & CVM \\
    \midrule
    Gamma(2, 1) & 7.6   & 9.7   & 47.4  & 43.5  & 8.0   & 14.3  & 64.7  & 59.7 \\
    LogNormal & 6.0   & 5.9   & 13.3  & 8.6   & 5.4   & 8.2   & 28.3  & 23.4 \\
    Weibull(0.75) & 11.5  & 7.8   & 55.1  & 23.5  & 12.1  & 9.5   & 65.1  & 39.4 \\
    Weibull(1.25) & 6.6   & 11.3  & 29.1  & 27.3  & 5.6   & 12.5  & 20.8  & 19.2 \\
    GPDMix($-$0.4, 0.4) & 11.4  & 7.5   & 19.2  & 9.9   & 16.0  & 8.6   & 24.3  & 20.4 \\
    GPDMix(0, 0.4) & 7.8   & 6.0   & 6.5  & 5.9   & 7.4   & 5.9   & 9.6   & 5.6 \\
    GPDMix($-$0.25, 0.25) & 8.1   & 6.5   & 6.0  & 7.4   & 8.9   & 6.5   & 11.1  & 8.4 \\
    GPD(1, 0.25) & 6.9   & 5.5   & 6.7   & 6.1   & 5.0   & 5.2   & 5.2   & 5.2 \\
    \midrule
    Sample Size & \multicolumn{4}{c}{200} & \multicolumn{4}{c}{400} \\
    \cmidrule(lr){2-5}\cmidrule(lr){6-9}
    Test & Score & Moran & AD & CVM  & Score & Moran & AD & CVM \\
    \midrule
    Gamma(2, 1) & 15.4  & 23.3  & 95.3  & 93.1  & 36.5  & 42.2  & 100.0 & 100.0 \\
    LogNormal & 5.7   & 11.9  & 69.3  & 59.7  & 7.9   & 19.1  & 97.8  & 95.0 \\
    Weibull(0.75) & 16.5  & 10.7  & 84.8  & 66.4  & 32.1  & 14.6  & 98.2  & 93.0 \\
    Weibull(1.25) & 8.0   & 14.7  & 40.9  & 36.7  & 15.5  & 19.2  & 79.8  & 74.6 \\
    GPDMix($-$0.4, 0.4) & 31.5  & 9.7   & 45.1  & 44.0  & 63.8  & 11.9  & 79.9  & 80.2 \\
    GPDMix(0, 0.4) & 8.9   & 6.5   & 8.8   & 7.3   & 12.3  & 6.2   & 10.8  & 10.3 \\
    GPDMix($-$0.25, 0.25) & 13.9  & 6.7   & 16.6  & 14.8  & 26.2  & 7.9   & 33.0  & 32.4 \\
    GPD(1, 0.25) & 5.3   & 5.8   & 7.2   & 5.2   & 4.7   & 5.3   & 5.8   & 4.7 \\
    \bottomrule
    \end{tabular}
  \label{tab:gpd_powerstudy}
\end{table}

The rejection rates are summarized in Table~\ref{tab:gpd_powerstudy}. 
Samples in which the MLE failed were removed, 
which accounts for roughly 10.8\% of the Weibull samples with 
shape 1.25 and sample size 400, and around 10.7\% for the 
Gamma distribution with sample size 400. Decreasing the sample 
size in these cases actually decreases the percentage of failed 
MLE samples. This may be due to the shape of these two 
distributions, which progressively become more distinct from the 
GPD as their shape parameters increase. In the other distribution 
cases, no setting resulted in more than a 0.3\% failure rate.
As expected, all tests appear to hold their sizes, and 
their powers all increase with sample size.
The mixture of two GPDs is the hardest to detect. 
For the GPD mixture of shape parameters 0 and 0.4, quantile matching 
between a single large sample of generated data and the fitted GP 
distribution shows a high degree of similarity. 
In the vast majority of cases, the Anderson--Darling test appears to
have the highest power, followed by the Cram\'er--von Mises test.
Between the two, the Anderson--Darling statistic is a modification of
the Cram\'er--von Mises statistic giving more weight to observations
in the tail of the distribution, which explains the edge of the former.

\section{Simulation Study of the Automated Procedures}
\label{ch3:sim}

The empirical properties of the two stopping rules for the Anderson--Darling 
test are investigated in simulation studies. To check the empirical 
FWER of the StrongStop rule, data only need to be generated under 
the null hypothesis. For $n \in \{50, 100, 200, 400\}$, $\xi \in \{-0.25, 
0.25\}$, $\mu=0$, and $\sigma=1$, 10,000 GPD samples were generated in 
each setting of these parameters. Ten thresholds are tested by locating 
ten percentiles, 5 to 50 by 5. Since the data is generated 
from the GPD, data above each threshold is also distributed as GP, 
with an adjusted scale parameter. Using the StrongStop procedure and 
with no adjustment, the observed FWER is compared to the expected 
rates for each setting at various nominal levels. At a given 
nominal level and setting of the parameters, the observed FWER is 
calculated as the number of samples with a rejection of $H_0$ at any 
of the thresholds, divided by the total number of samples. The results 
of this study can been seen in Figure~\ref{fig:gpd_FWER_Check}.

\begin{figure}[!ht]
   \centering
      \includegraphics[width=\textwidth]{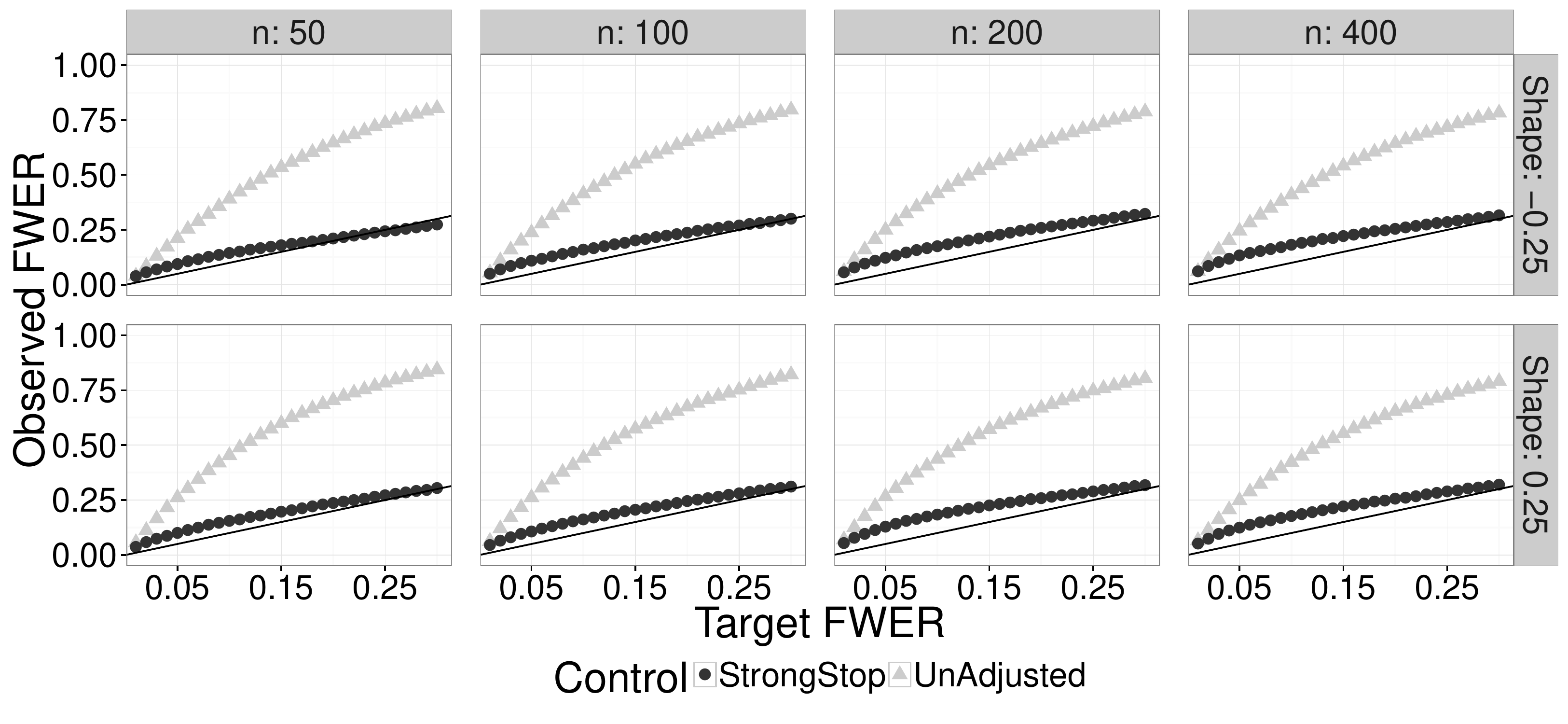}
    \caption{Observed FWER for the Anderson--Darling test (using StrongStop and 
    no adjustment) versus expected FWER at various nominal levels under the null 
    GPD at ten thresholds for 10,000 replicates in each setting as described in 
    Section~\ref{ch3:sim}. The 45 degree line indicates agreement between the 
    observed and expected rates under H0.}
    \label{fig:gpd_FWER_Check}
\end{figure}

It can be seen that the observed FWER under $H_0$ using the StrongStop 
procedure is nearly in agreement with the expected rate at most nominal 
levels for the Anderson--Darling test (observed rate is always 
within 5\% of the expected error rate). 
However, using the naive, unadjusted 
stopping rule (simply rejecting for any p-value less than the nominal 
level), the observed FWER is generally 2--3 times the expected rate 
for all sample sizes.

It is of interest to investigate the performance of the ForwardStop and 
StrongStop in selecting a threshold under misspecification. To check the 
ability of ForwardStop and StrongStop to control the false discover 
rate (FDR), data need to be generated under misspecification. Consider 
the situation where data is generated from a 50/50 mixture of 
distributions. Data between zero and five are generated from a Beta 
distribution with $\alpha = 2$, $\beta=1$ and scaled such that the 
support is on zero to five. Data above five is generated from the 
GPD($\mu=5$, $\sigma=2$, $\xi=0.25$). 
Choosing misspecification in this 
way ensures that the mixture distribution is at least continuous at 
the meeting point. See Figure~\ref{fig:gpd_MixtureDistribution} for 
a visual assessment.

\begin{figure*}[tbp]
    \centering
      \includegraphics[width=\textwidth]{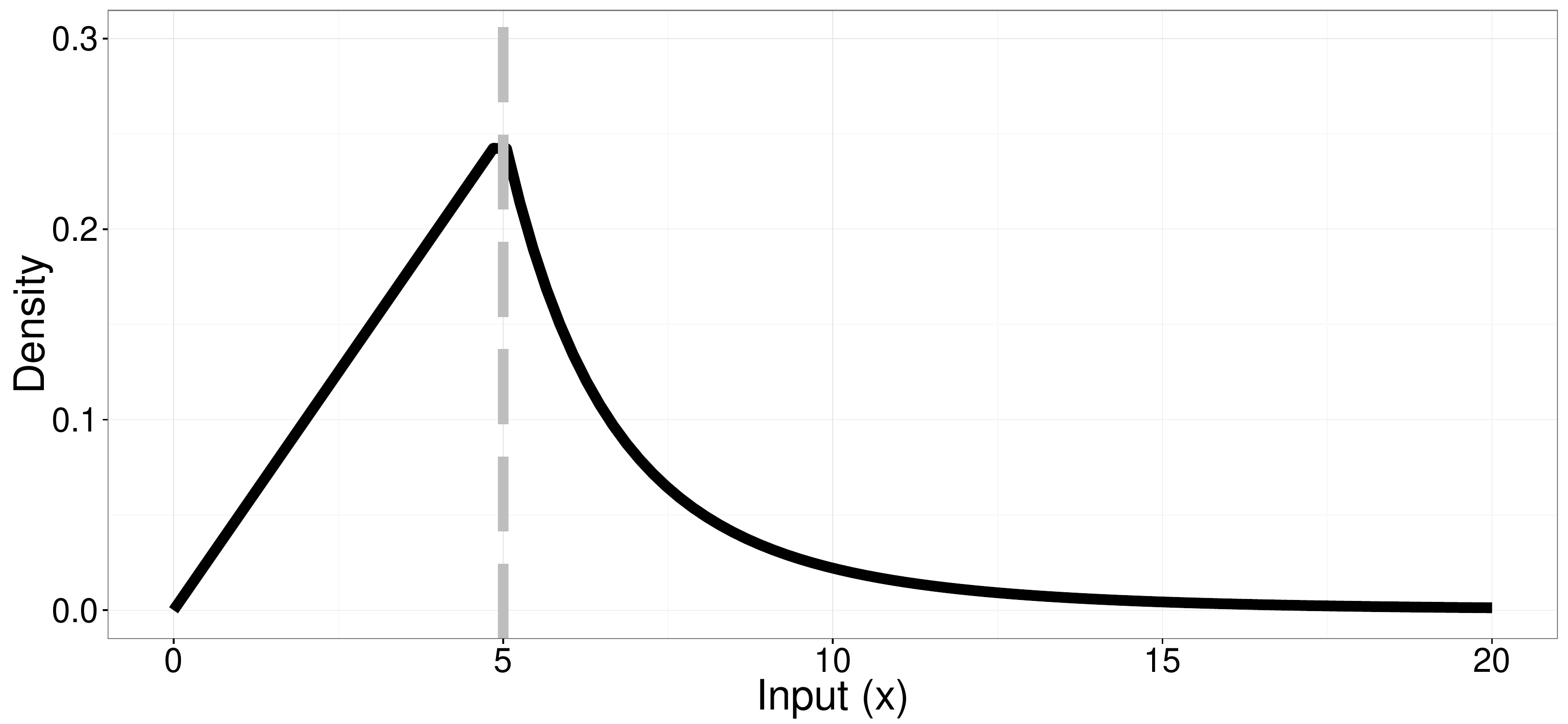}
    \caption{Plot of the (scaled) density of the mixture distribution used 
    to generate misspecification of $H_0$ for the simulation in Section~\ref{ch3:sim}. 
    The vertical line indicates the continuity point of the two underlying distributions.}
    \label{fig:gpd_MixtureDistribution}
\end{figure*}

A total of $n=1000$ observations are generated, with $n_1 = 500$ 
from the Beta distribution and $n_2 = 500$ from the GP distribution. 
1000 datasets are simulated and 50 thresholds are tested, starting 
by using all $n=1000$ observations and removing the 15 lowest 
observations each time until the last threshold uses just 250 
observations. In this way, the correct threshold is given when 
the lowest $n_1$ observations are removed.

The results are presented here. The correct threshold is the 34th, 
since $n_1 = 500$ and $\ceil{500 / 15} = 34$. Rejection of less than 
33 tests is problematic  since it allows contaminated data to 
be accepted. Thus, a conservative selection criteria is desirable. 
Rejection of more than 33 tests is okay, although some non-contaminated 
data is being thrown away. By its nature, the StrongStop procedure has 
less power-to-reject than ForwardStop, since it has a stricter 
error control (FWER versus FDR).

Figure~\ref{fig:Threshold_Sliced} displays the frequency distribution 
for threshold choice using the Anderson--Darling test at the 
5\% nominal level for the three stopping rules.
There is a clear hierarchy in terms of power-to-reject between 
the ForwardStop, StrongStop, and no adjustment procedures. 
StrongStop, as expected, provides the least power-to-reject, 
with all 1000 simulations selecting a threshold below the correct 
one. ForwardStop is more powerful than the no adjustment 
procedure and on average selects a higher threshold. The median 
number of thresholds rejected is 33, 29, and 22 for ForwardStop, 
no adjustment, and StrongStop respectively.

\begin{figure*}[tbp]
    \centering
      \includegraphics[width=\textwidth]{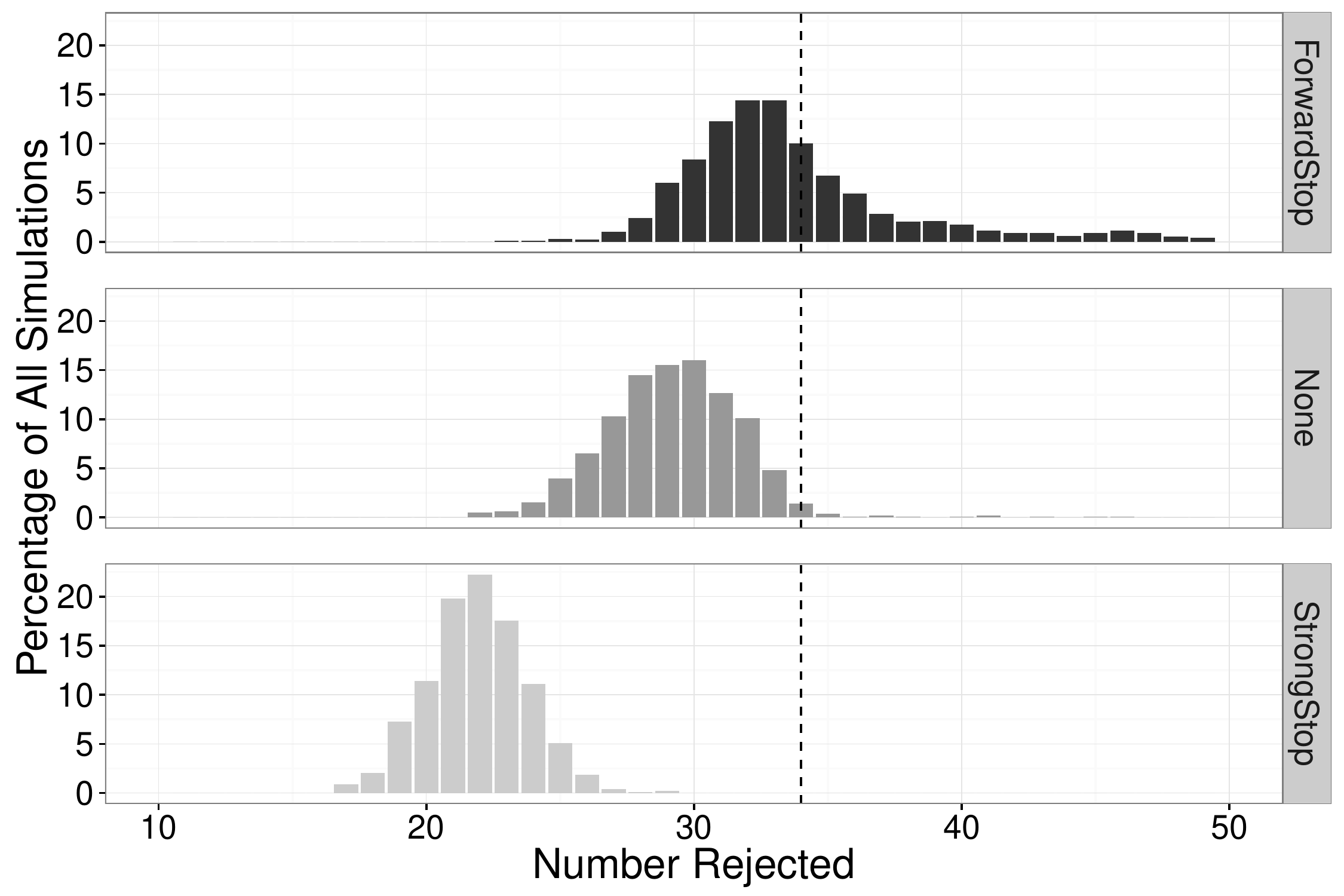}
    \caption{Frequency distribution (out of 1000 simulations) of the number of 
    rejections for the Anderson--Darling test and various stopping rules 
    (ForwardStop, StrongStop, and no adjustment), at the 5\% nominal level, for 
    the misspecified distribution sequential simulation setting 
    described in Section~\ref{ch3:sim}. 50 thresholds are tested, with 
    the 34th being the true threshold.}
    \label{fig:Threshold_Sliced}
\end{figure*}

The observed versus expected FDR using ForwardStop and the observed 
versus expected FWER using StrongStop using the Anderson--Darling test 
for the data generated under misspecification can be seen in 
Figure~\ref{fig:FWER_FDR_Combined}. There appears to be reasonable 
agreement between the expected and observed FDR rates using ForwardStop, 
while StrongStop has observed FWER rates well below the expected rates.

\begin{figure*}[tbp]
    \centering
      \includegraphics[scale = 0.4]{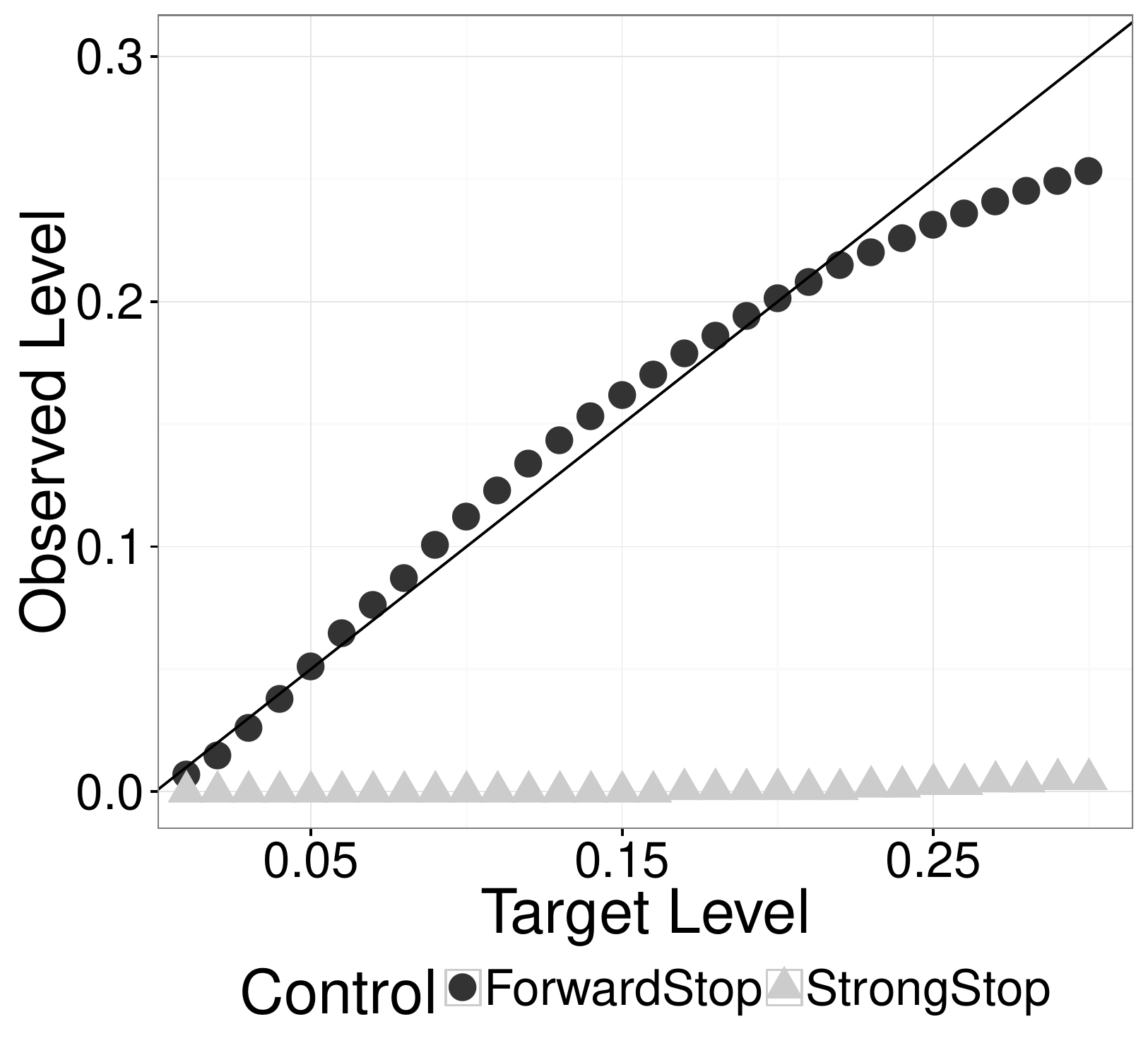}
    \caption{Observed FDR (using ForwardStop) and observed FWER 
    (using StrongStop) versus expected FDR and FWER respectively 
    using the Anderson--Darling test, at various nominal levels. This 
    is for the sequential simulation setting under misspecification 
    described in Section~\ref{ch3:sim}. The 45 degree line indicates agreement 
    between the observed and expected rates.}
    \label{fig:FWER_FDR_Combined}
\end{figure*}

To further evaluate the performance of the combination of tests 
and stopping rules, it is of interest to know how well the data 
chosen above each threshold can estimate parameters of interest. 
Two such parameters are the shape and return level. The $N$ year 
return level for the GPD~\citep[e.g.,][Section~4.3.3]{coles2001introduction} 
is given by 
\begin{equation}
z_N = 
\begin{cases} 
  u + \frac{\sigma}{\xi}[(N n_y \zeta_u)^\xi - 1], & \xi \neq 0, \\
  u + \sigma \log(N n_y \zeta_u), & \xi = 0,
\end{cases}
\end{equation}
for a given threshold $u$, where $n_y$ is the number of observations 
per year, and $\zeta_u$ is the rate, or proportion of the data 
exceeding $u$. The rate parameter has a natural estimator simply 
given by the number of exceeding observations divided by the total 
number of observations. A confidence interval for $z_N$ can easily 
be found using the delta method, or preferably and used here, profile 
likelihood. The `true' return levels are found by letting $u = 34$, 
$n_y = 365$, and treating the data as the tail of some larger dataset, 
which allows computation of the rate $\zeta_u$ for each threshold. 

For ease of presentation, focus will be on the performance of the 
Anderson--Darling test in conjunction with the three stopping 
rules. In each of the 1000 simulated datasets 
(seen in Figure~\ref{fig:gpd_MixtureDistribution}), the threshold 
selected for each of the three stopping rules is used to determine 
the bias, squared error, and confidence interval coverage 
(binary true/false) for the shape parameter and 50, 100, 250, and 
500 year return levels. An average of the bias, squared error, 
and coverage across the 1000 simulations is taken, 
for each stopping rule and parameter value. For each parameter and 
statistic of interest (mean bias, mean squared error (MSE), and mean coverage), 
a relative percentage is calculated for the three stopping rules. A 
visual assessment of this analysis is provided in 
Figure~\ref{fig:gpd_StoppingRuleCompare}.

\begin{figure*}[tbp]
    \centering
      \includegraphics[width=\textwidth]{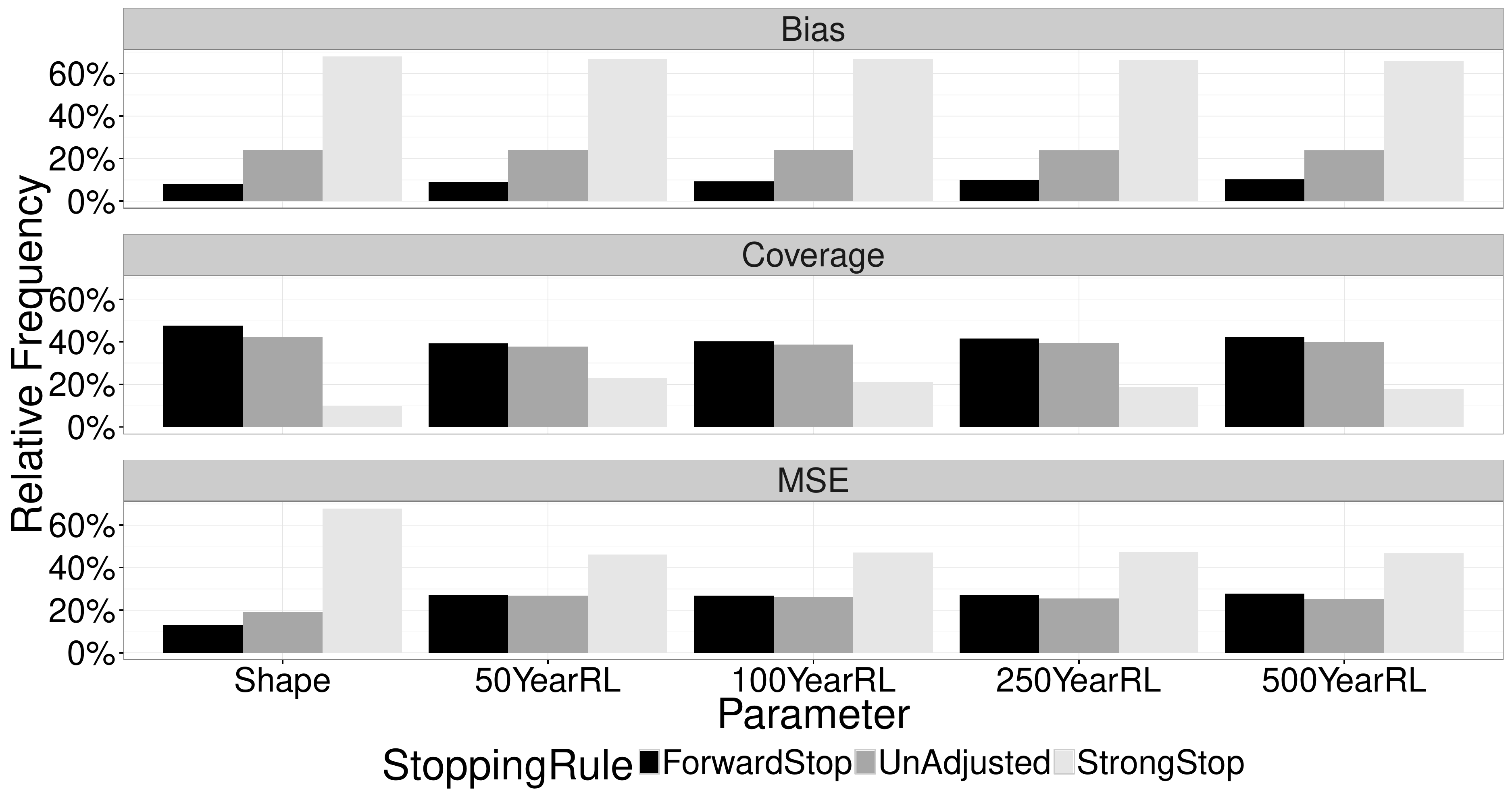}
    \caption{Average performance comparison of the three stopping rules 
    in the simulation study under misspecification in Section~\ref{ch3:sim}, 
    using the Anderson--Darling test for various parameters. Shown are the 
    relative frequencies of the average value of each metric (bias, squared 
    error, and coverage) for each stopping rule and parameter of interest. 
    For each parameter of interest and metric, the sum of average values 
    for the three stopping rules equates to 100\%. RL refers to return level.
    }
    \label{fig:gpd_StoppingRuleCompare}
\end{figure*}

It is clear from the result here that ForwardStop is the 
most preferable stopping rule in application. For all parameters, 
it has the smallest average bias (by a proportion of 2-1) and 
highest coverage rate. In addition, the MSE is comparable or 
smaller than the unadjusted procedure in all cases. Arguably, 
StrongStop has the worst performance, obtaining the highest 
average bias and MSE, and the lowest coverage rates for all 
parameters. Replacing the Anderson--Darling test with the other 
tests provides similar results (not shown here).

\section{Application to Return Level Mapping of Extreme Precipitation}
\label{ch3:app}

One particularly useful application of the automated threshold 
selection method is generating an accurate return level map of 
extreme precipitation in the three western US coastal states 
of California, Oregon, and Washington. The automated procedure 
described in Section~\ref{ch3:seq_testing} provides a method to 
quickly obtain an accurate map without the need for visual 
diagnostics at each site. Return level maps are of great 
interest to hydrologists~\citep{katz2002statistics} and 
can be used for risk management and land 
planning~\citep{blanchet2010mapping, lateltin1999hazard}.

Daily precipitation data is available for tens of thousands 
of surface sites around the world via the Global Historical 
Climatology Network (GHCN). A description of the data network 
can be found in~\citet{menne2012overview}. After an initial 
screening to remove sites with less than 50 years of 
available data, there are 720 remaining sites across the 
three chosen coastal states.

As the annual maximum daily amount of precipitation mainly 
occurs in winter, only the winter season (November - March) 
observations are used in modeling. A set of thresholds for each 
site are chosen based on the data percentiles; for each site the 
set of thresholds is generated by taking the 75th to 97th 
percentiles in increments of 2, then in increments of 0.1 
from the 97th to 99.5th percentile. This results in 37 
thresholds to test at each site. If consecutive percentiles 
result in the same threshold due to ties in data, 
only one is used to guarantee uniqueness of thresholds and 
thus reducing the total number of thresholds tested at that site.

As discussed in the beginning of Section~\ref{ch3:seq_testing}, 
modeling the exceedances above a threshold with the Generalized 
Pareto requires the exceedances to be independent. This is 
not always the case with daily precipitation data; a large 
and persistent synoptic system can bring large storms 
successively. The extremal index is a measure of the 
clustering of the underlying process at extreme levels. Roughly 
speaking, it is equal to (limiting mean cluster size)$^{-1}$. 
It can take values from 0 to 1, with independent series 
exhibiting a value of exactly 1. To get a sense for the 
properties of series in this dataset, the extremal index 
is calculated for each site using the 75th percentile as the 
threshold via the R package \texttt{texmex}~\citep{Southworth2013}. 
In summary, the median estimated extremal index for all 720 
sites is 0.9905 and 97\% of the sites have an extremal index 
above 0.9. Thus, we do not do any declustering on the data. 
A more thorough description of this process can be found in 
\citet{Ferro2003} and \citet{heffernan2012extreme}.

The Anderson--Darling test  is used to test the set of thresholds 
at each of the 720 sites, following the procedure outlined in 
Section~\ref{ch3:ad_cvm}. This is arguably the most powerful test 
out of the four examined in Section~\ref{ch3:power}. Three 
stopping rules are used --- ForwardStop, StrongStop, and with no 
adjustment, which proceeds in an ascending manner until an acceptance 
occurs. Figure~\ref{fig:ChosenThresholds} shows the distribution 
of chosen percentiles for the 720 sites using each of the three 
stopping rules.

\begin{figure*}[tbp]
    \centering
      \includegraphics[width=\textwidth]{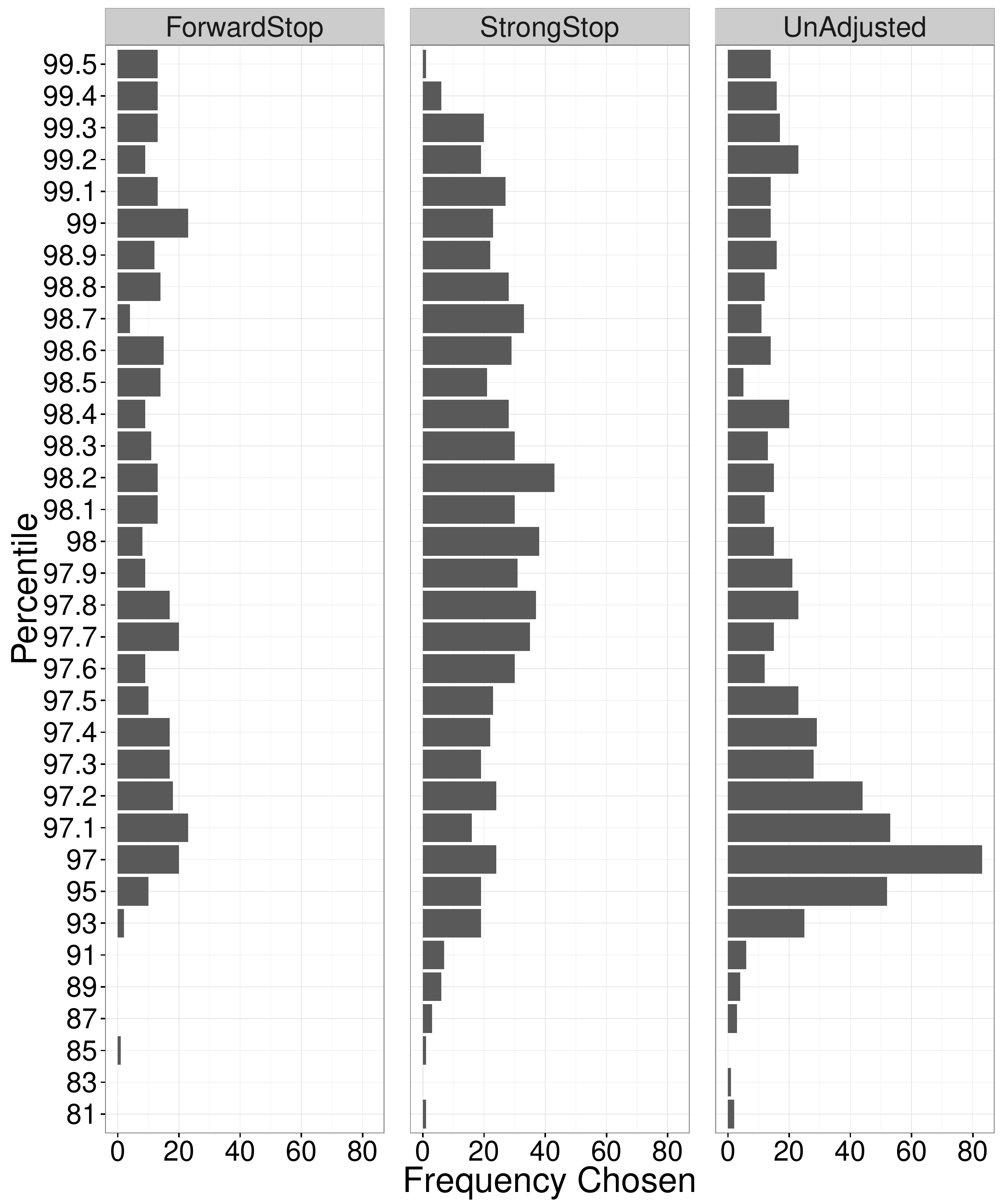}
    \caption{Distribution of chosen percentiles (thresholds) for the 720 
    western US coastal sites, as selected by each stopping rule. Note that 
    this does not include sites where all thresholds were rejected by the 
    stopping rule.}
    \label{fig:ChosenThresholds}
\end{figure*}

As expected, ForwardStop is the most conservative, rejecting all 
thresholds at 348 sites, with the unadjusted procedure rejecting 63, 
and StrongStop only rejecting all thresholds at 3 sites. 
Figure~\ref{fig:All_Thresholds_Rejected} shows the geographic 
representation of sites in which all thresholds are rejected. 
Note that there is a pattern of rejections by ForwardStop, 
particularly in the eastern portion of Washington and Oregon,
and the Great Valley of California. This may be attributed to 
the climate differences in these regions -- rejected sites had 
a smaller number of average rain days than non-rejected sites 
(30 vs. 34), as well as a smaller magnitude of precipitation (an 
average of 0.62cm vs. 1.29cm). The highly selective feature of 
the ForwardStop procedure is desired as it suggests not to fit 
GPD at even the highest threshold at these sites, a guard that 
is not available from those unconditionally applied, one-for-all 
rules. Further investigation is needed for the sites in which 
all thresholds were rejected; one possibility is to consider 
a generalization of the GPD in these cases.

\begin{figure}[!ht]
   \centering
      \includegraphics[width=\textwidth]{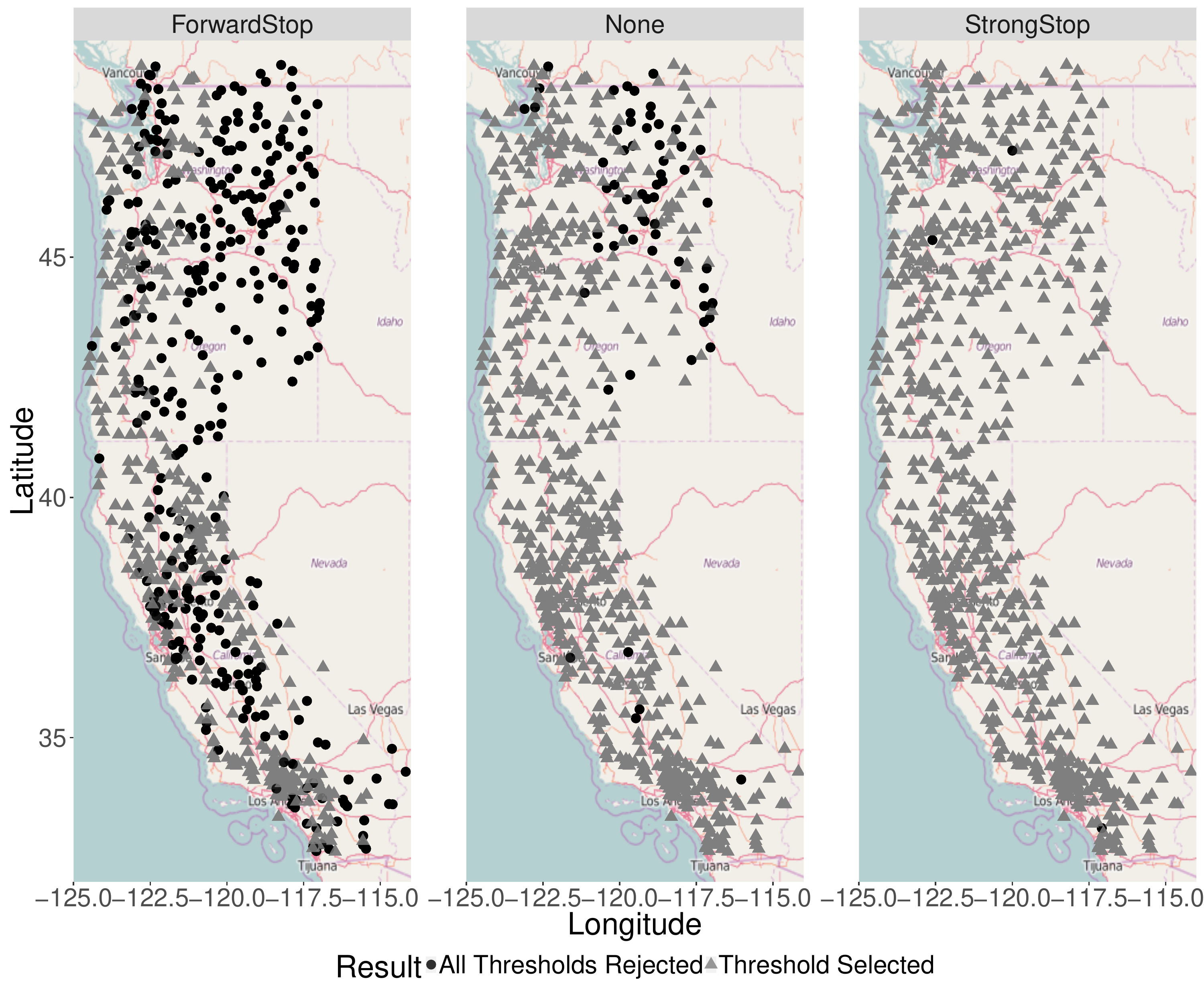}
    \caption{Map of US west coast sites for which all thresholds were 
    rejected (black / circle) and for which a threshold was selected 
    (grey / triangle), by stopping rule.}
    \label{fig:All_Thresholds_Rejected}
\end{figure}

For the sites at which a threshold was selected for both the 
ForwardStop and StrongStop rules, the return level estimates 
based on the chosen threshold for each stopping rule can be 
compared. The result of this comparison can be seen in 
Figure~\ref{fig:ForwardStop_vs_StrongStop}. For a smaller 
return period (50 years), the agreement between estimates 
for the two stopping rules is quite high. This is a nice 
confirmation to have confidence in the analysis, however 
it is slightly misleading in that it does not contain the 
sites rejected by ForwardStop.

\begin{figure}[!ht]
   \centering
      \includegraphics[width=\textwidth]{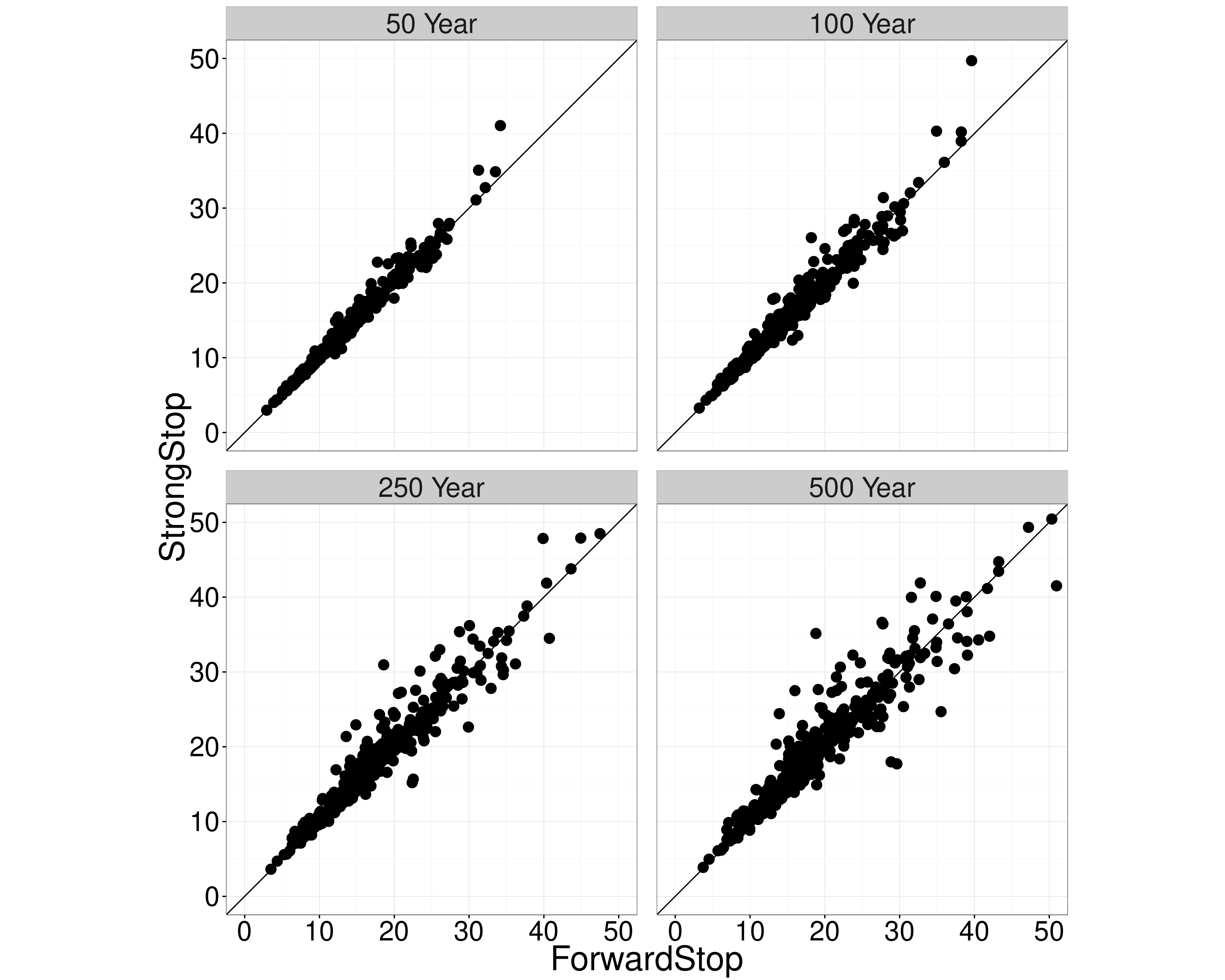}
    \caption{Comparison of return level estimates (50, 100, 250, 500 year) 
    based on the chosen threshold for ForwardStop vs. StrongStop for the US 
    west coast sites. The 45 degree line indicates agreement between the 
    two estimates. This is only for the sites in which both stopping rules 
    did not reject all thresholds.}
    \label{fig:ForwardStop_vs_StrongStop}
\end{figure}

The end result of this automated batch analysis provides a map 
of return level estimates. The 50, 100, and 250 year map 
of return level estimates from threshold selection using 
ForwardStop and the Anderson--Darling test can be seen in 
Figure~\ref{fig:ForwardStopRL}. To provide an estimate at 
any particular location in this area, some form of 
interpolation can be applied.

\begin{figure}[!ht]
   \centering
      \includegraphics[width=\textwidth]{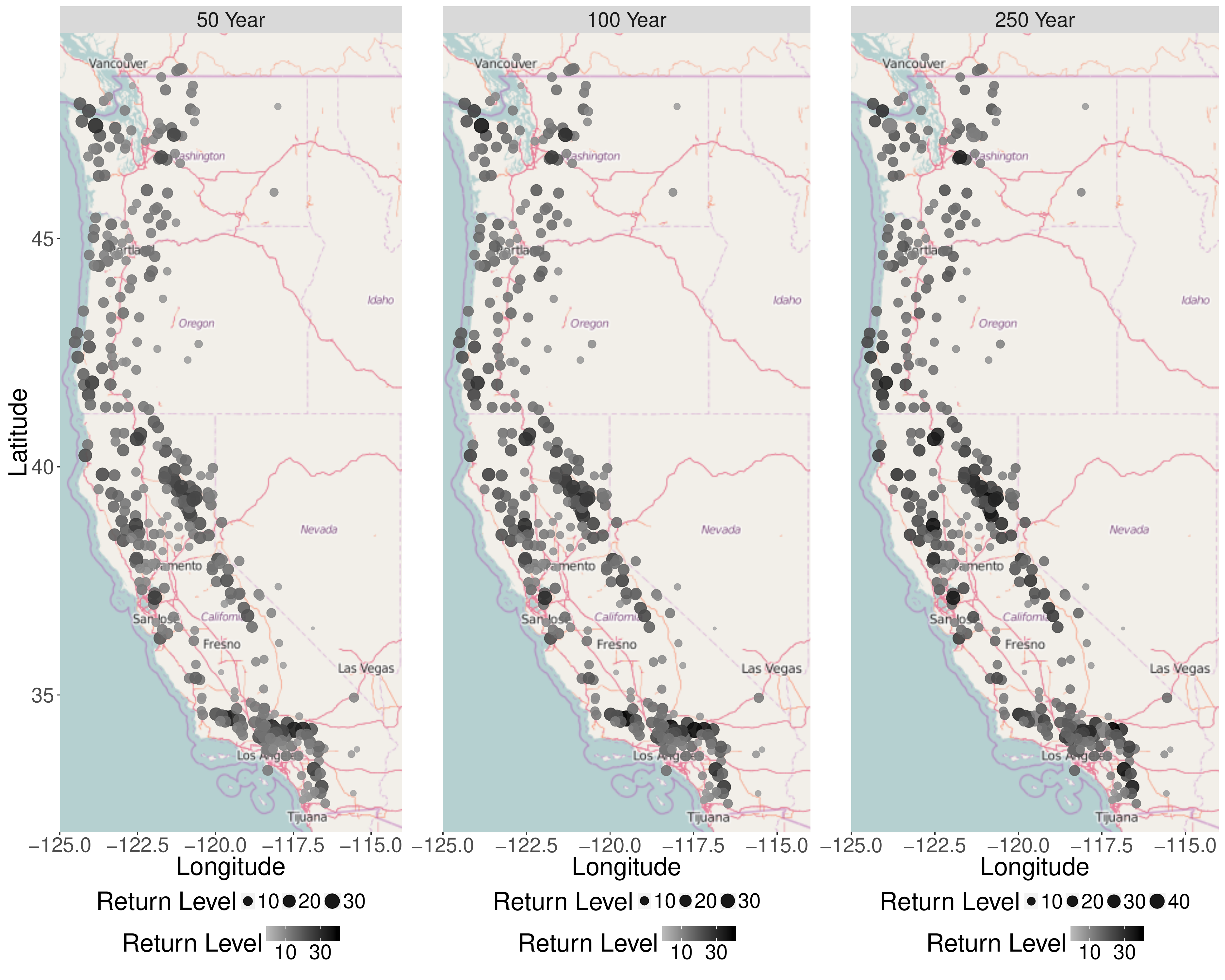}
    \caption{Map of US west coast sites with 50, 100, and 250 year 
    return level estimates for the threshold chosen using ForwardStop 
    and the Anderson--Darling test. This is only for the sites in which 
    a threshold was selected.}
    \label{fig:ForwardStopRL}
\end{figure}

\section{Discussion}
\label{ch3:disc}

We propose an intuitive and comprehensive methodology for automated 
threshold selection in the peaks over threshold approach. In addition, 
it is not as computationally intensive as some competing resampling 
or bootstrap procedures~\citep{danielsson2001using, ferreira2003optimising}. 
Automation and efficiency is required when prediction using the 
peaks over threshold approach is desired at a large number of 
sites. This is achieved through sequentially testing a set of 
thresholds for goodness-of-fit to the generalized Pareto 
distribution (GPD). Previous instances of this sequential testing 
procedure did not account for the multiple testing issue. 
We apply two recently developed stopping rules~\citep{g2015sequential} 
ForwardStop and StrongStop, that control the false discovery 
rate and familywise error rate, respectively in the setting of 
(independent) ordered, sequential testing. 
It is a novel application of them to the threshold selection problem.
There is a slight caveat in our setting, that the tests are not
independent, but it is shown via simulation that these stopping rules
still provide reasonable error control here.

Four tests are compared in terms of power to detect departures from 
the GPD at a single, fixed threshold and it is found that the 
Anderson--Darling test has the most power in various non-null 
settings. \citet{choulakian2001goodness} derived the asymptotic 
null distribution of the Anderson--Darling test statistic. 
However this requires solving an integral equation. Our 
contribution, with some advice from the author, provides an 
approximate, but accurate and computationally efficient version 
of this test. To investigate the performance of the stopping rules 
in conjunction with the Anderson--Darling test, a large scale 
simulation study was conducted. Data is generated from a plausible 
distribution -- misspecified below a certain threshold and generated 
from the null GPD above. In each replicate, the bias, coverage, and 
squared error is recorded for the stopping threshold of each stopping 
rule for various parameters.

The results of this simulation suggest that the ForwardStop 
procedure has on average the best performance using the 
aforementioned metrics. Thus, a recommendation would be to 
use the Anderson--Darling test, in unison with the ForwardStop 
procedure for error control to test a fixed set of thresholds.
The methodology is applied to daily precipitation data 
at hundreds of sites in three U.S. west coast states, with 
the goal of creating a return level map.

\chapter{Robust and Efficient Estimation in Non-Stationary RFA}
\label{ch:rfa}

\section{Introduction}
\label{ch4:intro}

Regional frequency analysis (RFA) is commonly used when historical 
records are short, but are available at multiple sites within a 
homogeneous region. Sites within a homogenous region are assumed 
to possess similar characteristics (defined in Section~\ref{ch4:model}). 
From this assumption, with some care, 
the data can essentially be pooled in order to estimate the shared 
parameters across sites. By doing so, efficiency in the shared 
parameter estimation is greatly improved when record length is short. 
Various models can be used within this context, depending on the 
application and outcome variable. One area of application is to 
model extremes, such as maximum precipitation or wind speeds. 
For this, the Generalized Extreme Value (GEV) distribution is a 
natural model choice and will be focused on in this paper.

There have been many applications of RFA with non-stationarity modeled 
through external covariates; see~\cite{khaliq2006frequency} for a review 
in the one sample case. The most common estimation method is maximum 
likelihood estimation (MLE), typically with climate indices and temporal 
trends as covariates; see~\citep{katz2002statistics, lopez2013non, 
hanel2009nonstationary, leclerc2007non, nadarajah2005extremes} 
as examples. The MLE has some nice properties, in that the model parameters 
and covariates can easily be adjusted, and standard errors are provided 
implicitly when maximizing the likelihood function. There are some drawbacks 
however; the MLE can provide absurd estimates of the shape parameter or fail 
entirely, typically in small samples.

\cite{coles1999likelihood} develop a penalized maximum likelihood 
estimator that still has the flexibility of the traditional MLE, but has 
better performance in small samples. \cite{martins2000generalized} propose 
a generalized maximum likelihood (GML) method which imposes a prior distribution 
on the shape parameter to restrict its estimates to plausible values within 
the GEV model. \cite{el2007generalized} applied the GML approach with 
non-stationarity, via temporal trends in the parameters. Note that both 
of these approaches are only explicitly developed here for the one sample 
case. At the time of writing, it appears that neither of these approaches 
have been applied to the RFA setting.

In the stationary RFA model, L-moments can be used as an alternative to MLE. 
The L-moment approach~\citep{hosking1985estimation} estimates the parameters 
by matching the sample moments with their population counterparts and then 
solving the system of equations. It has the advantage over MLE in that it 
only requires the existence of the mean, and has been shown to be more 
efficient in small samples~\citep{hosking1985estimation}. \cite{hosking1988effect} 
show that any bias in the parameter estimates is unaffected by intersite 
dependence. To apply the L-moment method in the stationary RFA model, 
the site-specific scaling parameter is estimated, then the data is scaled 
and pooled across sites. Afterwards, the shared, across site parameters can be 
estimated from the pooled data. This method is widespread; 
see~\cite{smithers2001methodology, kumar2005regional, kjeldsen2002regional} 
for examples.

It is not straightforward to extend this methodology to the non-stationary 
case; in particular when non-site specific covariates are used, moment 
matching may not be possible. One approach to estimate time trends is 
by applying the stationary L-moment approach over sliding time 
windows~\citep{kharin2005estimating, kharin2013changes}; that is, estimate 
the stationary parameters in $t$-year periods and look for changes between 
each. In the one sample case, there has been some progress to combine 
non-stationarity and L-moment estimation. \cite{ribereau2008estimating} 
provide a method to incorporate covariates in the location parameter, 
by estimating the covariates first via least squares, and then 
transforming the data to be stationary in order to estimate the 
remaining parameters via L-moments. \cite{coles1999likelihood} briefly 
discuss an iterative procedure to estimate covariates through maximum 
likelihood and stationary parameters through L-moments.

Two new approaches are implemented here for estimation in the 
RFA setting with non-stationarity. The first uses maximum product 
spacing (MPS) estimation~\citep{cheng1983estimating}, which maximizes 
the geometric mean of spacings in the data. It is as efficient as maximum 
likelihood estimation and since it only relies on the cumulative density 
function, it can provide estimators when MLE fails. It can easily incorporate 
covariates in the objective function as well. The second approach builds 
on the ideas of~\cite{coles1999likelihood} and~\cite{ribereau2008estimating}, 
providing extensions to a hybrid estimation procedure (L-moment / MLE) 
in the RFA framework. Both of these methods are shown to outperform MLE 
in terms of root mean squared error (RMSE), through a large scale simulation 
of spatially dependent data. To illustrate the differences in the three 
estimation methods, a non-stationary regional frequency model is applied 
to extreme precipitation events at rain gauge sites in California. To 
account for spatial dependence, a semi-parametric bootstrap procedure is 
applied to obtain standard errors for all the estimators.

Section~\ref{ch4:model} discusses the model and existing methods 
of estimation within this framework. Section~\ref{ch4:new_methods} 
introduces the two new estimation methods and implementation in the 
RFA setting. A large scale simulation study of homogeneous multi-site 
data generated under various extremal and non-extremal spatial 
dependence settings is carried out in Section~\ref{ch4:sim} to establish 
the superiority of the new estimation methods over MLE. In 
Section~\ref{ch4:data} a real analysis of extreme precipitation 
at 27 sites in California is conducted to examine the effect of 
El Ni\~{n}o--Southern Oscillation on these events. Lastly, a 
discussion is provided in Section~\ref{ch4:disc}.

\section{Non-Stationary Homogeneous Region Model}
\label{ch4:model}

In regional frequency analysis, data from a number of sites 
within a homogeneous region can be pooled together in a particular 
way to improve the reliability and efficiency of parameter 
estimates within a model. Suppose there exist $m$ sites in 
a homogeneous region, over $n$ periods. Define $Y_{st}$ to 
be the observation at site $s \in \{1, \ldots, m\}$ in period 
$t \in \{1, \ldots, n\}$. Observations from period to period 
are assumed to be independent, but observations 
from sites within the same period exhibit some form of 
spatial dependence. The type of spatial dependence 
will be defined explicitly later. The assumption is that 
each site has the full record of observations. Methodology 
to accommodate missing records is not straightforward -- 
imputation would require additional assumptions and 
in the case of data missing completely at random (MCAR), 
care is still needed due to the semi-parametric bootstrap 
procedure necessary to obtain parameter confidence intervals. 
A further discussion is provided in Section~\ref{ch6:future}.

The idea behind the flood index model is that the variables 
$Y_{st}$, after being scaled by a site specific scaling 
factor, are identically distributed. In other words, the 
$T$-period return level or quantile function of $Y_{st}$ 
can be represented as
\begin{equation*}
Q_t(T) = g(\mu_{st} , q_t(T))
\end{equation*}
where $q_t(T)$ is the non-site specific quantile function 
and the flood index is $\mu_{st}$. Typically, this is chosen 
as the mean or median of $Y_{st}$. In this paper, the Generalized Extreme 
Value (GEV) distribution~\eqref{eq:gev_cdf} is used to 
model the extreme events at each site. It is essentially the limiting 
distribution of the maximum of a sequence of independent and 
identically distributed random variables, normalized by appropriate 
constants. Because of this mathematical formulation, it is a 
natural distribution choice to model extremes.

Formally, let $Y_{st} \sim$ GEV($\mu_{st}, \gamma_t \mu_{st}, \xi_t$) 
where here the flood index is given by $\mu_{st}$. It follows that 
$Y_{st} = \mu_{st} Z_t$ where $Z_t \sim$ GEV($1, \gamma_t, \xi_t$). 
It is then clear that each site in the homogeneous region marginally 
follows a GEV distribution indexed only in terms of its location parameter, 
with a shared proportionality scale parameter and shape parameter.

To incorporate the non-stationarity, link 
functions between the covariates and parameters can be used 
as follows. Let 
\begin{align*}
\mu_{st} &= f_{\mu} \Big( \beta^{\mu}, x^{\mu}_{st} \Big) \\
\gamma_t &= f_{\gamma} \Big( \beta^{\gamma}, x^{\gamma}_{t} \Big) \\
\xi_t &= f_{\xi} \Big( \beta^{\xi}, x^{\xi}_{t} \Big)
\end{align*}
where $f(\cdot)$ is the parameter specific link function. 
The set of parameters $\beta^{\mu}, \beta^{\gamma}, \beta^{\xi}$ 
and covariates $x^{\mu}_{st}, x^{\gamma}_{t}, x^{\xi}_{t}$ are 
vectors of length $(m + p_{\mu}), p_{\gamma}, p_{\xi}$, respectively, 
representing the number of covariates associated with each. The 
indexing in the vector of location parameters allows for each site 
to have its own marginal mean $\beta_s^{\mu}$ for $s = 1, \ldots, m$. 
Covariates associated with $\gamma$ and $\xi$ must be shared across 
sites, while covariates associated with $\mu$ can be both shared and/or 
site-specific. ${\bf X}^{\mu}_{(m \cdot n) \times (m + p_{\mu})}, {\bf X}^{\gamma}_{n \times p_{\gamma}}, {\bf X}^{\xi}_{n \times p_{\xi}}$ 
are design matrices for the appropriate parameters, where each row 
represents the corresponding vector $x$.

To fix ideas, the following link functions and covariates will be 
used in the succeeding sections unless explicitly stated otherwise:
\begin{align}
\label{eq:links}
\mu_{st} &=  {x^{\mu}_{st}}^\mathsf{T} \beta^{\mu} \nonumber \\
\gamma_t &= \exp\Big({x^{\gamma}_{t}}^\mathsf{T} \beta^{\gamma}\Big) \nonumber \\
\xi_t &= {x^{\xi}_{t}}^\mathsf{T} \beta^{\xi} \nonumber \\
\end{align}
which implies $\sigma_{st} = \gamma_t \mu_{st}$. Also, note that 
the stationary (intercept) parameters are given by $\beta^{\gamma}_0$, 
$\beta^{\xi}_0$, which are shared across all sites and $\beta_s^{\mu}$, 
a stationary marginal mean for each site $s = 1, \ldots, m$. This is 
similar to the model used in~\cite{hanel2009nonstationary} 
and the exponential link in the proportionality parameter $\gamma$ 
assures the scale parameter at each site is positive valued since 
$\mu_{st}$ is generally positive in application. An explicit 
restriction on the positiveness of the scale parameter can be 
enforced in the estimation procedure. In the stationary case, 
the subscript $t$ is dropped, but the same model applies. This 
stationary model has been used previously by \cite{buishand1991extreme} 
and \cite{wang2014incorporating}.

\subsection{Existing Estimation Methods}
\label{ch4:existing_methods}

\subsubsection{Maximum Likelihood}
\label{ch4:ml}

The independence likelihood~\citep[e.g.][]{hanel2009nonstationary,kharin2005estimating} 
seeks to maximize the log-likelihood function
\begin{equation*}
\sum\limits_{s=1}^{m} \sum\limits_{t=1}^{n} \log f(y_{st} | \mu_{st}, 
\gamma_t \mu_{st}, \xi_t, x^{\mu}_{st}, x^{\gamma}_{t}, x^{\xi}_{t})
\end{equation*}
where $f(\cdot)$ is the probability density function~\eqref{eq:gev_pdf} 
of site $s$ at time $t$. Note that estimation in this approach assumes 
independence between sites and only specifies the marginal density 
at each site. Other likelihood approaches such as a pairwise likelihood 
can be used, but then the dependence between sites must be specified. 
See~\cite{ribatet2009user} for details.

\subsubsection{Pseudo Non-Stationary L-moments}
\label{ch4:non_stat_lmom}

To the best of the author's knowledge, currently there are no 
methods that directly incorporate L-moment estimation in the 
non-stationary RFA setting. \cite{kharin2013changes} uses the 
stationary method in conjunction with moving time windows to get 
a sense of the trend in the parameters. That is, the data is subset 
into mutually exclusive time windows (say 20 years) and the stationary 
L-moment method is used to obtain GEV parameter estimates. This 
can provide some indication of non-stationarity, however it is 
difficult to interpret precisely.

\cite{ribereau2008estimating} provides methodology for 
mixture L-moment estimation in the one sample case 
with non-stationarity modeled through the location 
parameter. The idea is to fit a GEV regression model, where the 
error terms $\epsilon$ are assumed to be IID stationary 
GEV($0, \sigma, \xi$) random variables and the design matrix 
consists of covariates that describe the location parameter. 
From this regression setup, 
estimates for the location parameters can be found. Next, those 
estimates are used to transform the data to stationary 
``psuedo-residuals'' from which estimates for $\sigma$ and 
$\xi$ can be obtained by a variant of L-moments.

It should be possible to extend this setup to the RFA setting 
by fitting a GEV regression model to data from all the sites. 
To handle this, the error structure can be given, for example, 
as
\begin{align*}
Cov(\epsilon_{si}, \epsilon_{si}) &= \delta_s  & & i  \in \{1, \ldots, n\} & & s \in \{1, \ldots, m\} \\
Cov(\epsilon_{s_1i}, \epsilon_{s_2j}) &= 0  & & i \neq j \in \{1, \ldots, n\} & & s_1, s_2 \in \{1, \ldots, m\} \\
Cov(\epsilon_{s_1i}, \epsilon_{s_2i}) &= \rho_{s_1s_2}  & & i  \in \{1, \ldots, n\} & & s_1 \neq s_2 \in \{1, \ldots, m\}
\end{align*}
and then estimates for the non-stationary location parameters 
can be found using suitable methods (e.g. weighted least 
squares or generalized estimating equations). These ideas 
were not pursued here, but it may be of interest to look 
into in the future.

\section{New Methods}
\label{ch4:new_methods}

\subsection{Hybrid Likelihood / L-moment Approach}
\label{ch4:hybrid_lik}

Our approach to handle non-stationarity in a regional frequency 
model using the L-moment method involves an iterative hybrid procedure. 
\cite{coles1999likelihood} briefly discussed this approach 
in the one-sample case with non-stationarity in the location 
parameter. The procedure is as follows.

\begin{enumerate}
\item[]
Initialize with some starting estimates for all the parameters 
$\beta^{\mu}, \beta^{\gamma}, \beta^{\xi}$.

\item
Transform the variables $Y_{st}$ (via a probability integral 
transformation) such that each site has a (site-specific) stationary 
mean and shared (across sites) stationary proportionality and shape 
parameters.\footnote{This can be done in R using functions gev2frech and frech2gev 
in package \texttt{SpatialExtremes}~\citep{ribatet2011spatialextremes}.} 
Call this transformed variable $\tilde{Y}_{st}$ and note that $\tilde{Y}_{st} \sim$ 
GEV($\beta^{\mu}_s, \beta^{\mu}_s \exp(\beta^{\gamma}_0), \beta^{\xi}_0$).

\item
Estimate $\beta^{\mu}_s$ from the transformed data $\tilde{Y}_{st}$, 
$t = 1, \ldots, n$ for each site $s$ using the one sample 
stationary L-moment method. Let 
$\tilde{Y}'_{st} = \tilde{Y}_{st} / \hat{\beta}^{\mu}_s$ and 
note that $\tilde{Y}'_{st} \sim$ 
GEV($1, \exp(\beta^{\gamma}_0), \beta^{\xi}_0$).

\item
Pool the $\tilde{Y}'$ data from all sites and estimate 
($\beta^{\gamma}_0, \beta^{\xi}_0$) by setting the first two 
sample L-moments ($I_1, I_2$) of the pooled data to match 
its population counterparts. This yields the set of estimating 
equations:
\begin{align*}
& I_1 = 1 - \frac{\exp(\beta^{\gamma}_0) [1 - \Gamma(1 - \beta^{\xi}_0)]}{\beta^{\xi}_0} & I_2 = - \frac{\exp(\beta^{\gamma}_0) (1 - 2^{\beta^{\xi}_0}) \Gamma(1 - \beta^{\xi}_0)}{\beta^{\xi}_0}
\end{align*}
where $\beta^{\xi}_0 < 1$.
The solution\footnote{An additional, explicit restriction is imposed 
such that $\beta^{\xi}_0 > -1$ to ensure the log-likelihood in 
step 4 does not become irregular.} to the set of equations yields the 
L-moment estimates for ($\beta^{\gamma}_0, \beta^{\xi}_0$).

\item 
Maximize the log-likelihood
\begin{equation*}
\sum\limits_{s=1}^{m} \sum\limits_{t=1}^{n} \log f(y_{st} | \beta^{\mu}, 
\beta^{\gamma}, \beta^{\xi}, x^{\mu}_{st}, x^{\gamma}_{t}, x^{\xi}_{t})
\end{equation*}
with respect to only the non-stationary parameters to obtain their 
estimates.

\item
Go back to step 1 using the estimates obtained in steps 2--4, 
unless convergence is reached (i.e. the change in maximized 
log-likelihood value is within tolerance between iterations).

\end{enumerate}

While this procedure still requires a likelihood function, it 
generally provides better estimates of the non-stationary parameters 
than a regression approach~\citep{coles1999likelihood}. Additionally, 
the regression approach is restricted to incorporating non-stationarity 
in the location parameter. The hybrid procedure only needs to maximize 
the likelihood with respect to the non-stationary parameters, thus 
reducing the complexity of the optimization routine.

\subsection{Maximum Product Spacing}
\label{ch4:mps}

\subsubsection{One Sample Case}
\label{ch4:mps_onesample}

Maximum product spacing (MPS) is an alternative estimation procedure 
to the MLE and L-moment methods. The MPS was established 
by~\cite{cheng1983estimating}. It allows efficient 
estimation in non-regular cases where the MLE may not exist. This is 
especially relevant to the GEV distribution, in which the MLE does not 
exist when $\xi < -1$~\citep{smith1985maximum}. In the GEV distribution 
setting, the MPS also exhibits convergence at the same rate as MLE, 
and its estimators are asymptotically normal~\citep{wong2006note}. 
While the L-moment approach has been shown to be efficient in small 
samples, it inherently restricts $\xi < 1$~\citep{hosking1985estimation}, 
and asymptotic confidence intervals are not available for $\xi > 0.5$.

The general setup for MPS estimation in the one sample case is given 
in Section~\ref{ch3:moran}. The advantage of MPS over MLE can be 
seen~\eqref{eq:gpd_morans} directly in the form of the objective function 
$M(\theta)$. By using only the GEV cumulative density function, $M(\theta)$ 
does not collapse for $\xi < -1$ as $y \downarrow \mu - \frac{\sigma}{\xi}$.

\citet[][Section 4]{wong2006note} perform extensive simulations 
comparing the MPS, MLE, and L-moment estimators in the one 
sample case. They find that the MPS is generally more stable 
and comparable to the MLE, and does not suffer from bias as does 
the L-moment in certain cases. In addition, they find that MLE 
failure rates are generally higher than those of MPS.

\subsubsection{RFA Setting}
\label{ch4:mps_rfa}

Let $Y_{st}$ be the observation in year $t \in \{1, \ldots, n\}$ from 
site $s \in \{1, \ldots, m\}$. It has been assumed that $Y_{st} \sim$ 
GEV($\mu_{st}, \gamma_t \mu_{st}, \xi_t$). Via the transformation, 
discussed in detail in~\cite{coles2001introduction},
\begin{equation*}
Z_{st} = \Big[ 1 + \xi_t \Big(\frac{Y_{st} - \mu_{st}}{\gamma_t \mu_{st}} \Big) \Big]^{1 / \xi_t}
\end{equation*}
it follows that $Z_{st}$ is distributed according to the standard 
unit Fr\'{e}chet distribution.

Let $F(z)$ be the cdf of $Z$, and the MPS estimation proceeds as 
follows. Order the observations within each site $s$ as 
$z_{s(1)} < \ldots < z_{s(n)}$ and define the 
spacings as before:
\begin{equation*}
D_{si}(\theta) = F(z_{s(i)}) - F(z_{s(i-1)})
\end{equation*}
for $i=1, 2, \ldots, n+1$ where $F(z_{s(0)}) \equiv 0$ 
and $F(z_{s(n+1)}) \equiv 1$. The MPS estimators are 
then found by minimizing
\begin{equation}
\label{eq:moran2}
M(\theta) = - \sum\limits_{s=1}^{m} \sum\limits_{i=1}^{n+1} \log D_{si}(\theta).
\end{equation}
Denote the MPS estimates as 
$\{ \check{\beta}^{\mu}, \check{\beta}^{\gamma}, \check{\beta}^{\xi} \}$. 
MPS estimation assumes the set of observations are 
independent and identically distributed. Here, observations 
can be assumed to be identically distributed, although the 
transformed observations $z$ may still exhibit spatial 
dependence. However, point estimation can still be carried 
out with semi-parametric bootstrap procedures ensuring 
adequate adjustment to the variance of the estimators.

\section{Simulation Study}
\label{ch4:sim}

A large-scale simulation study is conducted to assess the 
performance of the estimation procedures for data generated 
under various types of spatial dependence. Data from 
the GEV distribution were generated for $m$ sites 
over $n$ years in a study region $[l_1, l_2] = [0, 10]^2$. Max-stable 
processes are used in order to obtain a spatial dependence 
structure via R package \texttt{SpatialExtremes} 
\citep{ribatet2011spatialextremes}. 
Roughly speaking, a max-stable process is the limit process of the
component-wise sample maxima of a stochastic process, normalized by
appropriate constants  \citep[e.g.,][]{de2007extreme}. 
Its marginal distributions are GEV and its copula has to 
be an extreme-value (or equivalently, max-stable) 
copula~\citep{gudendorf2010extreme}. Two characterizations of max-stable 
processes are used - the Smith (SM) model~\citep{smith1990max} and the 
Schlather (SC) model~\citep{schlather2002models}. In addition, a 
non-extremal dependence structure is examined. One such choice is the 
Gaussian copula~\citep{renard2007use}, also known as the meta-Gaussian 
model and will be denoted here as GC. See~\cite{davison2012statistical} 
for a thorough review of these processes.

There are four factors that dictate 
each setting of the simulation: the number of sites $m$, the number of 
observations within 
each site $n$, the spatial dependence model (SM, SC, or GC), and the 
dependence level. Each setting generates 1,000 data sets, with 
$m \in \{10, 20\}$, $n \in \{10, 25\}$, and three levels of spatial 
dependence -- weak, medium, strong abbreviated as W, M, S respectively. 
The dependence levels are chosen as in \cite{wang2014incorporating}. 
For the Smith model, the correlation structure is given as 
$\Sigma = \tau I_2$ where $I_2$ is the identity matrix of dimension 2 
and $\tau$ is 4, 16, 64 corresponding to W, M, and S. For the Schlather 
model, the correlation function is chosen as a special case of the powered 
exponential $\rho(h) = \exp[ -(||h|| / \phi)^\nu ]$ with shape parameter fixed 
at $\nu=2$ and range parameter $\phi$ equal to 2.942, 5.910 and 12.153 
corresponding to W, M, and S dependence levels. These were 
chosen such that the extremal coefficient function from the SC model 
matches as close as possible to that from the SM model. For the GC model, 
the exponential correlation function $\rho(h) = \exp[ -(||h|| / \phi) ]$ 
is used, with range parameter $\phi$ as 6, 12, and 20 corresponding to 
weak, medium, and strong dependence respectively. This leads to 36 possible 
scenarios to consider.

Fr\'{e}chet scale data are generated using R package \texttt{SpatialExtremes} 
with the parameters across all replicates and settings assigned the following 
values,
\begin{align*}
\beta^{\gamma}_0 = -1.041  & &  \beta^{\xi}_0 = -0.0186 & & 
\beta^{\mu}_0 = 0.003      & &  \beta^{\gamma}_1 = \beta^{\xi}_1 = 0
\end{align*}
with the marginal location mean values generated once for the $m$ sites from 
a normal distribution with mean $5.344$ and standard deviation of $1.865$. 
Data are then transformed from the original Fr\'{e}chet scale to GEV 
following the link functions in~\eqref{eq:links}. The parameter values chosen 
here are from the estimated parameters fitting the flood index model to a subset of 
northern California sites from the GHCN dataset described in the next 
section. The covariates corresponding to the non-stationary 
location parameter $\beta^{\mu}_0$ are the latest available winter-averaged 
Southern Oscillation Index (SOI) values described in Section~\ref{ch4:data}.

The estimation methods proposed should produce asymptotically 
unbiased estimates under various forms of spatial dependence; 
as noted by~\cite{stedinger1983estimating}, as well 
as~\cite{hosking1988effect}, intersite dependence does not 
introduce significant bias, but can have a dramatic effect 
on the variance of these estimators. Another advantage of our proposed 
methods of estimation is the ability to obtain adjusted variance 
estimates, without needing to specify the form of the dependence.

Since bias has been previously shown to be negligible in the RFA setting, 
root mean squared error (RMSE) is used to assess the performance of 
each of the three estimators (MLE, MPS, and L-moment / likelihood 
hybrid). RMSE is defined as, for some parameter $\omega$, 
\begin{equation}
\label{eq:RMSE}
\sqrt{\frac{1}{J} \sum\limits_{j=1}^{J} (\omega_j - \hat{\omega}_j)^2}
\end{equation}
where $\omega_j$ is the true and $\hat{\omega}_j$ is the estimated 
parameter value in replicate $j$.

Optimization failure rates for the three estimation methods can be seen 
in Table~\ref{tab:failure_avgs}. For the L-moment hybrid method, starting 
parameters are initialized by setting each marginal mean $\beta^{\mu}_s$ 
to the stationary one-sample L-moment estimate for data from site $s$, 
non-stationary parameters equal to zero, and $\beta^{\gamma}_0, \beta^{\xi}_0$ 
equal to the stationary L-moment estimators from pooling data across 
all sites. For the MLE and MPS methods, optimization is attempted first 
with these starting parameters and if the routine fails, then uses 
the estimates from the L-moment hybrid method to initialize. In all 
settings, the MLE has the highest levels of optimization failures, 
and in some settings, the rate was 10-fold versus the MPS and L-moment 
hybrid. Between the MPS and L-moment hybrid methods, the rates are 
similar, with neither method dominating in all settings. However, 
even when the procedures did not strictly fail in optimization, 
there are cases where they provide absurd estimates. For example, 
in the setting of $m=10$, $n=10$, strong SC dependence, 
the maximum squared error seen for the shape parameter using MLE was 
over 24. Thus, a trimmed mean of squared errors are used in the 
RMSE calculations, discarding the top 2\% of values.

\begin{table}[tbp]
 \footnotesize
 \centering
 \caption{Failure rate (\%) in optimization for the three estimation methods 
 in the combined 12 settings of number of sites, observations, and dependence 
 levels within each spatial dependence structure (SC, SM, and GC) out of 10,000 
 replicates. The setup is described in detail in Section~\ref{ch4:sim}.}
    \begin{tabular}{rr ccc ccc ccc}
    \toprule
     \multicolumn{2}{c}{Setting} & \multicolumn{3}{c}{Weak} & \multicolumn{3}{c}{Medium} & \multicolumn{3}{c}{Strong} \\
     \cmidrule(lr){3-5}\cmidrule(lr){6-8}\cmidrule(lr){9-11}
                    &     & MPS    & MLE   & Hybrid & MPS   & MLE   & Hybrid & MPS   & MLE   & Hybrid      \\
  \midrule
    SC    & $m=10$, $n=10$ & 0.01  & 0.10  & 0.01  & 0.00  & 0.30  & 0.01  & 0.01  & 1.05  & 0.12 \\
          & $m=10$, $n=25$ & 0.01  & 0.25  & 0.00  & 0.00  & 0.29  & 0.01  & 0.01  & 0.35  & 0.02 \\
          & $m=20$, $n=10$ & 0.02  & 0.61  & 0.01  & 0.02  & 0.52  & 0.04  & 0.07  & 0.77  & 0.17 \\
          & $m=20$, $n=25$ & 0.05  & 0.79  & 0.01  & 0.08  & 0.68  & 0.01  & 0.03  & 0.56  & 0.04 \\[6pt]
    SM    & $m=10$, $n=10$ & 0.02  & 0.12  & 0.00  & 0.00  & 0.34  & 0.07  & 0.00  & 0.95  & 0.20 \\
          & $m=10$, $n=25$ & 0.00  & 0.26  & 0.00  & 0.00  & 0.26  & 0.00  & 0.01  & 0.23  & 0.00 \\
          & $m=20$, $n=10$ & 0.03  & 0.50  & 0.01  & 0.07  & 0.54  & 0.02  & 0.04  & 0.66  & 0.15 \\
          & $m=20$, $n=25$ & 0.08  & 0.75  & 0.02  & 0.05  & 0.62  & 0.02  & 0.05  & 0.61  & 0.02 \\[6pt]
    GC    & $m=10$, $n=10$ & 0.01  & 0.21  & 0.00  & 0.00  & 0.22  & 0.02  & 0.00  & 0.27  & 0.03 \\
          & $m=10$, $n=25$ & 0.01  & 0.37  & 0.00  & 0.00  & 0.30  & 0.00  & 0.02  & 0.26  & 0.00 \\
          & $m=20$, $n=10$ & 0.01  & 0.48  & 0.01  & 0.04  & 0.38  & 0.00  & 0.02  & 0.65  & 0.04 \\
          & $m=20$, $n=25$ & 0.03  & 0.66  & 0.01  & 0.06  & 0.71  & 0.00  & 0.08  & 0.59  & 0.00 \\
    \bottomrule
    \end{tabular}
\label{tab:failure_avgs}
\end{table}

\begin{figure}[tbp]
    \centering
      \includegraphics[scale=0.8]{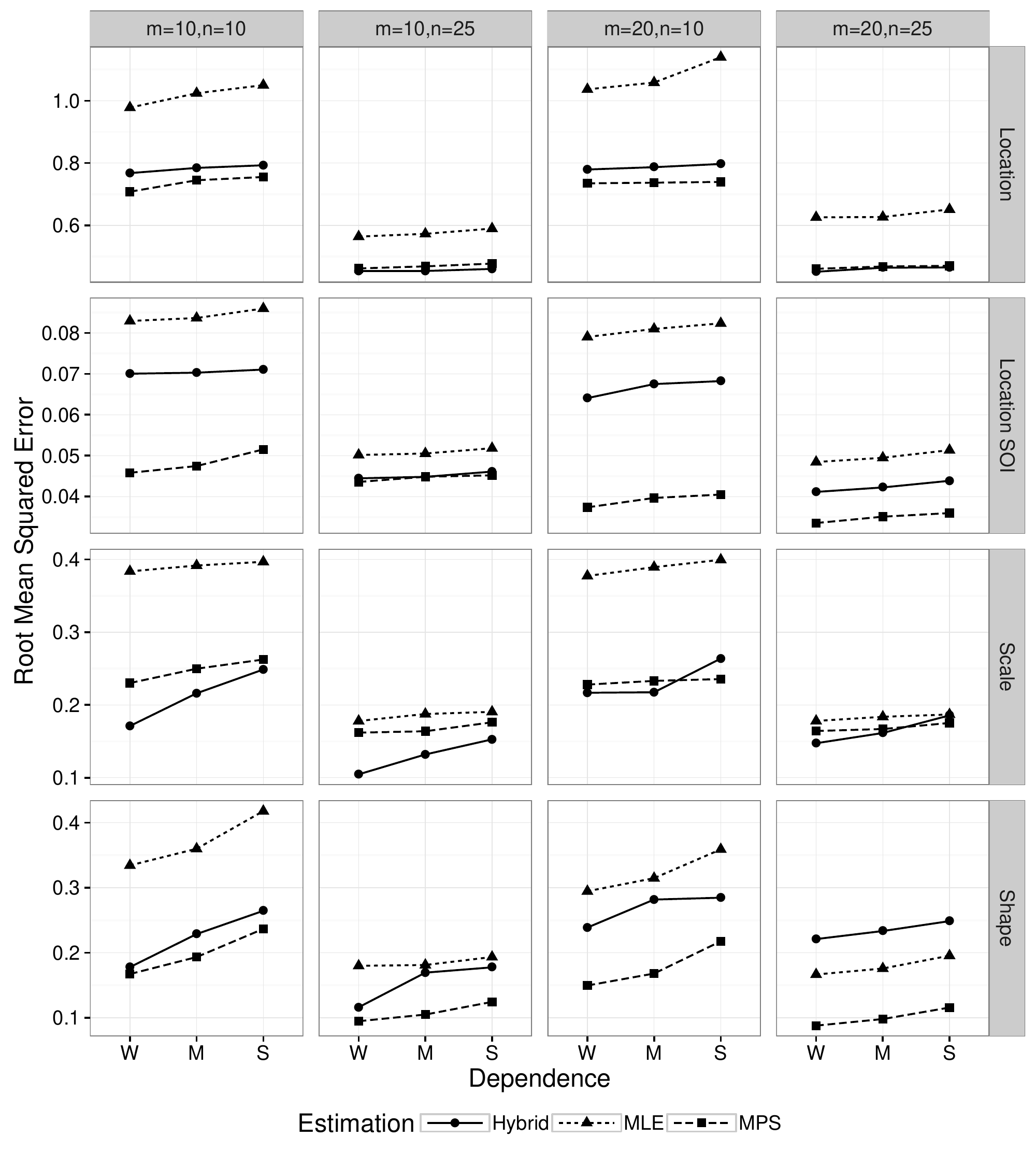}
    \caption{Schlather model root mean squared error of the parameters for each estimation 
    method, from 1000 replicates of each setting discussed in Section~\ref{ch4:sim}. 
    W, M, S refers to weak, medium, and strong dependence, with $m$ being the number of 
    sites and $n$, the number of observations within each site.}
    \label{fig:SC_RMSE}
\end{figure}

\begin{figure}[tbp]
    \centering
      \includegraphics[scale=0.8]{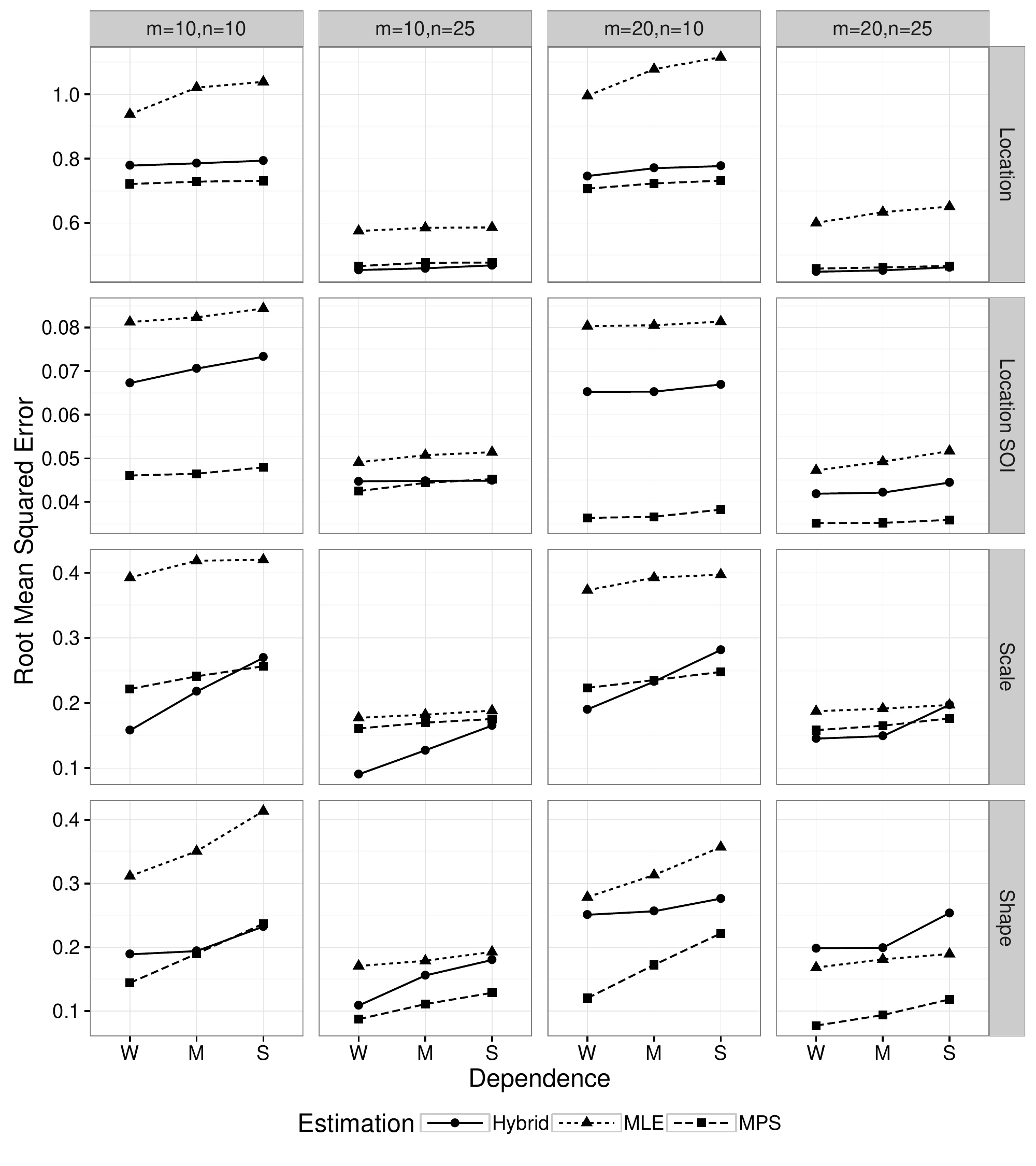}
    \caption{Smith model root mean squared error of the parameters for each estimation 
    method, from 1000 replicates of each setting discussed in Section~\ref{ch4:sim}. 
    W, M, S refers to weak, medium, and strong dependence, with $m$ being the number of 
    sites and $n$, the number of observations within each site.}
    \label{fig:SM_RMSE}
\end{figure}

\begin{figure}[tbp]
    \centering
      \includegraphics[scale=0.8]{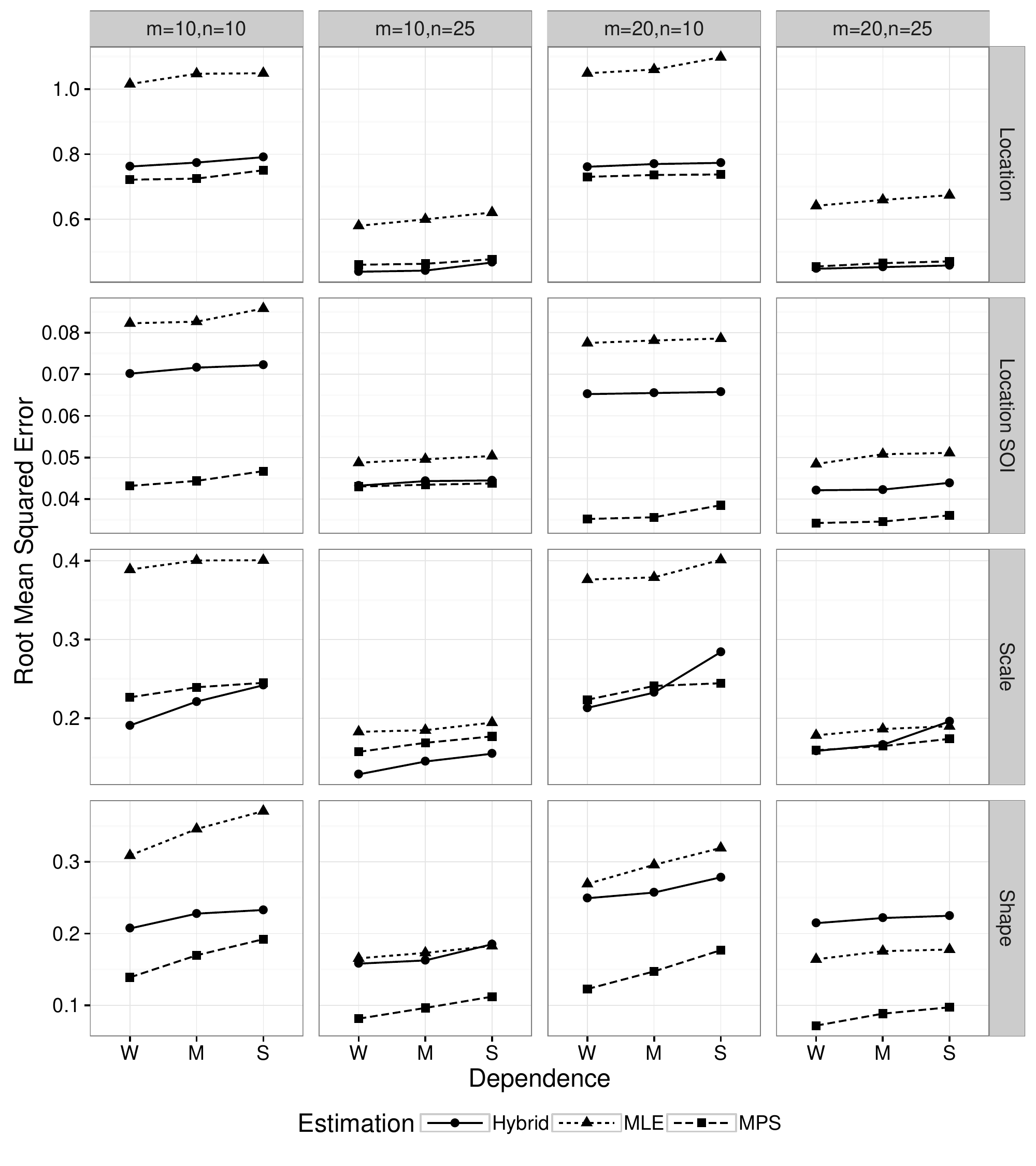}
    \caption{Gaussian copula model root mean squared error of the parameters for each 
    estimation method, from 1000 replicates of each setting discussed in Section~\ref{ch4:sim}. 
    W, M, S refers to weak, medium, and strong dependence, with $m$ being the number of 
    sites and $n$, the number of observations within each site.}
    \label{fig:GC_RMSE}
\end{figure}

\afterpage{\clearpage}

Figures \ref{fig:SC_RMSE}, \ref{fig:SM_RMSE}, and \ref{fig:GC_RMSE} show the RMSE 
for each of the simulation settings for the Schlather, Smith, and Gaussian Copula 
dependence structures, respectively. In every case, the RMSE of the 
MPS estimation procedure is smaller than that of MLE. The hybrid 
MLE / L-moment outperforms the MLE in the smallest sample case ($n=10$), 
but has a larger RMSE for the shape parameter when $m=20$ and $n=25$. As expected, 
an increase in sample size corresponds to a reduction in RMSE for all methods. 
However, the MLE and MLE / L-moment hybrid tend to see a more significant reduction 
than MPS, suggesting the stability of MPS estimates in small sample sizes.

\section{California Annual Daily Maximum Winter Precipitation}
\label{ch4:data}

Daily precipitation data (in cm) is available for tens of 
thousands of surface sites around the world via the Global 
Historical Climatology Network (GHCN). An overview of the 
dataset can be found in~\citet{menne2012overview}. Although 
some California sites have records dating before the year 1900, 
it is quite sparse. To ensure a large enough and non-biased 
collection of sites, only the most recent 53 years (1963 -- 2015) 
are considered. In this period, there are 27 sites with complete 
records available in the state of California. Their locations can 
be seen in Figure~\ref{fig:SiteMap}.

\begin{figure}[tbp]
    \centering
      \includegraphics[scale=0.6]{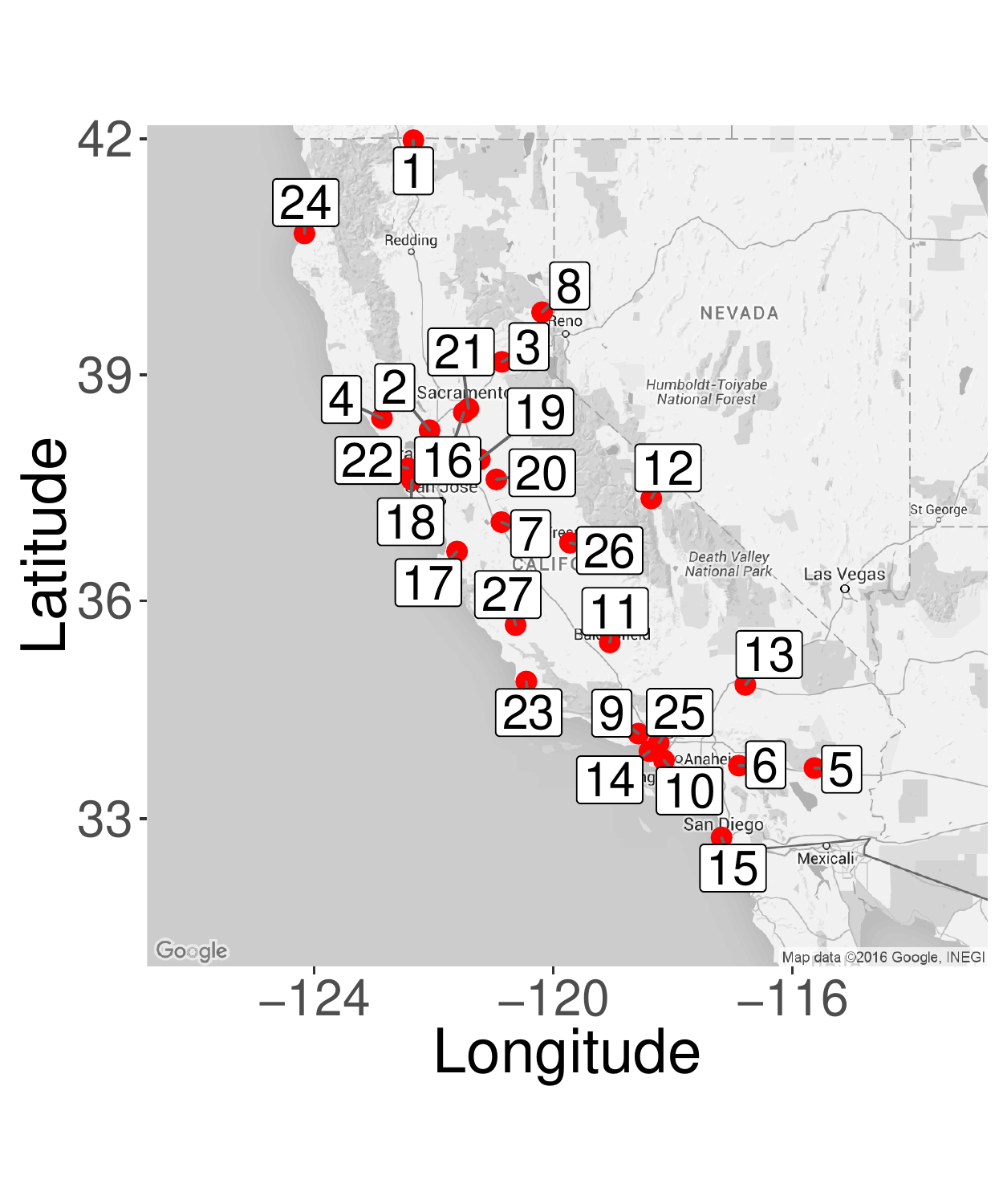}
    \caption{Locations of the 27 California sites used in 
    the non-stationary regional frequency analysis of annual daily 
    maximum winter precipitation events.}
    \label{fig:SiteMap}
\end{figure}

For each year, the block maxima is observed as the daily maximum 
winter precipitation event for the period of December 1 through 
March 31. If more than 10\% of the observations in this period 
were missing, the block maxima from that year is considered as 
missing. A process suspected to be related to extreme precipitation 
events~\citep[e.g.,][]{schubert2008enso} is El Ni\~{n}o--Southern 
Oscillation (ENSO). ENSO can be represented using the Southern 
Oscillation Index (SOI), a measure of observed sea level pressure 
differences and typically has one measurement per calendar month. 
For purposes here, the mean SOI value over current winter months 
(December through March) is treated as a covariate in the non-stationary 
regional frequency analysis model. The raw SOI data is taken from 
the Australia Bureau of Meteorology's website, 
\url{http://www.bom.gov.au/climate/current/soihtm1.shtml}.

To get a sense of the dependence among sites, Figure~\ref{fig:DistanceCorrelations} 
displays a scatterplot of the pairwise spearman correlations 
by euclidean distance for the 27 sites. There appears to be a 
decreasing trend as the distance between two sites increases. Over 
25\% of the correlation coefficients are above 0.41, so there 
is a need to account for this spatial dependence in the variance 
of the estimators.

\begin{figure}[tbp]
    \centering
      \includegraphics[scale=1.3]{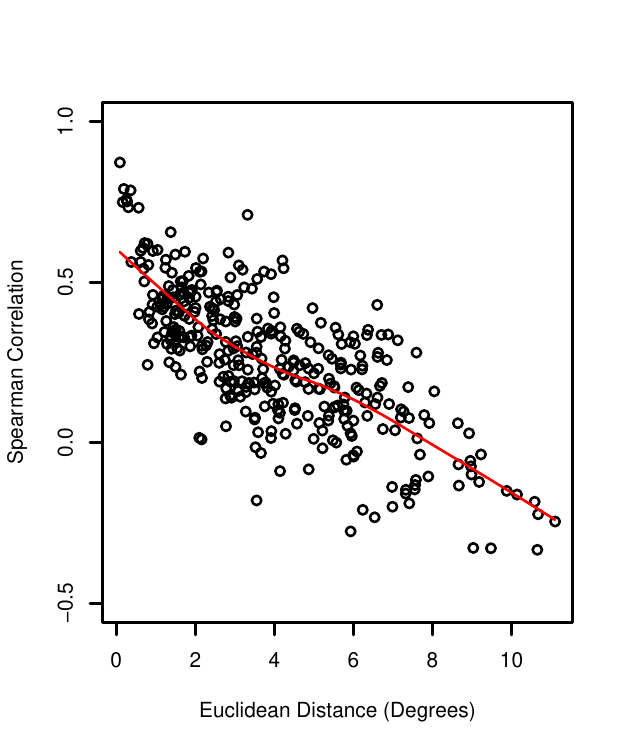}
    \caption{Scatterplot of Spearman correlations by euclidean 
    distance between each pair of the 27 California sites used 
    in the non-stationary regional frequency analysis of annual 
    daily maximum winter precipitation events.}
    \label{fig:DistanceCorrelations}
\end{figure}

The flood index model is fit with the links given in~\eqref{eq:links}, 
treating the intercept shape and proportionality terms as shared across 
sites, denoted as $\xi_0$ and $\gamma_0$ respectively. Site-specific 
marginal mean parameters are assigned. An exploratory visual analysis 
of SOI and the winter precipitation maximums show an increase in 
magnitude of events for negative SOI years, with a less clear trend 
for positive SOI values. Thus, piecewise linear terms for positive 
and negative SOI values are incorporated into the location parameter. 
The two distinct SOI piecewise linear terms can be defined as 
$\text{SOI} \cdot \mathbbm{1}(\text{SOI} > 0)$ and 
$\text{SOI} \cdot \mathbbm{1}(\text{SOI} \leq 0)$. The model 
parameters are estimated with the three methods discussed previously. 
Standard errors of the estimates need to be adjusted due to 
spatial dependence and can be obtained via semi-parametric 
bootstrap. The resampling method of~\cite{heffernan2004conditional} 
is implemented as described in Appendix~\ref{app:boot}. Though 
computationally expensive, this bootstrap procedure can be directly 
extended to run in parallel.

\begin{figure}[tbp]
    \centering
      \includegraphics[scale=0.7]{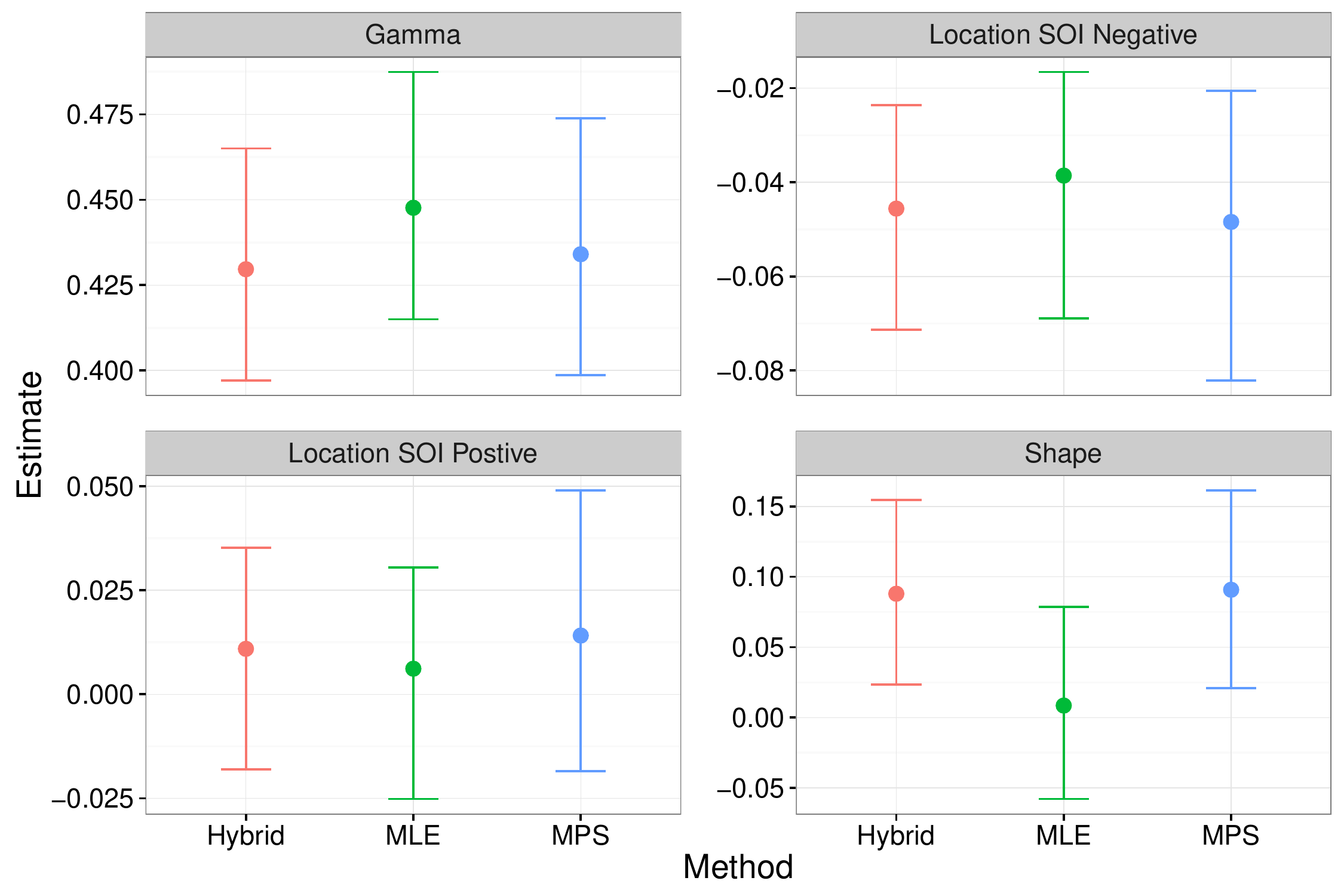}
    \caption{Estimates and 95\% semi-parametric bootstrap confidence intervals of the 
    location parameter covariates (postive and negative SOI piecewise terms), 
    proportionality, and shape parameters for the three methods of estimation 
    in the non-stationary RFA of the 27 California site annual winter maximum 
    precipitation events.}
    \label{fig:cov_estimates}
\end{figure}

The results of the analysis are seen in Figures~\ref{fig:cov_estimates} 
and~\ref{fig:loc_estimates}. The semi-parametric bootstrap confidence 
intervals show that all three methods of estimation provide negative 
estimates of the negative linear SOI term at the 5\% significance level. 
For the positive linear SOI term, all three methods do not detect a 
significant difference from zero. Interestingly, only maximum likelihood 
estimation fails to detect a significant difference from zero in the 
shape parameter; both MPS and the hybrid L-moment method indicate a 
heavy-tailed ($\xi > 0$) distribution.

\begin{figure}[tbp]
    \centering
      \includegraphics[scale=0.7]{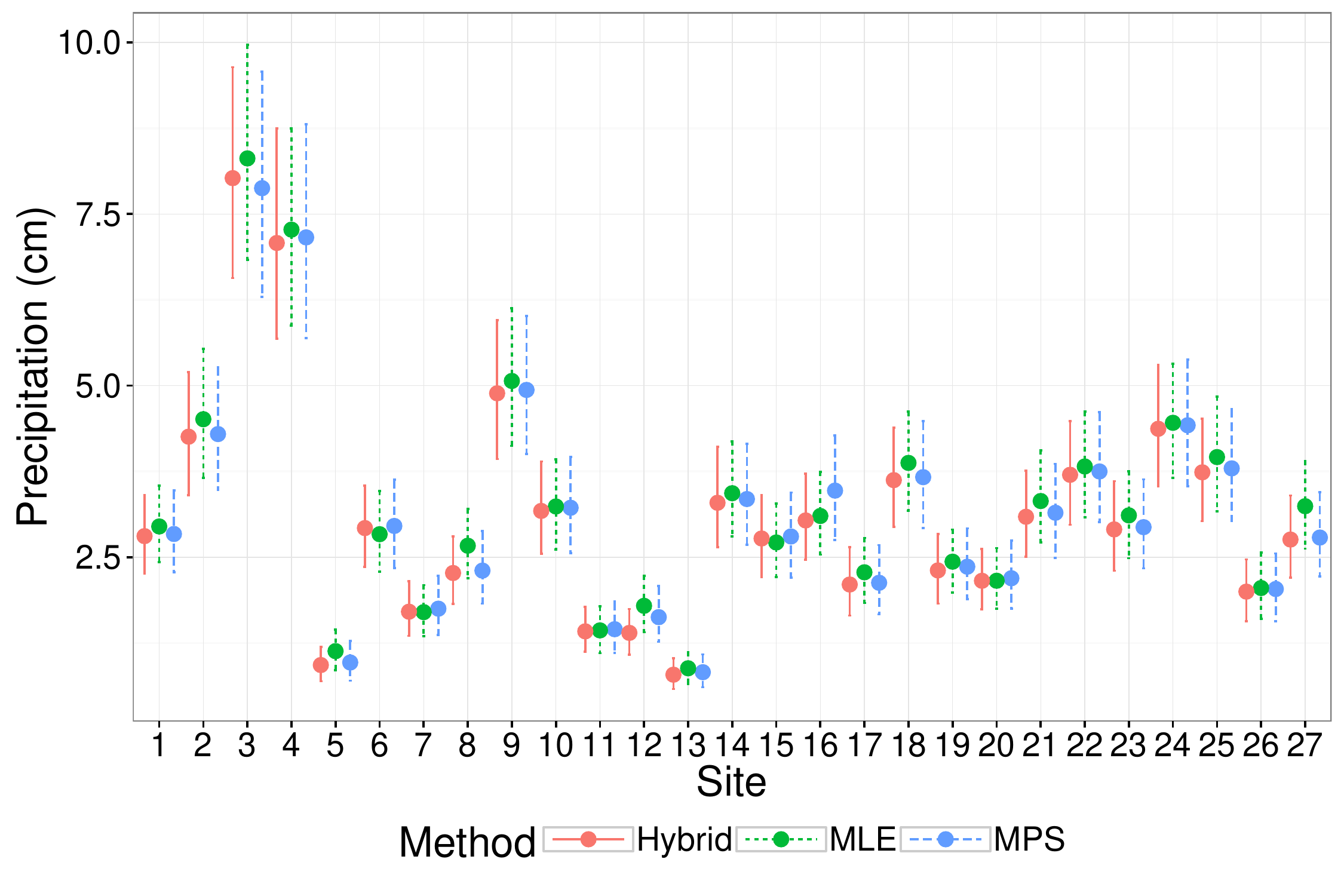}
    \caption{Estimates and 95\% semi-parametric bootstrap confidence intervals of the marginal 
    site-specific location means for the three estimation methods in the non-stationary RFA 
    of the 27 California site annual winter maximum precipitation events.}
    \label{fig:loc_estimates}
\end{figure}

To further demonstrate the differences between the three methods, 
a subset of years from the original 53 year sample was taken, keeping 
just one out of every three years, making this subset sample of record 
length 18. The same analysis was performed on this subset of data and 
estimates of the parameters are compared to the full sample. For all 
three methods, the two SOI parameters fail to detect departures from 
zero at the 5\% level. Stark differences can be seen in 
Figure~\ref{fig:shape_compare} for the shape parameter estimates 
between the three methods. The MPS and hybrid L-moment estimates 
for the shape parameter are quite stable between the two samples, 
both decreasing approximately the same, from 0.09 to 0.07. The MLE 
shape estimate for the subset sample is completely outside the 95\% 
interval range based on the full sample, switching from a positive to 
negative estimate. For simplicity, the remainder of the analysis (on 
the full sample) will use the MPS estimators.

\begin{figure}[tbp]
    \centering
      \includegraphics[scale=0.8]{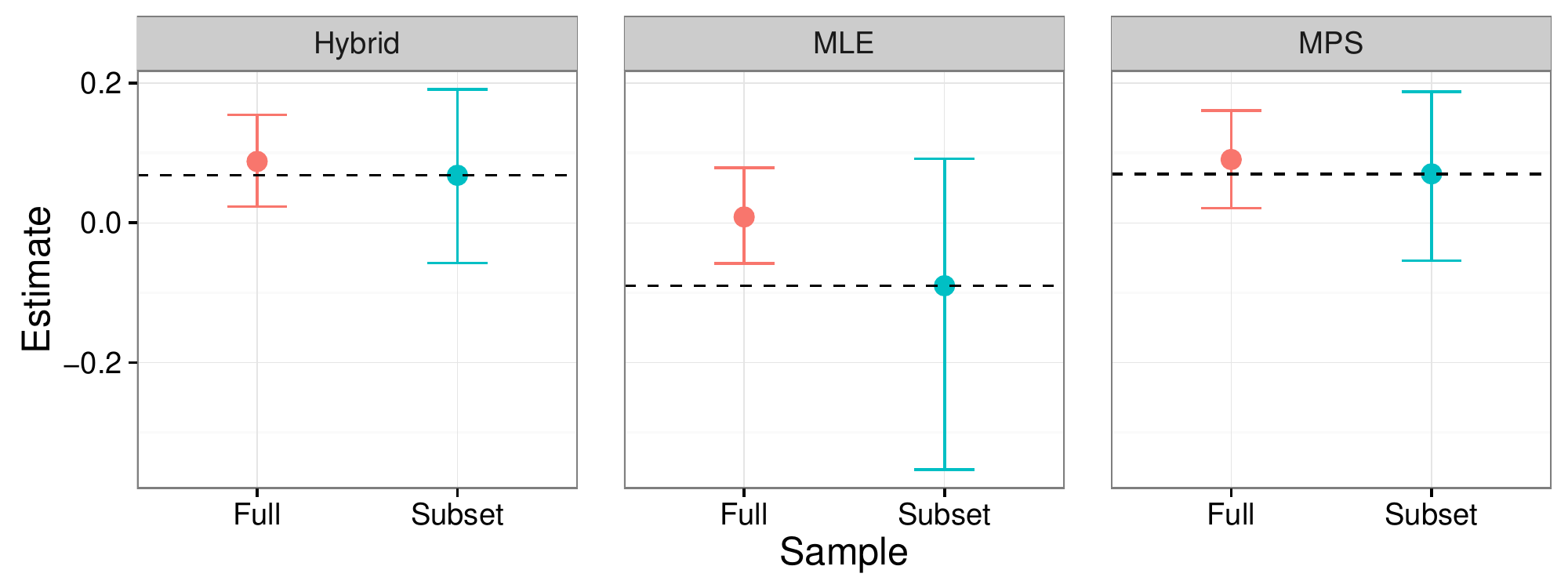}
    \caption{Estimates and 95\% semi-parametric bootstrap confidence intervals of the 
    shape parameter by the three estimation methods, for the full 53 year and 18 year 
    subset sample of California annual winter precipitation extremes. The horizontal dashed 
    line corresponds to the shape parameter estimate of each method for the subset sample.}
    \label{fig:shape_compare}
\end{figure}

A quantity of great interest to researchers are the $t$-year return 
levels, which can be thought of as the maximum event that will occur 
on average every $t$ years. The marginal return levels at each site 
are determined jointly from the regional frequency model, dependent 
on estimates of the shared parameters and site-specific marginal 
means. Calculation of the quantity can be found in~\eqref{eq:gev_rl}, 
replacing the stationary parameters with the appropriate non-stationary 
values. The left half of Figure~\ref{fig:rl_plots} shows the 
50 year return level estimates of precipitation (in centimeters) using 
the MPS estimation, conditioned on the mean value of SOI in the sample, 
$-0.40$. In other words, these are the estimated maximum 50 year daily 
winter precipitation events for an average ENSO year.

\begin{figure}[tbp]
    \centering
      \includegraphics[width=\textwidth]{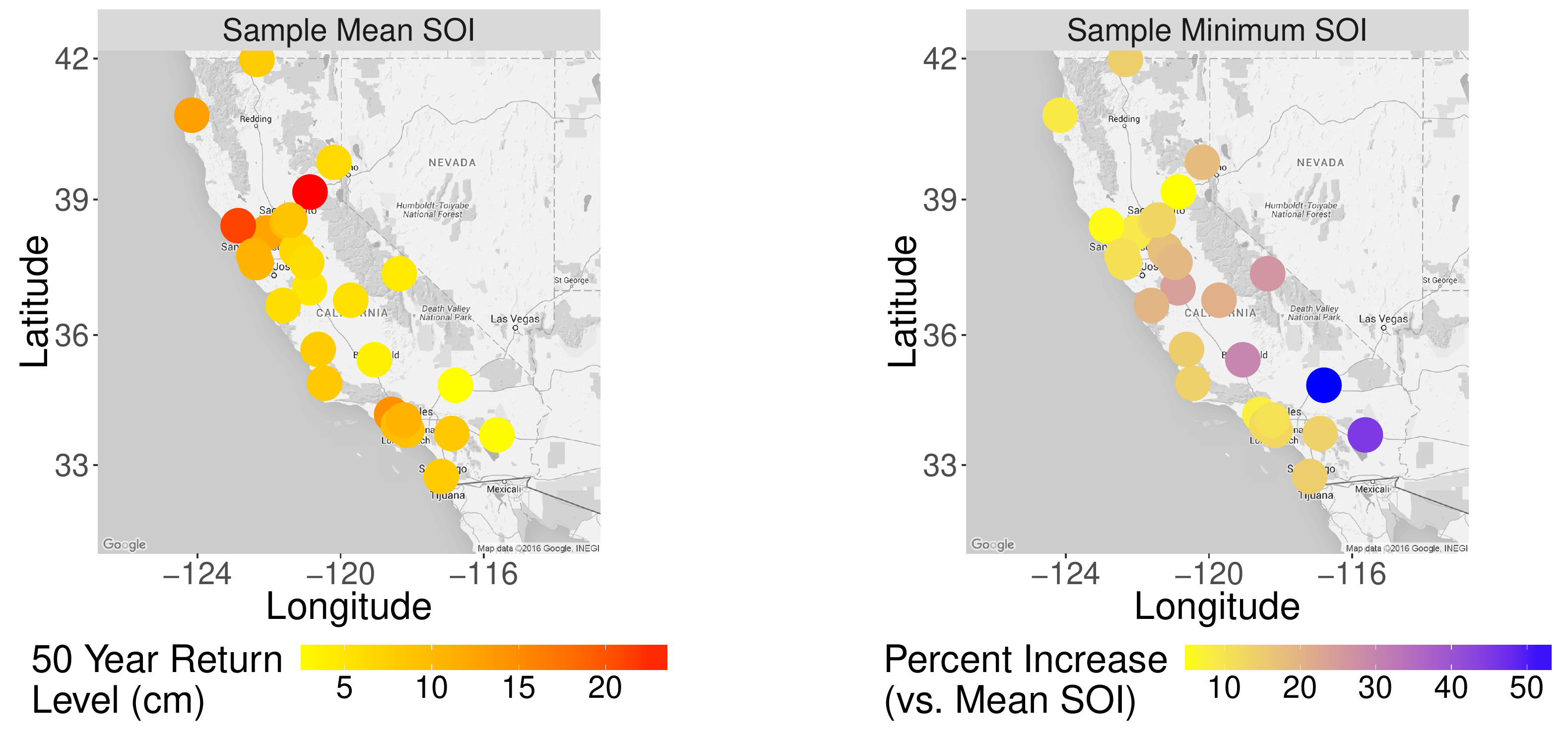}
    \caption{Left: 50 year return level estimates (using MPS) at the 27 sites, 
    conditioned on the mean sample SOI value ($-0.40$). 
    Right: Estimated percent increase in magnitude of the 50 year event 
    at the sample minimum SOI ($-28.30$) versus the mean SOI value.
    }
    \label{fig:rl_plots}
\end{figure}

To determine the effect of ENSO on return levels, we compare the 50 year 
return levels conditioned on the sample minimum ($-28.30$) SOI values 
to the mean SOI. The right half of Figure~\ref{fig:rl_plots} 
shows the percent change of the estimated 50 year return levels 
conditioned on the sample minimum SOI, versus the mean value of SOI. 
The findings suggest that strong El Ni\~{n}o (negative SOI) years 
increase the estimated effects of a 50 year precipitation event 
significantly (in some cases over 50\% larger estimates), 
while strong La Ni\~{n}a (positive SOI) events do not see a 
significant change from the average ENSO year.

\section{Discussion}
\label{ch4:disc}

We propose two alternative estimation methods to maximum likelihood 
(ML) for fitting non-stationary regional frequency models. The first, 
maximum product spacing (MPS) estimation seeks to maximize the 
geometric mean of spacing in the data and the second is a hybrid 
of ML and L-moment estimation. These are studied in the context of 
a flood index model with extreme value marginal distributions, but 
the methodology can be generalized to any regional frequency model.

In the context of distributions in the extremal domain, ML estimation 
has some known drawbacks; in small samples the L-moment estimation 
method is more efficient and the ML estimate may not exist for 
certain values of the shape parameter. Additionally, it has been 
shown~\cite{wong2006note} empirically that MLE can provide absurd 
estimates in small samples even if optimization does not fail entirely. 
In stationary RFA, an efficient alternative procedure can estimate 
the parameters using only L-moments~\citep[e.g.,][]{wang2014incorporating}. 
However, extending this method to the non-stationary case is not 
straightforward. Existing attempts have serious limitations and 
are reviewed in Section~\ref{ch4:existing_methods}.

MPS estimation is well-studied in the one-sample case and this 
contribution is an extension to the RFA setting. As with MLE, 
it is intuitive to incorporate covariates in all parameters, 
is as efficient asymptotically, and exists for all values 
of the shape parameter. The hybrid MLE / L-moment procedure is 
an iterative estimation procedure that estimates stationary 
parameters using L-moments and non-stationary parameters via 
likelihood. While the procedure itself is not as straightforward 
as MPS or strictly MLE, it is more computationally efficient as 
it only has to perform optimization over the non-stationary 
parameters.

The performance of the three estimation procedures are evaluated 
in the flood index non-stationary RFA model via a large scale study 
of data simulated under extremal and Gaussian dependence with parameters 
chosen from real data. It is shown that the root mean squared error 
(RMSE) of the MPS estimator is smaller than that of MLE in all cases. 
In most cases, the hybrid L-moment / MLE outperformed the pure 
MLE method as well. Additionally, the MPS method is stable for sample 
sizes as small as 10, which is desirable in environmental applications 
with extremal data as historical record is often short.

The three estimation methods are applied to a flood-index model for 
maximum annual precipitation events in California with an emphasis 
on using the Southern Oscillation Index (SOI) as a predictor. To 
avoid explicitly specifying the spatial dependence structure, a 
semi-parametric bootstrap procedure based on the ideas 
of~\cite{heffernan2004conditional} is used to obtain adjusted 
standard errors of estimates. The three methods find a 
statistically significant negative association between negative 
values of SOI (El Ni\~{n}o) and maximum precipitation magnitudes, 
but not for positive SOI values (La Ni\~{n}a). The MPS 
and hybrid L-moment methods indicate a positive shape parameter 
value, while the MLE does not find this to be significant. 
Using MPS, the effect of SOI on return levels is studied and 
it is found that El Ni\~{n}o has an increasing effect on the magnitude 
of these events. Strong El Ni\~{n}o years are estimated to increase 
the magnitude of 50 year return level events on average 23.6\% and 
up to 60\% at some sites. The findings agree with previous 
studies~\citep[e.g.,][]{el2007generalized,schubert2008enso,Shan:Yan:Zhan:enso:2011,zhang2010influence}.

Based on the simulation performance, stability in the subset versus 
full sample estimates for the California precipitation RFA, 
and ease to incorporate covariates, the recommendation to 
practitioners would be to use the MPS method. On a more general 
note, a current limitation of this analysis is 
the requirement of complete and balanced data. As 
missingness is common in historical climate records, there 
is a need to develop methodology to handle missing data 
in the semi-parametric resampling schemes in Appendix~\ref{app:boot}. 
A further discussion on this subject is provided in Section~\ref{ch6:future}.

\chapter{An R Package for Extreme Value Analysis: \texttt{eva}}
\label{ch:soft}

\section{Introduction}
\label{introduction-to-the-eva-package}

It is quite apparent that statistics cannot be separated from computing and
there are some particular challenges in extreme value analysis that need
to be addressed. A thorough review of existing packages can be found in
\cite{gilleland2016computing} and \cite{gilleland2013software}.

The \texttt{eva} package, short for `Extreme Value Analysis with
Goodness-of-Fit Testing', provides functionality that allows data
analysis of extremes from beginning to end, with model fitting and a set
of newly available tests for diagnostics. In particular, some highlights
are:

\begin{itemize}
\item
  Efficient handling of the near-zero shape parameter.
\item
  Implementation of the $r$ largest order statistics (GEV$_r$) model
  - data generation, non-stationary fitting, and return levels.
\item
  Maximum product spacing (MPS) estimation for parameters in the block
  maxima (GEV$_1$) and Generalized Pareto distribution (GPD).
\item
  Sequential tests for the choice of $r$ in the GEV$_r$ model, as
  well as tests for the selection of threshold in the
  peaks-over-threshold (POT) approach. For the bootstrap based tests, the
  option to run in parallel is provided.
\item
  P-value adjustments to control for the false discovery rate (FDR) and
  family-wise error rate (FWER) in the sequential testing setting.
\item
  Profile likelihood for return levels and shape parameters of the
  GEV$_r$ and GPD.
\end{itemize}

While some of these issues / features may have been implemented in some
existing packages, it is not predominant and there are some new features
in the \texttt{eva} package that to the best of our knowledge have not
been implemented elsewhere. In particular, completely new implementations 
are MPS estimation, p-value adjustments for error control in sequential 
ordered hypothesis testing, data generation and density calculations for 
the GEV$_r$ distribution. Additionally, the goodness-of-fit tests for the 
selection of $r$ in the stationary GEV$_r$ distribution are the only 
non-visual diagnostics currently available -- see Chapter~\ref{ch:r-largest} 
for more details. Similarly related, the only model fitting whatsoever for 
the GEV$_r$ distribution (stationary or non-stationary) can be found in 
R package \texttt{ismev}~\citep{ismev2016}, which requires specification of the exact 
design matrix and does not explicitly handle the numerical issue for 
small values of the shape parameter. The \texttt{eva} package introduces 
model formulations similar to base function \texttt{lm} and \texttt{glm}, 
provides improved optimization, and inherits usual class functionality for 
modeling data such as \texttt{AIC}, \texttt{BIC}, \texttt{logLik}, etc. 
A similar implementation is provided for fitting the non-stationary GPD. 
The last major improvement worth noting is profile likelihood estimation. 
Other R packages such as \texttt{extRemes}~\citep{extRemes2011} and \texttt{ismev} 
require specification of the boundaries of the interval, making automation 
less straightforward. The package \texttt{fExtremes}~\citep{fExtremes2013} 
does not require this; however it does not provide functionality for the 
GEV$_r$ distribution. The \texttt{eva} package includes semi-automated 
profile likelihood confidence interval estimation of the shape and 
return levels for the GPD and GEV$_r$ (implicitly block maxima) stationary 
models.

\section{Efficient handling of near-zero shape parameter}
\label{efficient-handling-of-near-zero-shape-parameter}

One of the interesting and theoretical aspects of the GEV / GEV$_r$
and GP distributions is the limiting form of each as the shape parameter
$\xi$ approaches zero. In practice, numerical instabilities can occur,
particularly when working with the form $(1 + x \xi)^{(-1 / \xi)}$ as
$\xi \rightarrow 0$. Figure \ref{fig:gev_naive_cdf} shows both a naive
and the efficient implementation in \texttt{eva} of the GEV cumulative
distribution function \eqref{eq:gev_cdf} at $x = 1$, $\mu = 0$ and
$\sigma = 1$, with $\xi$ varying from $-0.0001$ to $0.0001$ on the cubic
scale. Note that the true value of the CDF function at $\xi = 0$ is 
approximately 0.692. A similar demonstration can be carried out for the 
probability density function and for the GPD.

Handling of this issue is critical, particularly in optimization when
the shape parameter may change its sign as the parameter space is
navigated. The \texttt{eva} package handles these cases by rewriting
certain functions to utilize high precision operations such as 
\texttt{log1p}, \texttt{expm1} and some of the use cases in 
R discussed in \cite{machler2012accurately}.

\begin{Shaded}
\footnotesize
\begin{Highlighting}[]
\CommentTok{# A naive implementation of the GEV cumulative distribution function}
\NormalTok{pgev_naive <-}\StringTok{ }\NormalTok{function(q, }\DataTypeTok{loc =} \DecValTok{0}\NormalTok{, }\DataTypeTok{scale =} \DecValTok{1}\NormalTok{, }\DataTypeTok{shape =} \DecValTok{0}\NormalTok{) \{}
  \NormalTok{q <-}\StringTok{ }\NormalTok{(q -}\StringTok{ }\NormalTok{loc) /}\StringTok{ }\NormalTok{scale}
  \KeywordTok{ifelse}\NormalTok{(shape ==}\StringTok{ }\DecValTok{0}\NormalTok{, }\KeywordTok{exp}\NormalTok{(-}\KeywordTok{exp}\NormalTok{(-q)), }
         \KeywordTok{exp}\NormalTok{(-(}\DecValTok{1} \NormalTok{+}\StringTok{ }\NormalTok{shape *}\StringTok{ }\NormalTok{q)^(-}\DecValTok{1}\NormalTok{/shape)))}
\NormalTok{\}}
\end{Highlighting}
\end{Shaded}

\begin{figure}[tbp]
\centering
\includegraphics{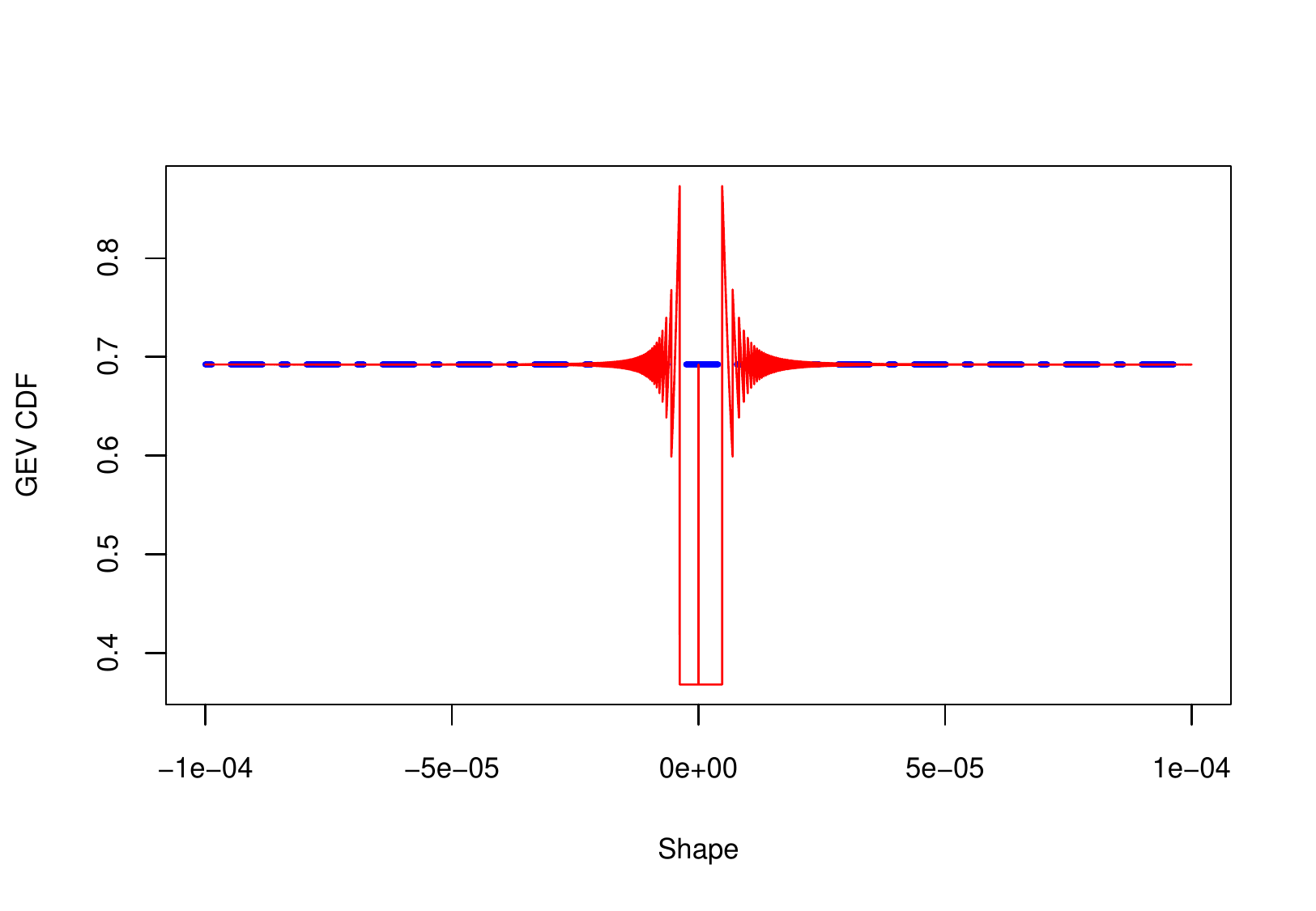}
\caption{\label{fig:gev_naive_cdf} Plot of GEV cumulative distribution
function with $x = 1$, $\mu = 0$ and $\sigma = 1$, with $\xi$ 
ranging from $-0.0001$ to $0.0001$ on the cubic scale. The naive implementation is 
represented by the solid red line, with the implementation in R package 
\texttt{eva} as the dashed blue line.}
\end{figure}

\section{The GEV$_r$ distribution}
\label{the-gevux5fr-distribution}

The GEV$_r$ distribution has the density function given in
\eqref{eq:gevr_pdf}. When $r = 1$, this distribution is exactly the
GEV distribution or block maxima. The \texttt{eva} package includes data
generation (\texttt{rgevr}), density function (\texttt{dgevr}), 
non-stationary fitting (\texttt{gevrFit}), and return levels (\texttt{gevrRl}) 
for this distribution. To the best of the author's knowledge, there 
is no current implementation whatsoever for the density function 
and data generation of the GEV$_r$ distribution. Appendix 
\ref{app:gevrsim} gives the psuedo code and a simplified version 
of the data generation algorithm.

\subsection{Goodness-of-fit testing}
\label{goodness-of-fit-testing}

When modeling the $r$ largest order statistics, if one wants to choose
$r > 1$, goodness-of-fit must be tested. This can be done using
function \texttt{gevrSeqTests}, which can be used for the selection of $r$, 
via the entropy difference test and score tests discussed in Sections
\ref{ch2:ed} and \ref{ch2:score} respectively. Take, for example, the
Lowestoft dataset in Section \ref{ch2:lowe} which includes the top
hourly sea level readings for the period of 1964 -- 2014.

\begin{Shaded}
\footnotesize
\begin{Highlighting}[]
\KeywordTok{data}\NormalTok{(lowestoft)}
\KeywordTok{gevrSeqTests}\NormalTok{(lowestoft, }\DataTypeTok{method =} \StringTok{"ed"}\NormalTok{)}
\end{Highlighting}
\end{Shaded}

\begin{verbatim}
##  r  p.values ForwardStop StrongStop   statistic  est.loc est.scale est.shape
##  2 0.9774687    1.455825  0.9974711  0.02824258 3.431792 0.2346591 0.10049739
##  3 0.7747741    1.163697  1.0869254 -0.28613586 3.434097 0.2397408 0.09172687
##  4 0.5830318    1.116989  1.1500564  0.54896156 3.447928 0.2404563 0.06802070
##  5 0.6445475    1.157363  1.2470248  0.46135009 3.452449 0.2376723 0.05451138
##  6 0.4361569    1.181963  1.2676077 -0.77869930 3.455478 0.2396332 0.04709329
##  7 0.6329943    1.334209  1.4133341  0.47751655 3.454680 0.2372572 0.04555449
##  8 0.4835074    1.444819  1.4790559 -0.70067260 3.455901 0.2376215 0.03838020
##  9 0.8390270    1.836882  2.0321878 -0.20313820 3.458135 0.2356543 0.02536685
## 10 0.8423291    1.847245  3.4235418  0.19891516 3.459470 0.2342272 0.01964612
\end{verbatim}

The entropy difference test fails to reject for any value of $r$ from
1 to 10. In the previous output, the adjusted ForwardStop
\eqref{eq:forwardstop} and StrongStop \eqref{eq:strongstop} p-values are
also included.

\subsection{Profile likelihood}
\label{profile-likelihood}

A common quantity of interest in extreme value analysis are the
$t$-year return levels, which can be thought of as the average maximum
value that will be exceeded over a period of $t$ years. For the
LoweStoft data, the 250 year sea level return levels, with 95\%
confidence intervals are plotted for $r$ from 1 to 10. The advantage
of using more top order statistics can be seen in the plots below. The
width of the intervals decrease by over two-thirds from $r=1$ to
$r=10$. Similarly decreases can be seen in the parameter intervals.

\begin{figure}[tbp]
\centering
\includegraphics{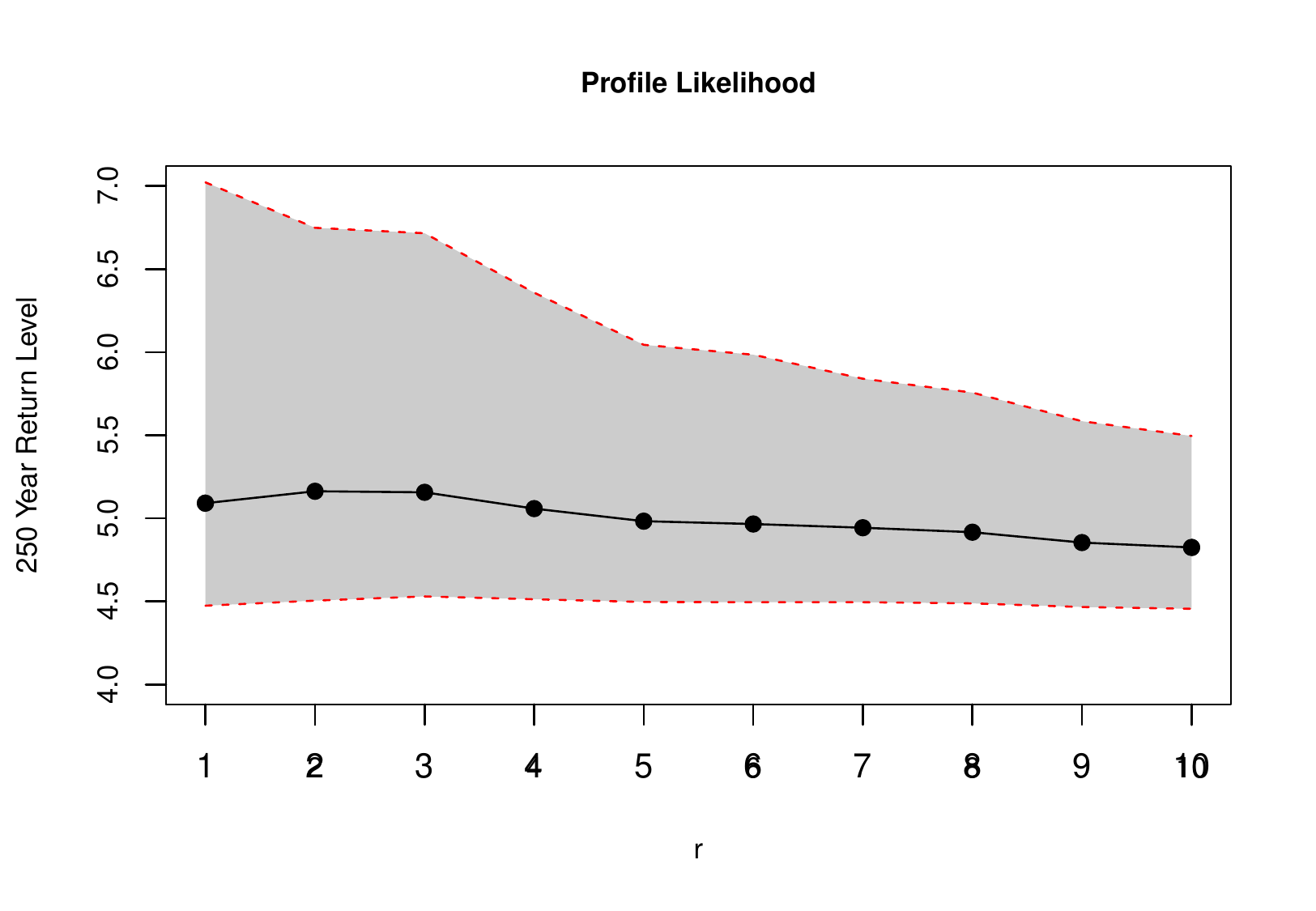}
\caption{\label{fig:lowestoft_prof} Estimates and 95\% profile likelihood 
confidence intervals for the 250 year return level of the LoweStoft sea level
dataset, for $r=1$ through $r=10$.}
\end{figure}

In addition, the profile likelihood confidence intervals (Figure
\ref{fig:lowestoft_prof}) are compared with the delta method intervals
(Figure \ref{fig:lowestoft_delta}). The advantage of using profile
likelihood over the delta method is the allowance for asymmetric
intervals. This is especially useful at high quantiles, or large return
level periods. At the moment, no other software package provides a
straightforward way to obtain profile likelihood confidence intervals
for the GEV$_r$ distribution. In Figure \ref{fig:lowestoft_prof}, the
asymmetry can be seen in the stable lower bound across values of $r$,
while the upper bound decreases. Profile likelihood is also implemented
for estimating the shape parameter and in the same settings for the GPD.

\begin{figure}[tbp]
\centering
\includegraphics{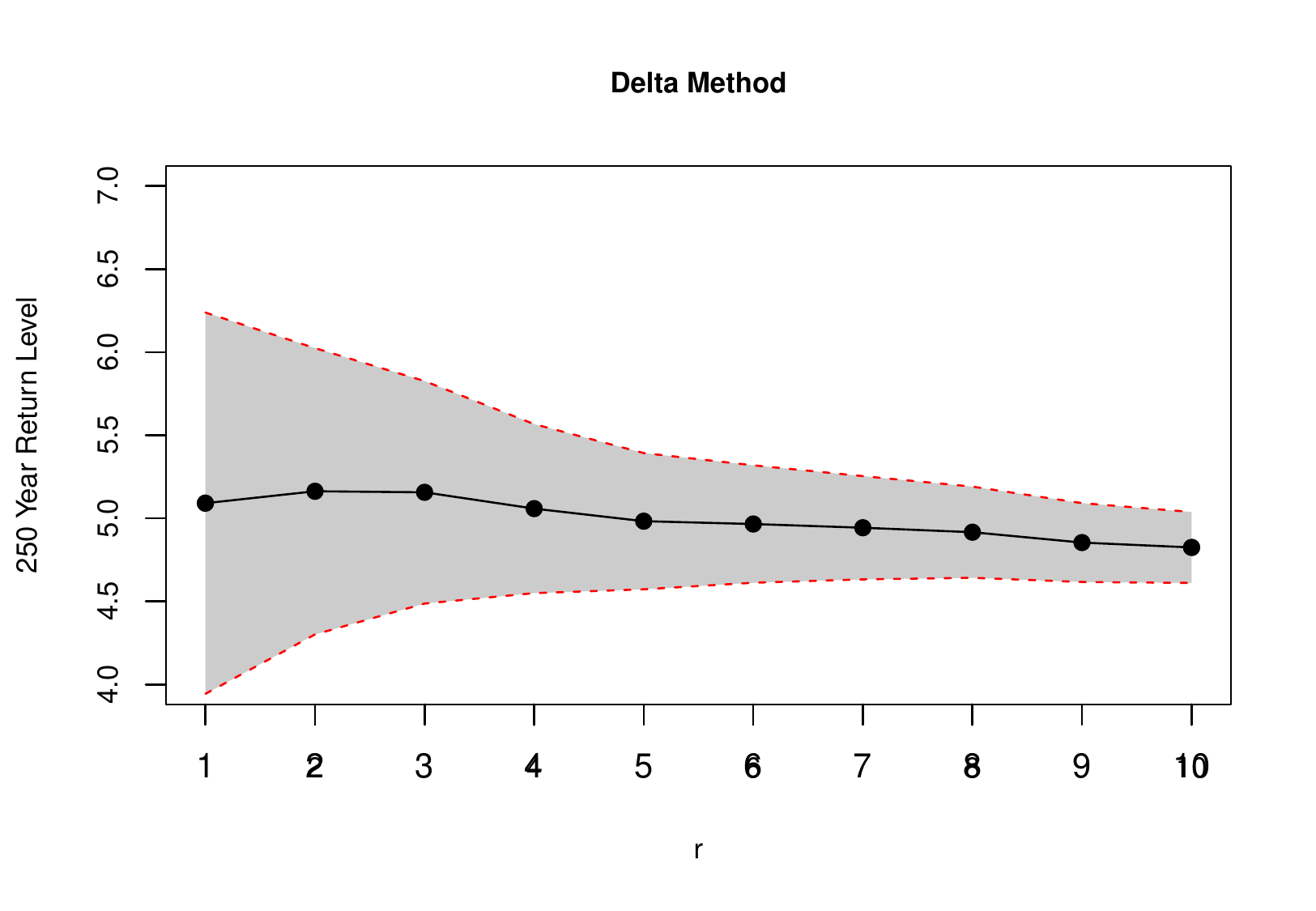}
\caption{\label{fig:lowestoft_delta} Estimates and 95\% delta method confidence 
intervals for the 250 year return level of the LoweStoft sea level dataset, for
$r=1$ through $r=10$.}
\end{figure}

\subsection{Fitting the GEV$_r$ distribution}
\label{fitting-the-gevux5fr-distribution}

The \texttt{eva} package allows generalized linear modeling in each
parameter (location, scale, and shape), as well as specifying specific
link functions. As opposed to some other packages, one can use formulas
when specifying the models, so it is quite user friendly. Additionally,
to benefit optimization, there is efficient handling of the near-zero
shape parameter in the likelihood and covariates are automatically
centered and scaled when appropriate (they are transformed back to the
original scale in the output). The \texttt{gevrFit} function includes 
certain class definitions such as \texttt{logLik}, \texttt{AIC}, 
and \texttt{BIC} to ease the steps in data analysis and model selection. 
As with profile likelihood, the equivalent functionality is implemented 
for fitting the non-stationary GPD.

In the following chunk of code, we look at non-stationary fitting in the
GEV$_r$ distribution. First, data of sample size 100 are generated
from the GEV$_{10}$ distribution with stationary shape parameter
$\xi = 0$. The location and scale have an intercept of 100 and 1, with
positive trends of 0.02 and 0.01, respectively. A plot of the largest
order statistic (block maxima) can be seen in Figure
\ref{fig:gevr_nonstat}.

\begin{Shaded}
\footnotesize
\begin{Highlighting}[]
\KeywordTok{set.seed}\NormalTok{(}\DecValTok{7}\NormalTok{)}
\NormalTok{n <-}\StringTok{ }\DecValTok{100}
\NormalTok{r <-}\StringTok{ }\DecValTok{10}
\NormalTok{x <-}\StringTok{ }\KeywordTok{rgevr}\NormalTok{(n, r, }\DataTypeTok{loc =} \DecValTok{100} \NormalTok{+}\StringTok{ }\DecValTok{1}\NormalTok{:n /}\StringTok{ }\DecValTok{50}\NormalTok{,  }\DataTypeTok{scale =} \DecValTok{1} \NormalTok{+}\StringTok{ }\DecValTok{1}\NormalTok{:n /}\StringTok{ }\DecValTok{100}\NormalTok{, }\DataTypeTok{shape =} \DecValTok{0}\NormalTok{)}
\end{Highlighting}
\end{Shaded}

\begin{figure}[tbp]
\centering
\includegraphics{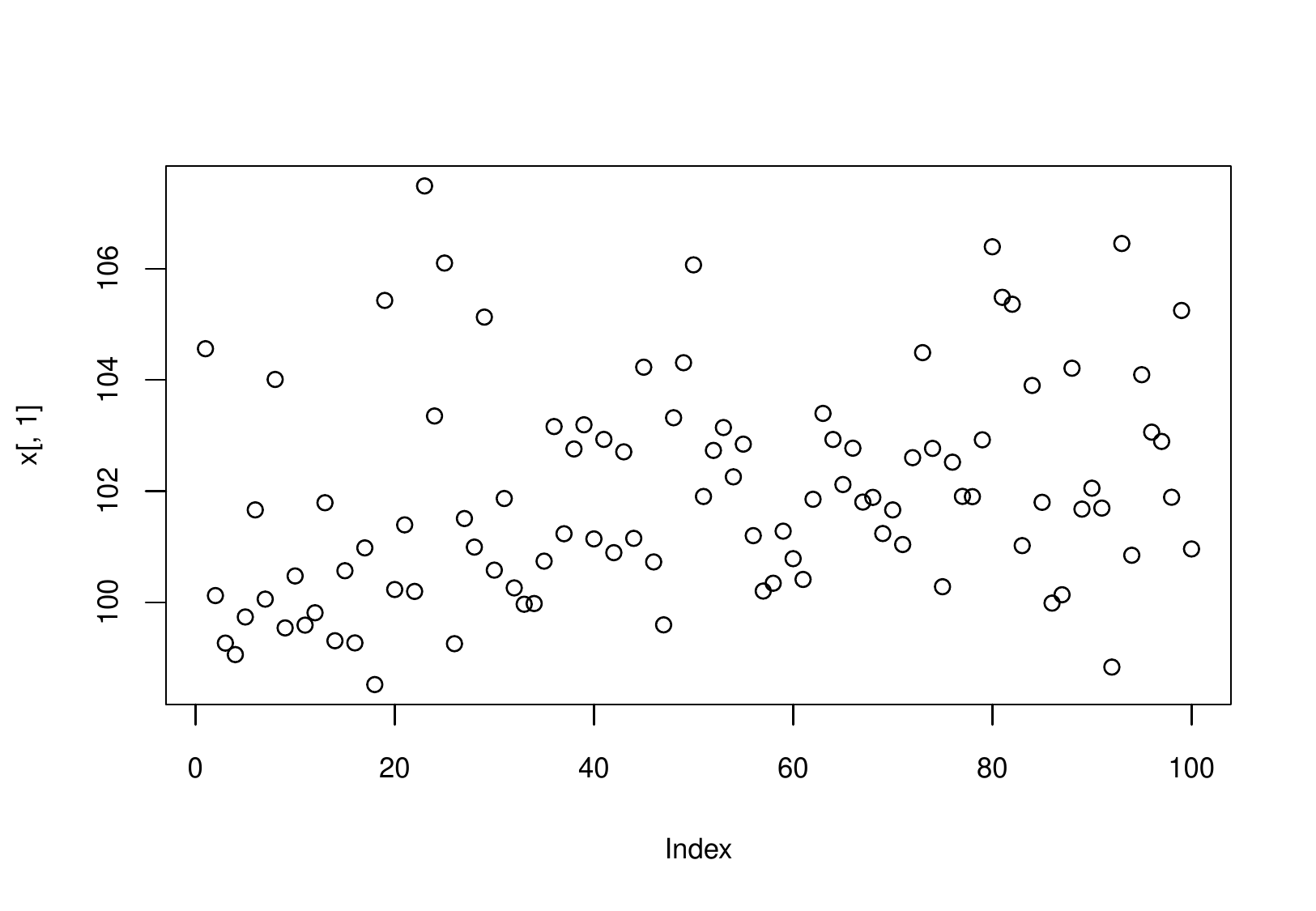}
\caption{\label{fig:gevr_nonstat} Plot of the largest order statistic
(block maxima) from a GEV$_{10}$ distribution with shape parameter
parameter $\xi = 0$. The location and scale have an intercept of 100
and 1, with positive trends of 0.02 and 0.01, respectively. The indices 
(1 to 100) are used as the corresponding trend coefficients.}
\end{figure}

\begin{Shaded}
\begin{Highlighting}[]
\NormalTok{## Creating covariates (linear trend first)}
\NormalTok{covs <-}\StringTok{ }\KeywordTok{as.data.frame}\NormalTok{(}\KeywordTok{seq}\NormalTok{(}\DecValTok{1}\NormalTok{, n, }\DecValTok{1}\NormalTok{))}
\KeywordTok{names}\NormalTok{(covs) <-}\StringTok{ }\KeywordTok{c}\NormalTok{(}\StringTok{"Trend1"}\NormalTok{)}
\NormalTok{## Create some unrelated covariates}
\NormalTok{covs$Trend2 <-}\StringTok{ }\KeywordTok{rnorm}\NormalTok{(n)}
\NormalTok{covs$Trend3 <-}\StringTok{ }\DecValTok{30} \NormalTok{*}\StringTok{ }\KeywordTok{runif}\NormalTok{(n)}

\NormalTok{## Use full data}
\NormalTok{fit_full <-}\StringTok{ }\KeywordTok{gevrFit}\NormalTok{(}\DataTypeTok{data =} \NormalTok{x, }\DataTypeTok{method =} \StringTok{"mle"}\NormalTok{, }
                    \DataTypeTok{locvars =} \NormalTok{covs, }\DataTypeTok{locform =} \NormalTok{~}\StringTok{ }\NormalTok{Trend1 +}\StringTok{ }\NormalTok{Trend2*Trend3, }
                    \DataTypeTok{scalevars =} \NormalTok{covs, }\DataTypeTok{scaleform =} \NormalTok{~}\StringTok{ }\NormalTok{Trend1)}

\NormalTok{## Only use r = 1}
\NormalTok{fit_top_only <-}\StringTok{ }\KeywordTok{gevrFit}\NormalTok{(}\DataTypeTok{data =} \NormalTok{x[, }\DecValTok{1}\NormalTok{], }\DataTypeTok{method =} \StringTok{"mle"}\NormalTok{, }
                        \DataTypeTok{locvars =} \NormalTok{covs, }\DataTypeTok{locform =} \NormalTok{~}\StringTok{ }\NormalTok{Trend1 +}\StringTok{ }\NormalTok{Trend2*Trend3, }
                        \DataTypeTok{scalevars =} \NormalTok{covs, }\DataTypeTok{scaleform =} \NormalTok{~}\StringTok{ }\NormalTok{Trend1)}
\end{Highlighting}
\end{Shaded}

From a visual assessment of the largest order statistic in Figure
\ref{fig:gevr_nonstat}, it is difficult to see trends in either the
location or scale parameters. The previous chunk of code creates some
covariates - the correct linear trend, labeled `Trend1' and two
erroneous covariates, labeled `Trend2' and `Trend3'. Then, the
non-stationary GEV$_r$ distribution is fit with the full data
($r=10$) and only the block maxima ($r=1$), modeling the correct
linear trend in both location and scale parameters, a stationary shape
parameter, and interaction / main effect terms for the erroneous
`Trend2' and `Trend3' parameters in the location.

\begin{Shaded}
\footnotesize
\begin{Highlighting}[]
\NormalTok{## Show summary of estimates}
\NormalTok{fit_full}
\end{Highlighting}
\end{Shaded}

\begin{verbatim}
## Summary of fit:
##                           Estimate Std. Error   z value   Pr(>|z|)    
## Location (Intercept)   100.0856115  0.2276723 439.60378 0.0000e+00 ***
## Location Trend1          0.0193838  0.0042366   4.57535 4.7543e-06 ***
## Location Trend2         -0.0716415  0.0949454  -0.75455 4.5052e-01    
## Location Trend3          0.0031478  0.0049866   0.63124 5.2788e-01    
## Location Trend2:Trend3   0.0032205  0.0053048   0.60710 5.4378e-01    
## Scale (Intercept)        1.0831201  0.0922684  11.73879 8.0627e-32 ***
## Scale Trend1             0.0097441  0.0016675   5.84352 5.1108e-09 ***
## Shape (Intercept)        0.0393603  0.0293665   1.34031 1.8014e-01    
## ---
## Signif. codes:  0 '***' 0.001 '*' 0.01 '*' 0.05 '.' 0.1 ' ' 1
\end{verbatim}

\begin{Shaded}
\footnotesize
\begin{Highlighting}[]
\NormalTok{fit_top_only}
\end{Highlighting}
\end{Shaded}

\begin{verbatim}
## Summary of fit:
##                          Estimate Std. Error   z value   Pr(>|z|)    
## Location (Intercept)   99.5896128  0.4373915 227.68985 0.0000e+00 ***
## Location Trend1         0.0235401  0.0052891   4.45072 8.5582e-06 ***
## Location Trend2        -0.3567545  0.2858703  -1.24796 2.1205e-01    
## Location Trend3         0.0170264  0.0170313   0.99971 3.1745e-01    
## Location Trend2:Trend3  0.0021764  0.0194109   0.11212 9.1073e-01    
## Scale (Intercept)       1.1098238  0.2654124   4.18151 2.8958e-05 ***
## Scale Trend1            0.0036369  0.0042099   0.86390 3.8764e-01    
## Shape (Intercept)       0.1375594  0.1063641   1.29329 1.9591e-01    
## ---
## Signif. codes:  0 '***' 0.001 '*' 0.01 '*' 0.05 '.' 0.1 ' ' 1
\end{verbatim}

From the output, one can see from the fit summary that using $r=10$,
both the correct trend in location and scale are deemed significant.
However, using $r=1$ a test for significance in the scale trend fails.

Next, we fit another two models (using $r = 10$). The first, labeled
`fit\_reduced1' is a GEV$_{10}$ fit incorporating only the true linear
trends as covariates in the location and scale parameters. The second,
`fit\_reduced2' fits the Gumbel equivalent of this model ($\xi = 0$).
Note that `fit\_reduced1' is nested within `fit\_full' and
`fit\_reduced2' is further nested within `fit\_reduced1'.

\begin{Shaded}
\begin{Highlighting}[]
\NormalTok{fit_reduced1 <-}\StringTok{ }\KeywordTok{gevrFit}\NormalTok{(}\DataTypeTok{data =} \NormalTok{x, }\DataTypeTok{method =} \StringTok{"mle"}\NormalTok{, }
                        \DataTypeTok{locvars =} \NormalTok{covs, }\DataTypeTok{locform =} \NormalTok{~}\StringTok{ }\NormalTok{Trend1,}
                        \DataTypeTok{scalevars =} \NormalTok{covs, }\DataTypeTok{scaleform =} \NormalTok{~}\StringTok{ }\NormalTok{Trend1)}

\NormalTok{fit_reduced2 <-}\StringTok{ }\KeywordTok{gevrFit}\NormalTok{(}\DataTypeTok{data =} \NormalTok{x, }\DataTypeTok{method =} \StringTok{"mle"}\NormalTok{, }
                        \DataTypeTok{locvars =} \NormalTok{covs, }\DataTypeTok{locform =} \NormalTok{~}\StringTok{ }\NormalTok{Trend1,}
                        \DataTypeTok{scalevars =} \NormalTok{covs, }\DataTypeTok{scaleform =} \NormalTok{~}\StringTok{ }\NormalTok{Trend1, }\DataTypeTok{gumbel =} \OtherTok{TRUE}\NormalTok{)}
\end{Highlighting}
\end{Shaded}

One way to compare these models is using the Akaike information
criterion (AIC) and choose the model with the smallest value. This
metric can be extracted by using \texttt{AIC($\cdot$)} on a 
\texttt{gevrFit} object. By this metric, the chosen model is 
`fit\_reduced2', which agrees with the true model (Gumbel, with 
linear trends in the location and scale parameters).

\begin{Shaded}
\footnotesize
\begin{Highlighting}[]
\NormalTok{## Compare AIC of three models}
\KeywordTok{AIC}\NormalTok{(fit_full)}
\end{Highlighting}
\end{Shaded}

\begin{verbatim}
## [1] 127.0764
\end{verbatim}

\begin{Shaded}
\footnotesize
\begin{Highlighting}[]
\KeywordTok{AIC}\NormalTok{(fit_reduced1)}
\end{Highlighting}
\end{Shaded}

\begin{verbatim}
## [1] 120.4856
\end{verbatim}

\begin{Shaded}
\footnotesize
\begin{Highlighting}[]
\KeywordTok{AIC}\NormalTok{(fit_reduced2)}
\end{Highlighting}
\end{Shaded}

\begin{verbatim}
## [1] 118.4895
\end{verbatim}

\section{Summary}
\label{summary}

Many of the same options and procedures that are available for the
GEV$_r$ distribution are also available for the Generalized Pareto
distribution. In particular, functionality is for the most part the same
for distribution functions (\texttt{r}, \texttt{d}, 
\texttt{q}, \texttt{p} prefixes), profile likelihood, 
and model fitting. Model fitting diagnostic plots are the same and
include probability, quantile, return level, density, and residual
vs.~covariates (each provided when appropriate, depending on
non-stationarity).

The tests developed in Chapter \ref{ch:r-largest} are implemented via
the wrapper function \texttt{gevrSeqTests} and those discussed in Chapter
\ref{ch:gpd} can be used through function \texttt{gpdSeqTests}, along 
with a few additional tests. All the tests that require a computational
bootstrap approach have the option to be run in parallel, via R package
\texttt{parallel}. The function \texttt{pSeqStop} provides an implementation 
of the recently developed ForwardStop and StrongStop p-value adjustments 
derived in \cite{g2015sequential}.

In closing, \texttt{eva} attempts to integrate and simplify all aspects
of model selection, validation, and analysis of extremes for end users.
It does so by introducing new methodology for goodness-of-fit testing in
the GEV$_r$ and GP distributions, making model fitting more
straightforward (formula statements, diagnostic plots, and certain class
inheritances), and efficient handling of numerical issues known to
extremes.
\chapter{Conclusion}
\label{ch:conclusion}

The research presented here provides new methodology to approach 
goodness-of-fit testing in extreme value models and alternative 
estimation procedures in non-stationary regional frequency 
analysis of extremal data. Chapter~\ref{ch1} presents an overview 
of extreme value theory and some of the outstanding issues that 
will be addressed in the following chapters.

In Chapter~\ref{ch:r-largest}, two goodness-of-fit tests are 
developed for use in the selection of $r$ in the $r$ largest 
order statistics model, or GEV$_r$ distribution. The first is 
a score test. Due to the non-regularity of the GEV$_r$ 
distribution, a bootstrap approach must be used to obtain 
critical values for the test statistic. Parametric bootstrap 
is simple, yet computationally expensive so an alternative 
multiplier approximation~\citep[e.g.,][]{kojadinovic2012goodness} 
can be used for efficiency. The second, called the entropy 
difference (ED) test, utilizes the difference in log-likelihood 
between the GEV$_r$ and GEV$_{r-1}$ and the distribution of 
the test statistic under the null hypothesis is shown to 
be asymptotically normal. The properties of the tests 
are studied and both are found to hold their size and 
have adequate power to detect deviations from the GEV$_r$ 
distribution. Recently developed stopping rules~\citep{g2015sequential} 
to control the familywise error rate (FWER) and false 
discovery rate (FDR) in the ordered, sequential testing 
setting are studied and applied in our framework. The 
utility of the tests are shown via applications to extreme 
sea levels and precipitation.

A similar, but distinct problem is the choice of threshold 
in the peaks-over-threshold (POT) approach. In Chapter~\ref{ch:gpd}, 
a goodness-of-fit approach is considered for testing a set of 
predetermined thresholds. Multiple tests are studied for a single, 
fixed threshold under various misspecification settings and it is 
found that the Anderson--Darling test is most powerful in the 
majority of the settings. The same stopping rules used in 
Chapter~\ref{ch:r-largest} are applied in this setting to control 
the error rate in an ordered, multiple threshold testing situation. 
The Anderson--Darling test for the Generalized Pareto distribution 
is studied in detail by~\cite{choulakian2001goodness} and they 
derive the asymptotic distribution of the test statistic, which 
requires obtaining the eigenvalues of an integral equation. An 
insightful and generous discussion with one of the authors, 
Prof. Choulakian, led to development of an approximate, but 
accurate method to compute p-values in a computationally 
efficient manner. This provides an automated multiple 
threshold testing procedure with the Anderson--Darling test 
that can be scaled in a straightforward manner. Its use is 
demonstrated by generating a map of precipitation return levels 
at hundreds of sites in the western United States.

In Chapter~\ref{ch:rfa}, two alternative estimation procedures are 
thoroughly studied and implemented for modeling extremal data 
with non-stationary regional frequency analysis (RFA). In stationary 
RFA, L-moment estimation can be used as an alternative to maximum 
likelihood (ML). In extremal distributions, L-moment estimation has 
some desirable properties compared to MLE such as efficiency in small 
samples and that it only requires existence of the 
mean~\citep{hosking1985estimation}. However, there is no straightforward 
extensions to non-stationary RFA and currently MLE is the only 
available well-studied option. We propose two alternative estimation 
procedures for the non-stationary setting. The first, is an extension 
of maximum product spacings (MPS) estimation~\citep{cheng1983estimating} 
from the one sample case. The second is a hybrid, iterative L-moment / 
maximum likelihood method. Both have theoretical and/or computational 
advantages over MLE and they are shown to have better performance (in 
terms of squared error) via a large scale simulation study of multivariate 
extremal data fit with a non-stationary RFA model. This improvement is 
more apparent in simulations with short record length. The three estimation 
methods are compared in an analysis of daily precipitation extremes in 
California, fit with a non-stationary flood index RFA model. Interest 
is in the effect of the El Ni\~{n}o--Southern Oscillation on winter 
precipitation extremes in this region. Semi-parametric bootstrap 
procedures~\citep[e.g.,][]{heffernan2004conditional} are used to obtain 
standard errors in the presence of unspecified spatial dependence.

Most of the methodology developed here has been implemented in the R 
package \texttt{eva}, which is publicly available on CRAN. In addition, 
it addresses deficiencies in other software for analyzing extremes. 
These include efficient handling of numerical instabilities when dealing 
with small shape parameter values, data generation and density functions 
for the GEV$_r$ distribution, profile likelihood for stationary return 
levels and the shape parameter, and more user-friendly functionality 
for model fitting and diagnostics.

\section{Future Work}
\label{ch6:future}

The goodness-of-fit tests developed in Chapter~\ref{ch:r-largest} 
for selection of $r$ in the GEV$_r$ distribution can be extended 
to allow covariates in the parameters. For example, extremal 
precipitation in a year may be affected by large scale climate 
indexes such as the Southern Oscillation Index 
(SOI), which may be incorporated as a covariate in the location 
parameter \citep[e.g.,][]{Shan:Yan:Zhan:enso:2011}. 
Both tests can be carried out with additional non-stationary 
covariates in any of the model parameters. However, this 
introduces additional complexity in performing model selection 
(i.e., how to choose covariates and select $r$?). When 
the underlying data falls into a rich class of dependence 
structures (such as time series), this dependence may be incorporated 
directly instead of using a procedure to achieve approximate 
independence (e.g. the storm length $\tau$ in 
Section~\ref{ch2:app}). For example, take the GEV-GARCH 
model~\citep{zhao2011garch} when $r=1$. It may be extended
to the case where $r>1$ and the tests presented here may 
be applied to select $r$ under this model assumption.

Threshold selection in the non-stationary GPD, as discussed in 
\citet{roth2012regional} and \citet{northrop2011threshold}, 
is an obvious extension to the automated threshold testing 
methodology developed in Chapter~\ref{ch:gpd}. However, one 
particular complication that arises is performing model selection 
(i.e., what covariates to include) while concurrently testing 
for goodness-of-fit to various thresholds. It is clear that 
threshold selection will be dependent on the model choice. 
Another possible extension to this work involves testing 
for overall goodness-of-fit across sites (one test statistic). 
In this way, a fixed or quantile regression based threshold 
may be predetermined and then tested simultaneously across 
sites. In this setup both spatial and temporal dependence need to 
be taken into account. Handling this requires some care due to 
censoring~\citep[e.g.,][Section 2.5.2]{yan2016extremes}. In other 
words, it is not straightforward to capture the temporal dependence 
as exceedances across sites are not guaranteed to occur at the 
same points in time.

Future work for estimation in non-stationary regional frequency 
analysis (RFA) as in Chapter~\ref{ch:rfa}, is to study the 
convergence rate of the estimators presented in terms of the 
number of sites, $m$ and sample size $n$ at each site. The 
relationship between $m$ and $n$ for estimator efficiency 
may be quite complicated and further investigation is needed. 
Additionally, the methodology can be applied 
to other heavy-tailed / non-regular distributions 
such as the Generalized Pareto, Kappa, and L\`{e}vy. Another 
important, yet not trivial task is to develop a heterogeneity 
measure in the presence of nonstationarity to determine regional 
homogeneity. \cite{hosking2005regional} discuss some metrics for the 
stationary case; however this becomes less straightforward when 
the data are not assumed to be identically distributed. Instead 
of assuming linear relationships in the covariates, splines can 
be used to create a smooth curve. This can be quite useful 
to describe the effect of extremes over a geographic region. 
Additional care would be needed here to choose the level of 
smoothness and type of spline used.

It would be desirable to have a mechanism to allow for missing data. 
Currently, only sites with full and balanced records are used in 
the analysis. This is a serious drawback, particularly with historical 
climate records that often exhibit missingness. Assuming data are 
missing completely at random (MCAR), the initial estimation procedure 
would not introduce bias and can be handled in a straightforward manner. 
To generate confidence intervals for the parameters, the semi-parametric 
bootstrap procedure in Appendix~\ref{app:boot} would need to be adjusted. 
In the presence of missing data, the cluster size of observations across 
sites within each year would not necessarily be equal.

On a broader note, variable selection is an important part of any data 
analysis. As discussed earlier, in extremes often there exists a limited 
number of observations. This can lead to the number of parameters to be 
estimated exceeding the number of observations (small $n$, large $p$). 
Even in cases where $n > p$, there is still the very real 
possibility of over-fitting; such rules of thumb such as the 
`one in ten' rule states that there must be at least ten 
observations for each predictor~\citep{harrell1984regression}. For 
applications of extremes to nature, there are many apparent predictors 
that one may be interested in for example - linear, quadratic, seasonal, 
or cyclic time trends, atmospheric/weather readings, coordinate locations, 
etc.

Surprisingly, very little work has been done in this area. The only 
literature that considers regularization is~\cite{fastextremes2010}. 
They explore the use of RaVE, a sparse variable selection method, and 
apply it to annual rainfall data with a large number of atmospheric 
covariates. RaVE is formulated as a Bayesian hierarchical model and 
can be applied to any model with a well-defined likelihood. 
\cite{menendez2009influence} use a modified stepwise selection method 
to select parameters. \cite{minguez2010pseudo} use modified forward 
selection, which penalizes based on AIC using a likelihood score 
perturbation method.

The goal would be to develop a sparse regularization technique for extremes 
via penalized estimation similar to lasso~\citep{tibshirani1996regression}. 
Implementation of this is not trivial. The GEV distribution is not in 
the exponential family and thus, some of the nice results in GLM may not 
be available. Unlike the (usual) case of normal regression with homoscedasticity, 
when modeling non-stationarity in the GEV distribution often times both the 
location and scale parameters are assumed to vary with covariates, so 
shrinkage of both sets of covariates is desired. This would need to be 
explored; for example, different penalties could be used for each 
covariate set. Lastly, this regularization can be combined with 
existing maximum likelihood penalties that ensure the shape parameter 
$\xi$ can be estimated~\citep{coles1999likelihood}.
\appendix       
\chapter{Appendix}
\label{ch:appendix}

\section{Data Generation from the GEV$_r$ Distribution}
\label{app:gevrsim}

The GEV$_r$ distribution is closely connected to the GEV distribution.
Let $(Y_1, \ldots, Y_r)$, $Y_1 > \cdots > Y_r$ follow 
the GEV$_r$ distribution~\eqref{eq:gevr_pdf}. 
It can be shown that the GEV$_1$ distribution is the GEV distribution 
with the same parameters, which is the marginal distribution of $Y_1$. 
More interestingly, note that, the conditional distribution of $Y_2$ 
given $Y_1 = y_1$ is simply the GEV distribution right truncated by $y_1$.
In general, given $(Y_1, \ldots, Y_k) = (y_1, \ldots, y_k)$ for $1 \le k < r$,
the conditional distribution of $Y_{k+1}$ is the GEV distribution
righted truncated at $y_k$.
This property can be exploited to generate the $r$ components 
in a realized GEV$_r$ observation.

The pseudo algorithm to generate a single observation is the following:
\begin{itemize}
  \item
  Generate the first value $y_1$ from the (unconditional) GEV distribution.

  \item
  For $i=2, \ldots, r$:
  
  \begin{itemize}
    \item
    Generate $y_i$ from the GEV distribution right truncated by $y_{i-1}$.
  \end{itemize}
  
\end{itemize}
The resulting vector ($y_1, \ldots, y_r$) is a 
single observation from the GEV$_r$ distribution.

For $\xi \to 0$, caveat is needed in numerical evaluation.
Using function \texttt{expm1} for $\exp(1 + x)$ for $x\to 0$
provides much improved accuracy in comparison to a few 
implementations in existing R packages. 
For readability, here is a simplified version of the 
implementation in R package \texttt{eva} \citep{Rpkg:eva}.
\begin{alltt}
## Quantile function of a GEVr(loc, scale, shape)
qgev <- function(p, loc = 0, scale = 1, shape = 0, log.p = FALSE) \{
  if (log.p) p <- exp(p)
  if(shape == 0) \{
    loc - scale * log(-log(p))
  \} else {
    loc + scale * expm1(log(-log(p)) * -shape) / shape)
  \}
}

## Random number generator of GEVr
## Returns a matrix of n rows and r columns, each row a draw from GEVr
rgevr <- function(n, r, loc = 0, scale = 1, shape = 0) \{
  umat <- matrix(runif(n * r), n, r)
  if (r > 1) \{
    matrix(qgev(t(apply(umat, 1, cumprod)), loc, scale, shape), ncol = r)
  \} else \{
    qgev(umat, loc, scale, shape)
  \}
\}
\end{alltt}

\section{Asymptotic Distribution of $T_{n}^{(r)}(\theta)$}
\label{app:tn}

\begin{proof}[Proof of Theorem~\ref{thm:ed}]
Consider a random vector $(Y_1, ... , Y_r)$ which follows
a GEV$_r$($\theta$) distribution.
The following result given by \citet[pg. 248]{tawn1988extreme} 
will be used:
\begin{align}
\label{eq:tawnmoment}
h(j | \theta, a, b, c) &\equiv E[Z^a_j (1 + \xi Z_j)^{-(\frac{1}{\xi}+b)} \log^c(1 + \xi Z_j)] \notag\\
&=  \frac{(-\xi)^{c-a}}{\Gamma(j)} \sum_{\alpha=0}^{a} (-1)^\alpha {a \choose \alpha} \Gamma^{(c)} (j + b\xi - \alpha \xi + 1)
\end{align}
where $Z_j = (Y_j - \mu) / \sigma$ and $\Gamma^{(c)}$ is the $c$th 
derivative of the gamma function, for 
$a \in \mathbb{Z}$, $b \in \mathbb{R}$, and $c \in \mathbb{Z}$,
such that $(j + b\xi - \alpha \xi + 1) \not\in \{0, -1, -2, \ldots \}$,
$\alpha = 0, 1, \ldots, a$.

Assume that $\xi \neq 0$ and $1 + \xi Z_j > 0$ for $j=1, \ldots, r$. 
The difference in log-likelihoods for a single observation from the
GEV$_r$($\theta$) and GEV$_{r-1}$($\theta$) distribution, $D_{ir}$, 
is given by~\eqref{eq:dll} in Section~\ref{ch2:ed}. Thus, the first 
moment of $D_{ir}$ is
\begin{align*}
 E[D_{1r}] = &\ - \log{\sigma} - h(r | \theta, 0, 0, 0) + h(r - 1 | \theta, 0, 0, 0) \\ 
 &\ - \left(\frac{1}{\xi} + 1\right) h(r | \theta, 0, -\xi^{-1}, 1)  \\
 = &\ -\log{\sigma} - 1 + (1+\xi)\psi(r)
\end{align*}
where $\psi(x) = \frac{\Gamma^{(1)}(x)}{\Gamma(x)}$.

To prove that the second moment of $D_{ir}$ is finite, note that
\begin{align*}
|D_{1r}| \leq &\ 4 \max \Bigg\{ \Big|\log{\sigma}\Big|, \Big|(1+\xi Z_{1r})^{-\frac{1}{\xi}}\Big|, \\
&\ \Big|(1+\xi Z_{1{r-1}})^{-\frac{1}{\xi}}\Big|, \Big|\Big(\frac{1}{\xi}+1\Big)\log(1+\xi Z_{1{r-1}})\Big|\Bigg\},
\end{align*}
which implies
\begin{align*}
D^2_{1r} \leq &\ 16 \Bigg(\max \Bigg\{ \Big|\log{\sigma}\Big|, \Big|(1+\xi Z_{1r})^{-\frac{1}{\xi}}\Big|, \\
&\ \Big|(1+\xi Z_{1{r-1}})^{-\frac{1}{\xi}}\Big|, \Big|\Big(\frac{1}{\xi}+1\Big)\log(1+\xi Z_{1{r-1}})\Big|\Bigg\}\Bigg)^2.
\end{align*}
The bound of $E(D_{1r}^2)$ can be established by applying~\eqref{eq:tawnmoment} 
to the last three terms in the $\max$ operator,
\begin{align*}
E[(1+\xi Z_{1r})^{-\frac{2}{\xi}}] &= h(r| \theta, 0, \xi^{-1}, 0) < \infty, \\
E[(1+\xi Z_{1{r-1}})^{-\frac{2}{\xi}}] &= h(r-1| \theta, 0, \xi^{-1}, 0) < \infty, \\
E[\log^2(1+\xi Z_{1{r-1}})] &= h(r-1 | \theta, 0, -\xi^{-1}, 2) < \infty.
\end{align*}
The desired result then follows from the central limit theorem and Slutsky's theorem.

The case where $\xi = 0$ in Theorem~\ref{thm:ed} can 
easily be derived by taking the limit as $\xi \to 0$ 
in~\eqref{eq:dll} and in~\eqref{eq:tawnmoment}
by the Dominated Convergence Theorem.
\end{proof}

\section{Semi-Parametric Bootstrap Resampling in RFA}
\label{app:boot}

Described below and borrowed from~\cite{heffernan2004conditional} 
is the iterative procedure used to obtain bootstrapped standard errors 
for parameters of a regional frequency model in the presence of 
unspecified spatial dependence.

\begin{enumerate}
\item
Fit the non-stationary RFA model to the data and obtain parameter estimates.

\item
Transform the original GEV distributed data into standard uniform residuals 
using the estimated parameters from Step 1.

\item
Calculate the ranks of the $n$ residuals within each site.

\item
Repeat $B$ times, for some large number $B$:

\begin{enumerate}

\item
Resample (with replacement) across years from the ranks in Step 3.
	
\item
Generate a new $n \times m$ sample of standard uniform residuals 
and arrange them (within each site) according to the resampled 
rankings.\footnote{For ties, can use a random assignment process.}

\item
Transform the new sampled residuals back into GEV margins 
using the estimated parameters in Step 1.

\item
Fit the model again and obtain the estimators.

\end{enumerate}

\end{enumerate}

%

\singlespacing
\bibliographystyle{chicago}
\bibliography{thesis}

\end{document}